\documentclass[aps,superscriptaddress,nofootinbib]{revtex4}

\setcitestyle{numbers,open={(},close={)}}

\usepackage[letterpaper, left=0.8in, right=0.8in, top=0.7in, bottom=0.7in]{geometry}
\usepackage{titlesec}

\usepackage{cmap} 
\usepackage[utf8]{inputenc}
\usepackage[T1]{fontenc}
\usepackage{amsmath}
\usepackage[english]{babel}
\usepackage{amsfonts}
\usepackage{xcolor}
\definecolor{blueviolet}{rgb}{0.2, 0.2, 0.6}
\definecolor{webgreen}{rgb}{0,.5,0}
\definecolor{webbrown}{rgb}{.6,0,0}
\usepackage[pdftex,
	bookmarks=false,
	colorlinks=true, 
	urlcolor=webbrown,
	linkcolor=blueviolet, 
	citecolor=webgreen,
	pdfstartpage=1,
	pdfstartview={FitH},  
	bookmarksopen=false
	]{hyperref}
\usepackage{tikz}
\usepackage{natbib}
\usepackage{mathtools}
\usepackage{graphicx}
\usepackage{braket}
\usepackage{dsfont}
\usepackage{listings}
\usepackage{qcircuit}

\allowdisplaybreaks

\usepackage{pdfpages}
\usepackage{multirow}
\usepackage{amssymb}

\usepackage{amsthm}

\usepackage{color}

\definecolor{deepblue}{rgb}{0,0,0.5}
\definecolor{deepred}{rgb}{0.6,0,0}
\definecolor{deepgreen}{rgb}{0,0.5,0}

\usepackage{nicefrac}

\DeclareFixedFont{\ttb}{T1}{txtt}{bx}{n}{9} 
\DeclareFixedFont{\ttm}{T1}{txtt}{m}{n}{9}  

\newcommand\pythonstyle{\lstset{
language=Python,
basicstyle=\ttm,
morekeywords={self},              
keywordstyle=\ttb\color{deepblue},
emph={MyClass,__init__},          
emphstyle=\ttb\color{deepred},    
stringstyle=\color{deepgreen},
frame=tb,                         
showstringspaces=false
}}

\lstnewenvironment{python}[1][]
{
\pythonstyle
\lstset{#1}
}
{}

\newcommand\pythoninline[1]{{\pythonstyle\lstinline!#1!}}

\definecolor{orange}{RGB}{255,127,0}

\def\bra#1{\ensuremath{\mathinner{\langle{#1}|}}}
\def\ket#1{\ensuremath{\mathinner{|{#1}\rangle}}}

\newcommand{\ketbra}[2]{\lvert #1 \rangle \! \langle #2 \rvert}
\newcommand{\norm}[1]{\left\lVert#1\right\rVert}

\newcommand{\tr}{\text{Tr}}

\DeclareMathOperator{\SWAP}{SWAP}

\newtheorem{proposition}{Proposition}
\newtheorem{conjecture}{Conjecture}
\DeclareMathOperator{\Tr}{tr}

\newtheoremstyle{custom}
{3pt}
{3pt}
{}
{}
{\bfseries}
{.}
{.5em}
{}
\theoremstyle{custom}

\newtheorem{theorem}{Theorem}
\newtheorem{corollary}{Corollary}
\newtheorem{definition}{Definition}

\newtheorem{lemma}{Lemma}
\newtheorem{fact}{Fact}

\newcommand{\Id}{I}

\usepackage[noend]{algpseudocode}
\usepackage{algorithm,algorithmicx}

\algrenewcommand\alglinenumber[1]{\sf\scriptsize\color{blue}{#1}}
\algrenewcommand\algorithmicrequire{\textbf{Input:}}
\algrenewcommand\algorithmicensure{\textbf{Output:}}

\begin{document}

\title{Generative quantum advantage for classical and quantum problems}

\author{Hsin-Yuan Huang}
    \email[Corresponding author: ]{hsinyuan@caltech.edu}
	\affiliation{Google Quantum AI, 340 Main Street, Venice, CA 90291, USA}
    \affiliation{Department of Physics, California Institute of Technology, Pasadena 91125, CA, USA}
\author{Michael Broughton}
	\affiliation{Google Quantum AI, 340 Main Street, Venice, CA 90291, USA}
\author{Norhan Eassa}
	\affiliation{Google Quantum AI, 340 Main Street, Venice, CA 90291, USA}
        \affiliation{Department of Physics and Astronomy, Purdue University, West Lafayette, IN 47906, USA}
    \author{Hartmut Neven}
	\affiliation{Google Quantum AI, 340 Main Street, Venice, CA 90291, USA}
    \author{Ryan Babbush}
	\affiliation{Google Quantum AI, 340 Main Street, Venice, CA 90291, USA}
\author{Jarrod R.~McClean}
    \email[Corresponding author: ]{jmcclean@google.com}
	\affiliation{Google Quantum AI, 340 Main Street, Venice, CA 90291, USA}
\date{\today}

\begin{abstract}
Recent breakthroughs in generative machine learning, powered by massive computational resources, have demonstrated unprecedented human-like capabilities. While beyond-classical quantum experiments can generate samples from classically intractable distributions, their complexity has thwarted all efforts toward efficient learning. This challenge has hindered demonstrations of \emph{generative quantum advantage}: the ability of quantum computers to learn and generate desired outputs substantially better than classical computers. We resolve this challenge by introducing families of generative quantum models that are hard to simulate classically, are efficiently trainable, exhibit no barren plateaus or proliferating local minima, and can learn to generate distributions beyond the reach of classical computers. Using a $68$-qubit superconducting quantum processor, we demonstrate these capabilities in two scenarios: learning classically intractable probability distributions and learning quantum circuits for accelerated physical simulation. Our results establish that both learning and sampling can be performed efficiently in the beyond-classical regime, opening new possibilities for quantum-enhanced generative models with provable advantage.
\end{abstract}

\maketitle



\section{Introduction}
Fueled by the overwhelming growth of computational power, recent progress in machine learning has continued to redefine the boundaries of what was once deemed impossible. Much of this progress has focused on large language models (LLMs)~\cite{vaswani2017attention} and diffusion models~\cite{ho2020denoising}, which represent specific types of generative machine learning models~\cite{goodfellow2016deep}. The promise of quantum computers to dramatically expand our computational capabilities for certain problems has naturally led to speculation about how quantum technology will reshape the landscape of machine learning. Here, we establish a new type of advantage for generative tasks supported by both theory and experiment, with strong implications for both classical and quantum models.

Perhaps due to its potential impact, quantum machine learning~\cite{biamonte2017quantum} is a field characterized by equal amounts of optimism, skepticism, and ambiguity. It encompasses both how quantum technology may impact machine learning and how machine learning may impact quantum computing. When data arrives as quantum states, exponential advantages have been both proven mathematically and demonstrated experimentally~\cite{huang2021information,chen2022exponential,huang2022quantum}. However, when data arrives as classical information to be processed by quantum computers, the story becomes much more subtle~\cite{aaronson2015read,huang2020power,schuld2022quantum,abbas2023quantum,jaques2023qram}.

Modern machine learning is driven by empirical successes on large datasets with immense computing power. However, the current state of quantum computers more closely resembles the early development of neural networks: absent the requisite computational power to test ideas empirically, we must resort to proofs in simplified computational frameworks. Using learning theory, there are approaches with provable exponential advantages on classical data~\cite{bshouty1995learning,arunachalam2017guest,arunachalam2020quantum,atici2007quantum,cross2015quantum,bernstein1993quantum,lewis2025quantum}, but these often rely on embedding somewhat artificial problems, such as discrete logarithm, into learning problems or assume models of fast access like \textsf{QRAM} that remain controversial in their practicality~\cite{jaques2023qram}. Moreover, these same frameworks quickly conclude the absence of advantages when faced with general distributions~\cite{arunachalam2018optimal,atici2005improved,servedio2004equivalences,zhang2010improved}, echoing early pessimistic and distribution-specific hardness results~\cite{brutzkus2017globally, Safran2018-xn, Daniely2020-eu, Kiani2024-qt} in theoretical classical machine learning. This leaves a considerable gap in our understanding of quantum advantages for empirical and less structured distributions, which this work helps resolve.

\begin{figure*}[t]
    \centering
    \includegraphics[width=\textwidth]{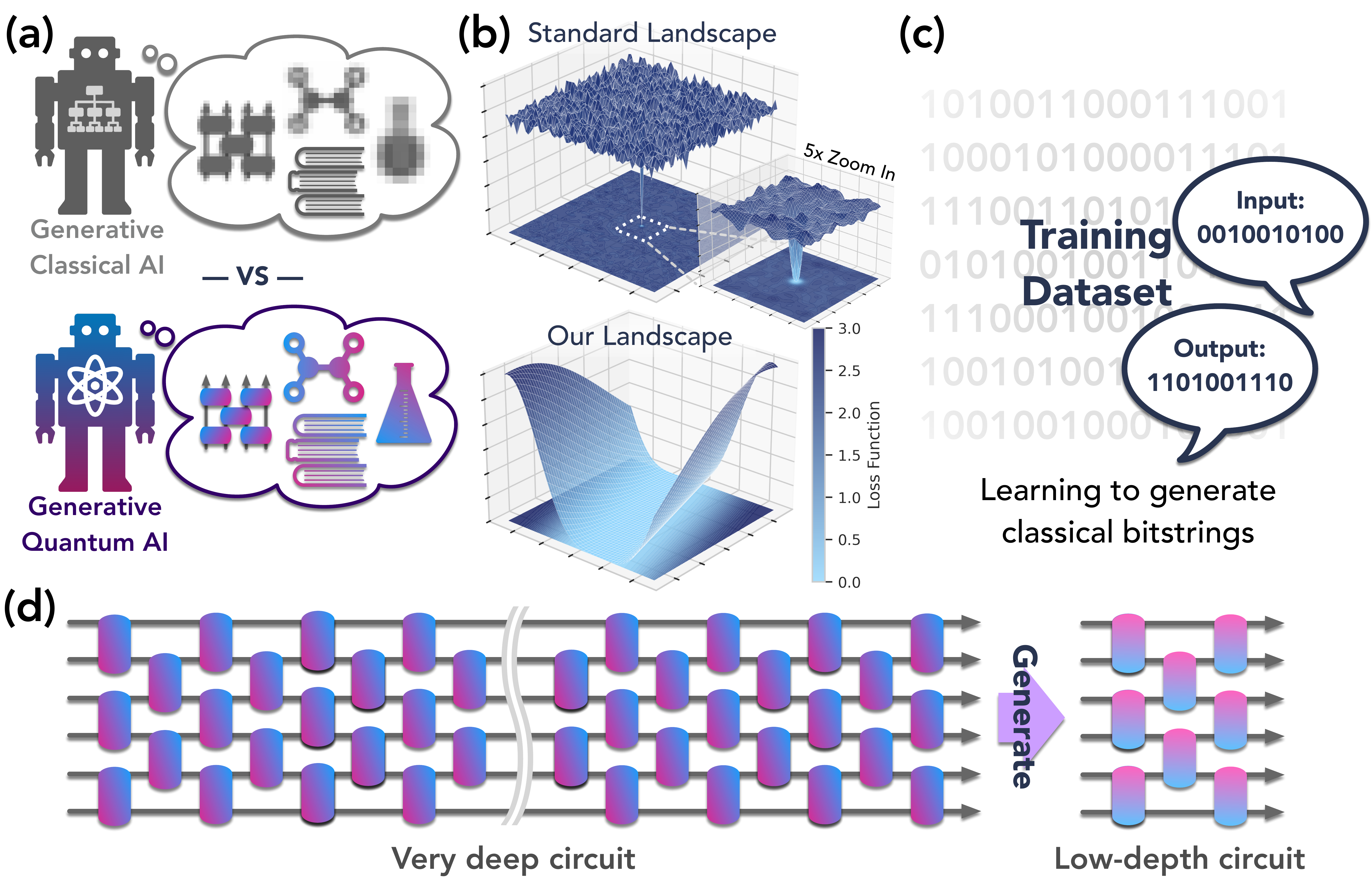}
    \caption{\textbf{Generative quantum advantage overview.} \textbf{(a)} Generative AI enabled by quantum technology allows learning and sampling from distributions that are classically hard to sample from. \textbf{(b)} Standard approaches to learning generative quantum AI models suffer from poor training landscapes due to issues such as proliferation of suboptimal local minima and barren plateaus, in contrast to the techniques proposed here that enable favorable landscapes with efficient training even in classically hard regimes. \textbf{(c)} As a key application, we theoretically and experimentally demonstrate advantage in sampling from classical distributions that are classically intractable, scaling up to $816$ shallow qubits with inferred results beyond $34{,}000$ shallow qubits. This is enabled by an exact deep-to-shallow circuit mapping that allows exact sampling from very deep 2D circuits, which we prove are universal. \textbf{(d)} We prove that learning to generate compressed quantum circuits for physical simulation is classically hard. This stands in sharp contrast to the efficient techniques for learning to generate low-depth circuits demonstrated here experimentally on a real device using up to $40$ physical qubits.}
    \label{fig:schematic_cartoon}
\end{figure*}

Generative quantum advantage represents opportunities for quantum computers to expand what is possible in tasks where the goal is to generate new samples from a target probability distribution. Key first steps along this path have been the experimental demonstrations of beyond-classical circuit sampling~\cite{boixo2018characterizing,arute2019quantum,morvan2024phase,gao2025establishing}, which provide evidence for the existence of models where quantum computers exhibit significant advantages in inference. However, the same structural properties of these distributions that support their difficulty in classical inference led many to believe that efficiently learning these models from data would not be possible, severing the link required to view these as learning models rather than simple circuits. Without the ability to both efficiently train and sample (conduct inference) from these distributions, one cannot claim a generative quantum advantage for these tasks. Many works have aimed to frame and establish generative quantum advantage~\cite{perdomo2018opportunities,benedetti2019generative,coyle2020born,sweke2021quantum,gao2022enhancing,zhu2022generative,niu2022entangling,riofrio2024characterization,rudolph2024trainability,hibat2024framework}. However, to date, these methods either embed Shor's algorithm into a learning context inaccessible to experimental demonstration, or fail to establish provably efficient training due to well-known issues such as barren plateaus~\cite{mcclean2018barren,larocca2025barren} or the proliferation of local minima~\cite{anschuetz2022quantum}. In contrast, we experimentally demonstrate provably efficient learning and sampling from distributions accessible to near-term devices while quantifying precise beyond-classical resource crossovers not present in existing theoretical works. We achieve this through an exact mapping we develop between families of deep circuits and shallow circuits in conjunction with experimental results.

In this work, we close the loop initiated by beyond-classical demonstrations to achieve full generative quantum advantage by showing experimentally and theoretically that there exist distributions that may be both learned and sampled efficiently through quantum models, for which sampling is provably beyond the reach of classical computers. This advantage builds upon recent progress in efficiently learning shallow quantum circuits~\cite{huang2024learning,landau2024learning,fefferman2024anti,haah2025short} and the established difficulty in sampling from them~\cite{neill2018blueprint,boixo2018characterizing,haferkamp2020closing,movassagh2023hardness}. Remarkably, the training of these quantum models is classically efficient, but inference requires a quantum computer. The generative quantum advantage is enabled by a quantum model equipped with efficient classical training but no efficient classical inference.

The work is structured as follows. We first define and explain the concept of generative quantum advantage, establishing a theoretical foundation with several key results about the power of shallow quantum circuits and the necessity of quantum computers for efficient model compression. We then review the sewing technique~\cite{huang2024learning} and demonstrate that, in conjunction with a local cost function, it provably simplifies model parameter landscapes by eliminating both barren plateaus~\cite{mcclean2018barren,larocca2025barren} and local minima~\cite{anschuetz2022quantum} in shallow circuits.  To substantiate these theoretical insights, we explicitly define a problem where we achieve generative quantum advantage on a classical distribution that is both easy-to-train and hard-to-sample from without a quantum computer. We support this with experimental demonstrations that map to an effective $816$ shallow qubits, qubits in the shallow representation of the quantum model that we sample using deep quantum circuits with $68$ physical qubits, and an analytical mapping that leverages existing beyond-classical results~\cite{morvan2024phase} to train and emulate in the beyond-classical regime up to an effective $34{,}304$ shallow qubits. This provides an upper bound on the required quantum resources, as more optimal mappings may be discovered in the future. We then present an experimental demonstration of generative quantum advantage in quantum circuit compression, providing insights into how the sewing technique transforms the learning landscape. An overview of our results and major objectives is presented in Fig.~\ref{fig:schematic_cartoon}. We conclude with a discussion of potential applications and future directions for generative quantum advantage.

\section{Generative quantum advantage}

Quantum advantage manifests in various forms~\cite{huang2025vast}. Traditional approaches have focused on computational speedups for specific problems, such as integer factoring or quantum simulation. However, as both quantum technologies and artificial intelligence advance, a new frontier of quantum advantage is emerging. We study an emerging paradigm termed \emph{generative quantum advantage}, which we anticipate will become increasingly significant as quantum technologies and generative machine learning continue to advance. This form of advantage focuses on the ability to learn and generate complex patterns.

\begin{definition}[Generative problem]
A generative problem involves learning to produce outputs that match specific patterns. The learner must develop the ability to generate new outputs that follow these patterns from samples. The patterns may represent mappings from inputs to outputs, underlying distributions, transformations between representations, or circuits that achieve desired behaviors. The patterns can be defined explicitly through specifications or implicitly through examples.
\end{definition}

We emphasize that both learning and inference are essential in our definition of a generative problem. Previous experimental demonstrations have established strong evidence of quantum advantage in generating classically hard probability distributions~\cite{arute2019quantum,morvan2024phase,gao2025establishing}, but useful generative advantage requires efficient learning from data as well. Having established what constitutes a generative problem, we now characterize when quantum computers offer an advantage in solving such problems. The characterization of generative quantum advantage encompasses both the quality and efficiency of the generation process.

\begin{definition}[Generative quantum advantage]
A generative problem exhibits quantum advantage when a quantum computer can learn to generate the desired outputs substantially better than any classical computer. Here, \emph{better} means achieving reduced sample complexity, higher accuracy, faster learning or generation time, or the capability to generate outputs that are practically infeasible for classical computers. The outputs may include classical data, quantum states, circuits, algorithms, or other well-defined objects.
\end{definition}

To illustrate these concepts concretely, we present three fundamental tasks that exemplify the potential for generative quantum advantage. The first involves generating classical bitstrings according to complex probability distributions, the second focuses on generating improved simulation circuits by learning hidden structures, and the third addresses the generation of quantum states under unitary transformations.

\subsubsection*{Example task 1: Learning to generate bitstrings}

\begin{definition}[Learning to generate bitstrings]
Consider an unknown conditional distribution $p(y|x)$ mapping $n$-bit input strings $x$ to $m$-bit output strings $y$. The goal is to learn a generative model that, given any input string $x \in \{0,1\}^n$, can produce output strings $y \in \{0,1\}^m$ sampled from $p(y|x)$. The quality of the generative model is determined by how well its output distribution approximates $p(y|x)$ across different $x$'s.
\end{definition}

This task is facilitated by our model of Generative QNNs, which elevates QNNs to a generative setting through a particular encoding of the input and output processes. We prove these models are universal for all classical conditional probabilities $p(y|x)$. The technical definition of Generative QNN is given in Appendix~\ref{app:sec-gen-adv} as Definition~\ref{def:gen-qnn} and the proof of their universality as Proposition~\ref{app:prop-universal-approx}.

\begin{theorem}[Informal: Classically hard, quantumly easy generative models]\label{thm-inf:qadv-classical}
Under standard complexity-theoretic conjectures, there exist distributions $p(y|x)$ mapping classical $n$-bit strings to $m$-bit strings that a quantum computer can efficiently learn to generate using classical data samples, but are hard to generate with classical computers.
\end{theorem}

We prove this theorem for two families of models we leverage later in this work: tomographically-complete shallow QNNs and instantaneously deep QNNs (IDQNNs), presented in Appendix~\ref{app:sec-gen-adv} as Theorems~\ref{thm:qadv-tomo-shallow-qnn} and~\ref{thm:qadv-inst-deep-qnn}. We demonstrate that these quantum models can be learned from a modest number of classical data samples. Furthermore, an efficient classical learning algorithm suffices to learn these quantum models. However, performing inference and generation from these quantum models is computationally inefficient for classical machines. In particular, the quantum models are designed such that the more $0$'s there are in the input bitstring $x$, the harder it becomes to classically sample from $p(y | x)$, with $x = 0^n$ being the classically hardest case. The training dataset provided to both the classical and quantum generative ML models contains many output bitstrings from $x = 0^n$, along with many output bitstrings with varying fractions of $0$ in $x$. The details are given in Appendix~\ref{app:num-learning-method}.

Our experimental demonstrations for classical bitstring sampling leverage the properties of a special class of shallow generative QNNs along with a shallow-to-deep exact mapping sketched in Fig.~\ref{fig:classical_learning}(b) to optimize experimental resources. This family of shallow quantum generative models permits very efficient classical training, yet are remarkably powerful. They can generate distributions with highly non-local correlations that are classically hard to sample from. With a sufficient number of qubits, generating an output bitstring from the shallow quantum model is equivalent to sampling a deep quantum model and generating a bitstring from the deep quantum model. For this reason, in addition to their constant-time execution, we term them instantaneously deep quantum neural networks (IDQNNs) and define them formally in Definition~\ref{def:inst-deep-QNN} of Appendix~\ref{app:inst-deep-QNN}. IDQNNs yield the following:

\begin{proposition}[Informal: IDQNN samples and generates from deep quantum circuits] \label{prop-inf:idqnn-universal}
For any deep quantum circuit $C$, there exists an IDQNN that samples a deep circuit $D$ from a family containing $C$ and generates a bitstring from the deep circuit $D$ in constant time with at most polynomial overhead in the size of the circuit $C$.
\end{proposition}

This property is stated formally and proven in Appendix~\ref{app:inst-deep-QNN}. This remarkable fact relies on principles used in measurement-based quantum computing (MBQC)~\cite{briegel2009measurement} and insights regarding the power of shallow quantum circuits~\cite{bravyi2018quantum,watts2019exponential,gao2017quantum, bermejo2018architectures, haferkamp2020closing, bergamaschi2024quantum}. In contrast to some previous works, these shallow quantum circuits not only exhibit a separation against restricted classical computational models, but are intractable for any polynomial-time classical algorithms. Assuming that the non-uniform polynomial hierarchy does not collapse, sampling a deep quantum circuit $D$ and generating a bitstring from $D$ is classically hard, as articulated by Conjecture~\ref{conj:Hard-RCS} and proven in Proposition~\ref{prop:Hard-RCS} in Appendix~\ref{app:quantum-advantage-bitstring}. Hence, from Proposition~\ref{prop-inf:idqnn-universal}, accurately sampling from the output of a wide shallow circuit is also classically hard. Unlike shallow classical analogs, these circuits can use their quantum resources to probabilistically sample into the future from circuits that would otherwise require very long execution times. Probabilistic sampling is key here, as generating from a particular deep quantum circuit would require feedforward as well as non-constant depth similar to standard MBQC~\cite{briegel2009measurement}.

While the family of shallow circuits we use is quite powerful, there are other problems where a different type of shallow quantum model can be motivated. In the more general setting, we require an ancilla system to implement inference on the learned quantum generative model, also known as the sewing technique in Ref.~\cite{huang2024learning} and depicted in Fig.~\ref{fig:mbl_results}. Broadly speaking, the sewing technique is a simple method for using ancilla qubits to provably stitch together the action of local unitaries into the unique global unitary transformation for the entire system, even when the local unitaries do not commute and contain additional degrees of freedom. This enables an efficient form of divide-and-conquer learning for shallow circuits. To illustrate the more general learning scenario, we consider a second task: the ability to learn a constant-depth representation of any quantum circuit assuming it has one.

\subsubsection*{Example task 2: Learning to generate compressed simulation circuits}

\begin{definition}[Learning to generate compressed simulation circuits]
Given an input circuit $C$ that simulates a physical evolution, the goal of the generative model is to produce an alternative circuit $C'$ that implements a similar transformation while requiring significantly less execution time, as quantified by circuit depth. The quality of the generative model is determined by the extent to which it reduces circuit depth while maintaining high fidelity between the transformations implemented by $C$ and $C'$.
\end{definition}

At first glance, one might suspect that the promise of the existence of a constant-depth representation for a circuit could enable a classical compiler of quantum circuits to efficiently discover that representation. Indeed, this would be a natural goal for any quantum circuit compiler. However, we prove here that no such efficient classical compiler can exist in the general case, and we reveal that identifying a constant-depth representation efficiently requires a quantum computer, even when one is promised that such a representation exists.

\begin{theorem}[Informal: Learning to generate compressed simulation circuit is classically hard]\label{thm-inf:qadv-compress}
Given a general polynomial-sized input circuit $C$ that is promised to have a constant-depth representation $C'$, finding such a constant-depth circuit representation is not classically efficient unless $\mathsf{BPP}=\mathsf{BQP}$.
\end{theorem}

The precise statement is given as Theorem~\ref{thm:quantum-adv-compressing} in Appendix~\ref{app:quantum-adv-compressing}, along with the accompanying proof. While some circuits may still be amenable to efficient compression via classical compilers or other methods, this result proves that in the general case, the circuit learning and sewing techniques presented here are not efficient without a quantum computer. This theorem also has important implications for questions related to quantum circuit hiding and compilation, which have been of significant interest in the context of verifiable quantum advantage~\cite{aaronson2024verifiable}. In particular, there can exist polynomial-depth representations of constant-depth circuits that no efficient classical adversary can reduce to their constant-depth representation, even with the explicit promise that such a representation exists.

\begin{figure}[t]
    \centering
    \includegraphics[width=0.98\textwidth]{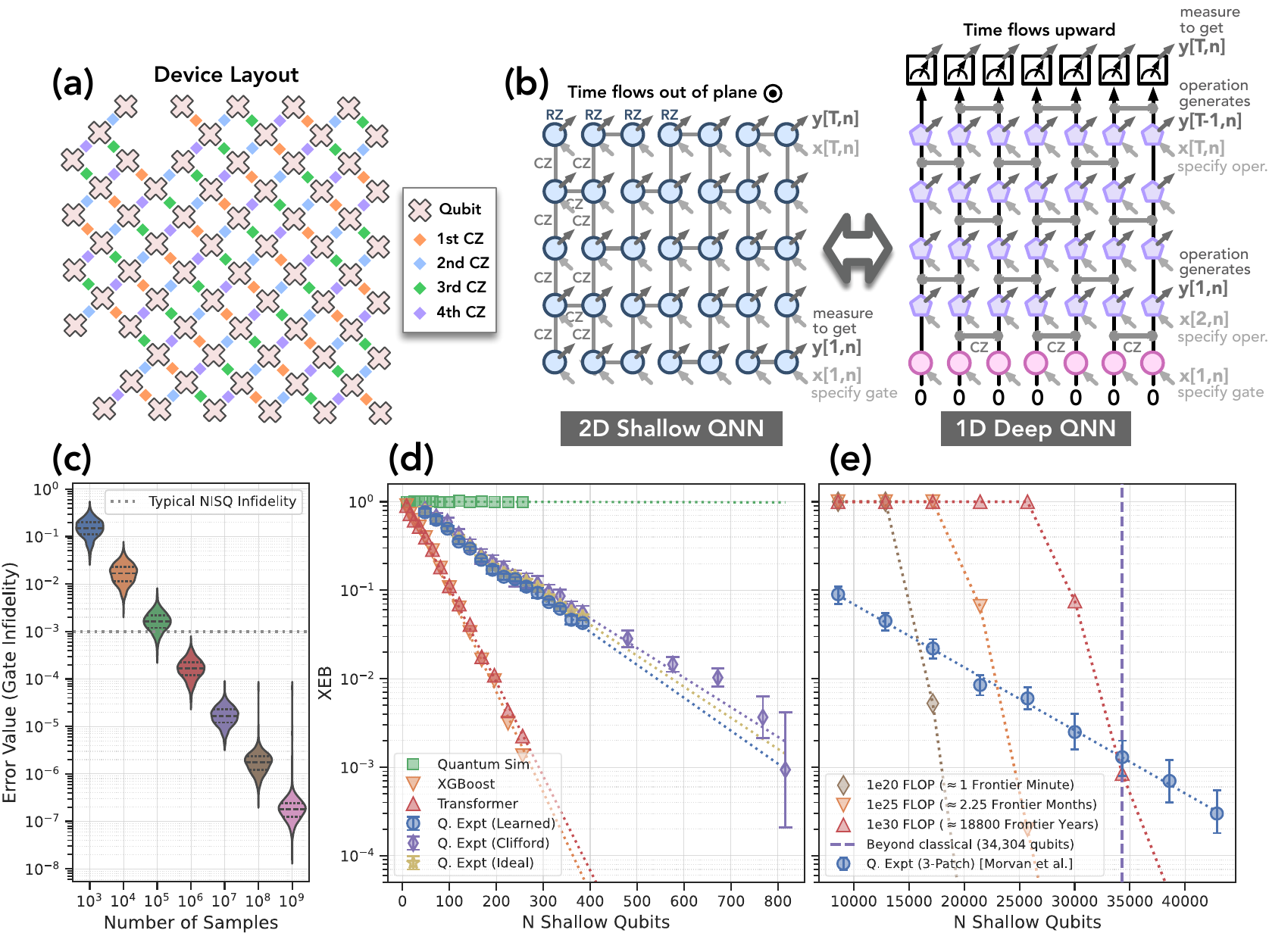}
    \caption{\textbf{Training and inference on classically-hard-to-sample distributions.} \textbf{(a)} A round of CZ gates on the device layout used for this experiment. There are 68 physical qubits on the grid, which correspond to up to 816 shallow qubits in the 3D shallow representation. \textbf{(b)} Schematic correspondence from a 2D shallow circuit to a 1D deep circuit, illustrating the mapping used experimentally between 3D shallow circuits and 2D deep circuits; see Appendix~\ref{app:example_IDQNN_mapping} for details. \textbf{(c)} The number of classical samples required to reach a given gate fidelity in a generative QNN with classically easy training and beyond-classical inference. The training dataset contains output bitstrings from $x=000\ldots 0$ as well as output bitstrings for random inputs $x$ with varying fractions of $0$'s. \textbf{(d)} Performance of quantum and classical models on hardware and simulation as quantified by XEB score. Inference is performed on $x=000\ldots 0$ to remove experimental requirements for feed-forward and maximize quantum-classical separation. \textbf{(e)} Performance inferred in the beyond-classical regime using our mapping and data from Ref.~\cite{morvan2024phase} with 67 physical qubits, evaluated on $x=000 \ldots 0$. Even with hardware noise, quantum models outperform current classical models by a significant margin.}
\label{fig:classical_learning}
\end{figure}

\subsubsection*{Example task 3: Learning to generate quantum states}

In the previous task, we assumed access to some representation of the quantum model to be learned, though perhaps not the most efficient representation. In the general family of generative learning, we may also consider cases where we can provide input states to a particular process and perform local measurements on the resulting states. We define this type of task as learning to generate quantum states.

\begin{definition}[Learning to generate quantum states]
Consider an unknown unitary transformation $U$ acting on $n$-qubit systems. The goal is to learn a quantum generative model that, given any input quantum state $\ket{\psi}$, can produce an output quantum state close to $U\ket{\psi}$. The quality of the quantum generative model is determined by how well its output state approximates $U\ket{\psi}$ across different input states $\ket{\psi}$.
\end{definition}

Until recently, whether this task was efficient for a quantum computer remained an open question even for any target unitary that has a shallow structure. As we show in the following theorem, this task is efficiently solvable on a quantum computer.

\begin{theorem}[Informal: Shallow quantum maps are efficient to learn with local measurements] \label{thm-inf:learning-shallow-circuits}
Given the ability to sample from an unknown shallow unitary $U$ using only local measurements, it is efficient to learn a circuit representation for implementing $U$ using a quantum computer.
\end{theorem}

Precise scaling and definitions for this theorem are provided in Appendix~\ref{sec:math-learning-QNC0} as well as in Ref.~\cite{huang2024learning}. We support this theorem here for the first time with an experimental demonstration, along with additional theorems that provide concrete applications and illuminate the underlying mechanisms responsible for this efficiency.

\subsubsection*{Improved optimization landscape}

At the core of the ability to efficiently learn in all the given example tasks thus far is the capacity to improve the learning landscape through a divide-and-conquer methodology. If one attempts to learn the entire generative process across all qubits simultaneously, the landscape becomes riddled with suboptimal local minima. This occurs even when the target process is essentially entirely classical and the learning model is customized to be well-aligned with the target model. This challenge is highlighted by the following theorem:

\begin{theorem}[Informal: Local swap circuits have exponentially many suboptimal local minima] \label{thm-inf:bad-landscapes}
Consider the task of learning to generate quantum states for a circuit that performs SWAPs on a 1D line between qubits $i$ and $i+3$ with some probability, using a learning model parameterized by nearest-neighbor SWAPs on that same line. This learning landscape contains exponentially many suboptimal local minima with large basins of attraction.
\end{theorem}

A detailed statement is given in Appendix~\ref{app:exp-many-local-minima-training} as Theorem~\ref{thm:exp-many-local-minima-training}, along with the accompanying proof. This means that even though such a problem is essentially classical, free of barren plateaus~\cite{mcclean2018barren}, and employs a model with an inductive bias strongly aligned with the target, any treatment of these problems by standard optimization methods will require exponentially many restarts as a function of system size to succeed with high probability.

In contrast, our divide-and-conquer learning algorithm decomposes the generative QNN over $n$ qubits into overlapping circuit pieces acting on local qubit regions, as illustrated in Fig.~\ref{fig:mbl_results}(b). For shallow generative QNNs, each piece of the circuit acts on a constant number of qubits and can be learned individually by training each piece to locally invert the transformation of the target unitary. After learning each piece independently, they are coherently combined using the sewing technique with ancilla qubits to reconstruct a global $n$-qubit unitary representing the trained generative QNN. When learning these circuit pieces over local regions of qubits, the number of suboptimal basins of attraction becomes constant, and hence so does the success probability for random restarts in standard optimization strategies. This is formalized in the following theorem:

\begin{theorem}[Informal: Landscapes associated with divide-and-conquer learning are benign] \label{thm-inf:good-landscapes}
Given any constant-depth unitary $U$, our divide-and-conquer learning algorithms produce optimization landscapes that are either strongly convex or have a constant number of local minima regions with respect to system size.
\end{theorem}

A detailed proof and statements of this theorem are found in Appendix~\ref{app:benign-landscapes} as Theorems~\ref{thm:strongly-convex-landscape}~and~\ref{thm:benign-landscape}. This theorem implies that any strategy based on learning these overlapping pieces of circuits with random restarts and simple gradient descent will be efficient in learning the target unitary and, by extension, the quantum generative model. Hence, one way of understanding how these techniques achieve efficient learning is by observing how they render the landscapes benign by avoiding both barren plateaus and proliferation of local minima.

We support this numerically in a later experiment depicted in Fig.~\ref{fig:landscape_results}, where one can visually observe the dramatic improvement in the landscapes as well as success probabilities for random restarts. Indeed, even for a modest number of qubits, the overwhelming advantages of our learning techniques become apparent.

\section{Experimental results}

To close the loop on generative quantum advantage, it is essential to demonstrate efficient training and inference in a regime where sampling is classically hard. While our theoretical results establish exact mappings and conditions under which this can be achieved in an asymptotic sense, questions remain about specific applications as well as resource thresholds for advantage with real devices and the particular generative models we introduce. Here we describe and perform experiments on superconducting quantum processors using the models we have introduced to ground these theoretical insights and demonstrate evidence for generative quantum advantage.

\begin{figure}[t]
    \centering
    \includegraphics[width=\textwidth]{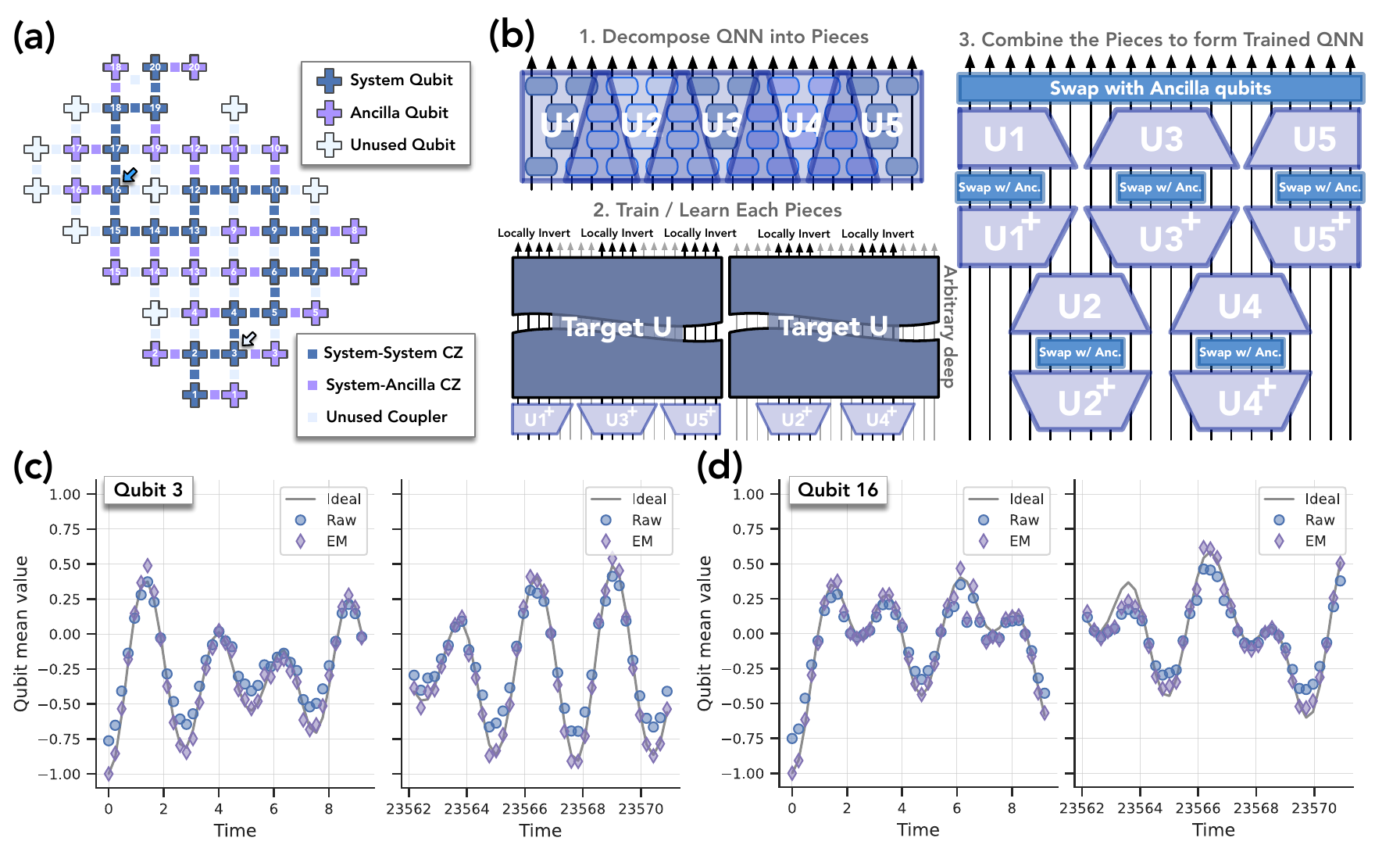}
    \caption{\textbf{Learning quantum generative models.} \textbf{(a)} The layout of system and ancilla qubits embedded into the superconducting qubit device used for learning with the sewing technique. Arrows indicate the qubits whose results are plotted in panels (c) and (d). \textbf{(b.1)} Decomposition of the QNN into circuit pieces, each containing a block of 4 output qubits. \textbf{(b.2)} Each circuit piece is trained to locally invert the target unitary of arbitrary depth. \textbf{(b.3)} The circuit pieces are combined using the sewing technique, demonstrating how learning in blocks of 4 enables parallelization requiring only 2 rounds of application. \textbf{(c)} Experimental results from the trained generative QNN showing the expected value of $Z_i$ for qubit 3, comparing raw data and error-mitigated results against the ideal theoretical prediction. High-quality learning and inference are demonstrated for both short times (left) and very long times (right). \textbf{(d)} The same comparison for qubit 16.}
    \label{fig:mbl_results}
\end{figure}

\subsubsection*{Experiments on learning to generate bitstrings}

Focusing first on the task of learning and generating samples from classically challenging distributions, we describe the experimental realization of an IDQNN in $(2+1)$D, utilizing a Google 2D superconducting chip architecture~\cite{arute2019quantum}. Additional device and performance details are provided in Appendix~\ref{app:device-details}. In the learning procedure, we generate data from an ideal set of pre-chosen parameters $\beta$ and learn these parameters to a precision of $10^{-3}$ through the procedure described in Appendix~\ref{app:qadv-inst-deep-qnn-qease}. We supply the same datasets to classical sequence-to-sequence learning models, namely transformers and XGBoost decision trees, for comparison. The time dimension is fixed to a depth of $d=12$ and we vary the physical qubits up to $68$, corresponding to an effective shallow dimension of $816$ qubits. We also consider a class of related models where $\beta$ is rounded to produce the nearest Clifford circuit, allowing for straightforward prediction of performance beyond sizes that are easily classically simulatable.

The results, as quantified by cross-entropy benchmarking (XEB; see also Appendix~\ref{app:cross-entropy}), measure the fidelity to the ideal distribution and are plotted in Fig.~\ref{fig:classical_learning}. More details on the specific experiments plotted in Fig.~\ref{fig:classical_learning} are provided in Appendix~\ref{app:classical_learning_expt}. The results for quantum models show that, although classical models can learn and generate the distribution for small sizes, as size increases the quantum models demonstrate a considerable performance advantage over the classical models, both of which use classical learning algorithms here. The classical models we tested aim to learn the generative model $p(y | x)$ directly, which has previously been shown to be hard~\cite{niu2020learnability}. In contrast, we train the parameters of the quantum models using classical computers, but generating new samples from the trained quantum models requires a quantum computer and remains classically hard. The Clifford circuit proxies track the performance of quantum models under real noise quite well and predict a sizable quantum performance advantage even with experimental noise at larger sizes.

However, while the IDQNN model can sample from classically hard distributions in the large size limit, it is designed to facilitate efficient learning with its structure and is not optimized to produce hard samples at very low qubit counts. Hence, $816$ shallow qubits are not yet expected to be in the beyond-classical regime. A more efficient quantum circuit simulator able to handle additional structure, such as a tensor network-based approach, may still be able to sample from these distributions at some sizes, constituting a quantum-inspired classical model for this dataset. This raises the question of where the crossover for beyond-classical sampling occurs for these models and whether we can demonstrate efficient learning in that regime as well.

To address this question directly, we use the fact that IDQNNs can sample and generate from deep random quantum circuits as shown in Corollary~\ref{cor:inst-deep-QNN-any} and the precise circuits and performance results from recent beyond-classical experimental demonstrations in Ref.~\cite{morvan2024phase}. Using our exact mapping, we decompose the beyond-classical circuits from Ref.~\cite{morvan2024phase} into an IDQNN form, such that learning and evaluation of any particular instance can be mapped between the two. Here, we are able to evaluate the overhead in the shallow representation of IDQNNs as well as the power of the RCS circuits, finding that an IDQNN representation of the beyond-classical limit is estimated to require $34{,}304$ shallow qubits, mapped to deep circuits on $67$ physical qubits.

In this representation, we can observe the performance of the trained model as a function of the number of samples in Fig.~\ref{fig:classical_learning}c, and we find that if we take typical NISQ infidelity of two-qubit gates today ($10^{-3}$), then around $10^6$ total samples suffice to learn these circuits to that precision. At current repetition rates, this requires only seconds for a quantum computer in both the learning and inference regimes, while the performance for comparable classical computations is shown in Fig.~\ref{fig:classical_learning}e. Even compared to the Frontier supercomputer equipped with approximately $18{,}000$ years of runtime, the quantum model mapped to real hardware today exhibits a performance advantage in learning and generating from this distribution. The inference evaluation is performed for the all-zero input $x=0^n$, whereas the training dataset contains output bitstrings from $x=0^n$ as well as many output bitstrings for random inputs $x$ with varying fractions of $0$'s. Observing outputs from different inputs is essential for efficient learning, as detailed in Appendix~\ref{app:num-learning-method}. This demonstration, together with our theorems, establishes experimental evidence of generative quantum advantage and further quantifies the crossover for our specific trainable model.

\begin{figure}[t]
    \centering
    \includegraphics[width=0.82\textwidth]{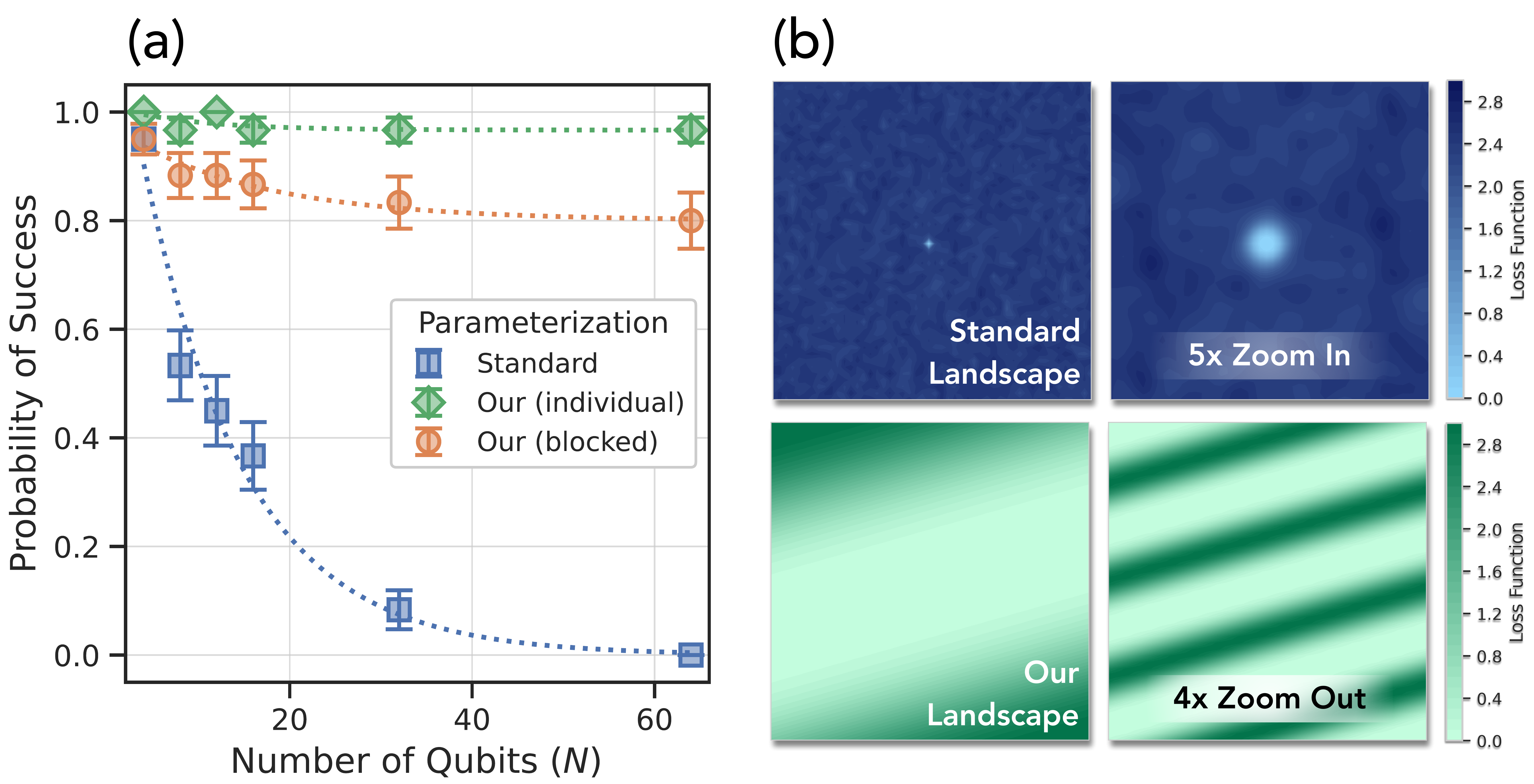}
    \caption{\textbf{Efficient learning from improved optimization landscapes.} \textbf{(a)} The probability of success for learning the SWAP-4 circuit using randomly restarted gradient descent with a local ansatz as a function of the number of qubits. Our divide-and-conquer-based learning algorithm based on learning local circuit pieces avoids the exponential decay in success probability observed with standard approaches. \textbf{(b)} Visualization of optimization landscapes for learning the SWAP-4 circuits with a local ansatz plotted in 2D. The divide-and-conquer nature of the sewing technique dramatically improves the learning landscape from a flat and chaotic terrain filled with local minima (top) to one with large basins of attraction where local minima are also global minima with zero loss function (bottom).}
    \label{fig:landscape_results}
\end{figure}

\subsubsection*{Experiments on learning to generate compressed simulation circuits}

Similar to many successes in classical machine learning, our experiments thus far have utilized a specialized model that facilitates efficient learning, inference, and compatibility with quantum hardware. While it has not been experimentally demonstrated until now, it has recently been shown that any shallow quantum circuit can be learned efficiently in the worst case with only product state inputs and local measurements, though this may require the use of $n$ ancilla qubits and greater depth in inference~\cite{huang2024learning}. This not only has applications in generative learning of classical distributions with potentially lower theoretical overhead than IDQNNs, but also opens an entirely new application domain. Namely, as proven in Theorem~\ref{thm-inf:qadv-compress}, even when given an explicit quantum circuit representation with the promise that a shallow representation exists, it is generally inefficient to find such a representation without a quantum computer. Hence this shallow learning technique can be used for circuit compression or compilation in cases where no classical approach would suffice. Moreover, it can be used to exploit hidden structure in problems to find more efficient representations automatically, which we demonstrate here.

We demonstrate, as an application of shallow circuit learning, the detection of hidden structure in physical dynamics. In doing so, we further develop the divide-and-conquer learning algorithm based on the sewing technique depicted in Fig.~\ref{fig:mbl_results} from Ref.~\cite{huang2024learning}. This allows the final inference sewing technique to be implemented at very low, constant depth and facilitates matching to the layout of the device as shown in Fig.~\ref{fig:mbl_results}. For this experiment in particular, we focus on learning hidden structure in the physical dynamics of a quantum system initially expressed as a quantum circuit, e.g., Trotter evolution. We assume we are given access to circuits representing $U=\exp(-i H t)$ in polynomial depth for $H=V^\dagger H' V$, where $H'$ has a simple, fast-forwardable structure, i.e., $H'=\sum_j h_j Z_j$, and $V$ is a constant-depth, random unitary. Details of $H$ and $U$ are given in Appendix~\ref{app:quantum-adv-compressing}.

While in specific cases, one may be able to successfully devise the shallow, hidden structure from the circuits alone classically, as proven in Theorem~\ref{thm-inf:qadv-compress}, in the most general case one must run the quantum circuits on a quantum computer to find the constant-depth representation efficiently. Here we invoke the general divide-and-conquer strategy to learn and perform inference of this physical system with evolution times $t$ ranging from very short to very long to demonstrate that this technique can locate this hidden, fast-forwardable structure even at very long times. The results and experimental layout are shown in Fig.~\ref{fig:mbl_results}. We observe that at both short and long times, we find a high-accuracy representation implementable in constant depth, and even with current experimental errors we are able to simulate the dynamics to high accuracy using 40 qubits, 20 for the system qubits and 20 for the ancilla register.

\subsubsection*{Numerical experiments on improved optimization landscape}

As shown in Theorems~\ref{thm-inf:bad-landscapes} and~\ref{thm-inf:good-landscapes}, a core element enabling efficient learning is the way in which decomposition and sewing change the landscape of the learning problem. One way of stating this is that the number of local minima becomes constant as a function of system size for any constant block size of qubits, such that simple restarted local optimization methods will be successful. As a final numerical experiment, we follow the example problem used in the proof of Theorems~\ref{thm-inf:bad-landscapes} and~\ref{thm-inf:good-landscapes} and evaluate the probability of success upon random restart for an extremely simple learning target with a model that is tailored to nearly match its structure. 

We observe from the results in Fig.~\ref{fig:landscape_results} that, as expected, learning the whole problem at once is exponentially unlikely to succeed, while learning using a divide-and-conquer strategy focusing on constant-sized blocks succeeds with high probability for any system size. Moreover, the visual depictions of the landscapes show a dramatic and clear improvement in that there are fewer, higher-quality minima with wider basins of attraction. This suggests promising directions for designing more general, efficient learning strategies for quantum models of all kinds.

\section{Discussion}

Through a combination of theoretical proofs and experimental demonstrations, we close the loop on generative quantum advantage that began with beyond-classical demonstrations of random circuit sampling and connects here through efficient learning in the same regimes. For a range of generative tasks, including sampling from classical distributions, compressing quantum circuits, and learning to generate quantum states, we provided evidence and helped quantify the resource crossover for definitive beyond-classical performance with modest resources. With the meteoric rise in interest for classical generative models, including LLMs and diffusion models today, the ability to learn and sample from entirely new and otherwise inaccessible distributions with relatively few quantum resources represents an exciting possibility.

That said, there is substantial work that remains to be done to bring these powerful foundations to fruition in modern machine learning. The precise remaining steps towards \emph{useful} generative quantum advantage remain an open question, but we provide some speculation. First, we believe continuing to expand the classes of models that are known to be both efficiently trainable and provably hard to sample from will be important. Perhaps as a next step, we could generalize the approaches here into a regime that naturally supports floating-point and integer numbers to provide inductive bias matches with real classical data. Second, sharpening the search for practical data distributions (either classical or quantum) that would be compatible with models that exhibit generative quantum advantage will be crucial. While continuing to curate classical data will be important, further investigation should be done to understand the interplay between quantum sensing apparatus coupled to qubits and their quantum data advantages \cite{huang2020power, huang2022quantum}. Finally, much of modern classical machine learning's success has hinged on the fortuitous empirical performance of certain models on real data. Quantum machine learning may follow a similar pattern where the truly optimal models that exhibit empirical or emergent advantage \cite{huang2025vast} benefit more from their trainability and hardware efficiency over their perceived expressivity.

The success of ML has been largely driven by the availability of data, and at its heart, the universe and its data are quantum.  Quantum machine learning has a unique opportunity to take advantage of this data in its natural form, and we believe that for generative processes involving quantum data, quantum computers will be essential.  While such a victory may also be empirical and emergent as in the history of classical ML, we hope this work sets the stage for broadly beneficial quantum generative advantage.

\vspace{-1em}
\section*{Acknowledgments}
\vspace{-1em}
The authors thank Sergio Boixo, Bill Huggins, and David Gosset for detailed comments and discussion on the draft.  The Google Quantum AI team fabricated the processor, built the cryogenic and control systems, optimized the processor performance, and provided the tools that enabled execution of this experiment.
We thank Salvatore Mandr{\`a} and Benjamin Villalonga for help with providing additional data and confirming frontier runtime estimates and methodologies from \cite{morvan2024phase} TABLE 1 when making the beyond classical plot in Figure \ref{fig:classical_learning} part (e).

\bibliographystyle{apsrev4-1_with_title_mod}
\bibliography{references}








\appendix


 

\begin{center}
\vspace{3em}
\textbf{\Large Appendices}
\end{center}
\tableofcontents

\renewcommand{\appendixname}{APPENDIX}
\renewcommand{\thesubsection}{\Alph{section}.\arabic{subsection}}
\renewcommand{\thesubsubsection}{\Alph{section}.\arabic{subsection}.\alph{subsubsection}}
\makeatletter
\renewcommand{\p@subsection}{}
\renewcommand{\p@subsubsection}{}
\renewcommand{\theequation}{S.\arabic{subsection}.\arabic{equation}}
\makeatother

\renewcommand{\thefigure}{S\arabic{figure}}
\renewcommand{\thetable}{S\arabic{table}}
\renewcommand{\figurename}{Figure}
\setcounter{figure}{0}
\setcounter{secnumdepth}{3}

\section{A brief review on quantum information and computation}

This section provides a concise review of basic concepts and notations in quantum information theory used throughout the remaining appendices.

\subsection{Quantum states and operations}

In quantum mechanics, a pure state of an $n$-qubit system is represented by a unit vector $\ket{\psi}$ in a $2^n$-dimensional complex Hilbert space $\mathcal{H} = (\mathbb{C}^2)^{\otimes n}$. The computational basis states are denoted as $\ket{0}, \ket{1}$ for a single qubit, and $\ket{x}$ for $x \in \{0,1\}^n$ in the $n$-qubit case. A quantum state can exist as a linear combination $\ket{\psi} = \sum_{x \in \{0,1\}^n} \alpha_x \ket{x}$ with complex amplitudes $\alpha_x$ satisfying $\sum_{x} |\alpha_x|^2 = 1$. Mixed states, representing statistical ensembles of pure states, are described by density matrices~$\rho$, which are positive semidefinite Hermitian operators with unit trace. For an ensemble given by $\{(p_i, \ket{\psi_i})\}$, the density matrix is $\rho = \sum_i p_i \ketbra{\psi_i}{\psi_i}$.

The evolution of closed quantum systems is governed by unitary operators $U$, satisfying $U^\dagger U = U U^\dagger = I$. Measurements in quantum mechanics are described by projection operators. A measurement in the computational basis transforms the state $\rho$ to $\ketbra{x}{x}$ with probability $p(x) = \bra{x}\rho\ket{x}$. More generally, quantum operations are described by completely positive trace-preserving (CPTP) maps $\mathcal{E}$, which can be represented using Kraus operators $\{E_k\}$ as $\mathcal{E}(\rho) = \sum_k E_k \rho E_k^\dagger$, where $\sum_k E_k^\dagger E_k = I$.

\subsection{Quantum circuits and complexity classes}

Quantum circuits provide a framework for describing quantum computations. The standard building blocks are single-qubit gates (e.g., Hadamard $H$, Pauli gates $X$, $Y$, $Z$, $Z$ rotation gates $R_Z(\theta) = \exp(-i \theta Z / 2)$) and two-qubit gates (e.g., CNOT, CZ). These gates are universal, meaning any unitary operation can be approximated to arbitrary precision using a finite sequence of these gates. A quantum circuit consists of $n$ qubits initialized in the state $\ket{0}^{\otimes n}$, followed by a sequence of quantum gates, and finally measurement of some or all qubits. The depth of a circuit is the length of the longest path from input to output, while its size is the total number of gates.

Complexity classes categorize computational problems based on the resources required by the algorithms. The class $\mathsf{BQP}$ (Bounded-error Quantum Polynomial time) consists of decision problems solvable with bounded error by a quantum computer in polynomial time. The class $\mathsf{BPP}$ (Bounded-error Probabilistic Polynomial time) consists of decision problems solvable with bounded error by a classical computer with access to random bits in polynomial time. The conjecture $\mathsf{BPP} \neq \mathsf{BQP}$ posits that quantum computers can efficiently solve problems that classical computers cannot. We will also consider the class $\mathsf{P/poly}$ (Polynomial time with polynomial-size advice) consists of decision problems solvable with by polynomial-size classical circuits. Adleman's theorem states that $\mathsf{BPP} \subseteq \mathsf{P/poly}$.

\subsection{Quantum channels and distance measures}

The total variation distance (TVD) between two probability distributions $p$ and $q$ is defined as:
\begin{equation}
\mathrm{TVD}(p, q) = \frac{1}{2}\sum_x |p(x) - q(x)|.
\end{equation}
For quantum states, the trace distance between density matrices $\rho$ and $\sigma$ is defined as:
\begin{equation}
\|\rho - \sigma\|_1 = \frac{1}{2}\mathrm{Tr}|\rho - \sigma|.
\end{equation}
The diamond norm provides a measure for the distance between quantum channels. For CPTP maps $\mathcal{E}$ and $\mathcal{F}$, the diamond norm distance is defined as:
\begin{equation}
\|\mathcal{E} - \mathcal{F}\|_\diamond = \max_{\rho} \|(\mathcal{E} \otimes I - \mathcal{F} \otimes I)(\rho)\|_1,
\end{equation}
where the maximization is over all density matrices $\rho$ of appropriate dimension, including an ancillary system, and $\|\cdot\|_1$ is the trace norm. This measure is operationally significant because it bounds the maximum probability of distinguishing between two channels by any physical experiment.

\subsection{Cross-entropy benchmarking}
\label{app:cross-entropy}
In the development of random circuit sampling~\cite{boixo2018characterizing,arute2019quantum}, an important question is the best way to characterize high fidelity implementation of the circuits.  One family of measures that works well for random circuits is so-called cross-entropy benchmarking (XEB)~\cite{boixo2018characterizing}.  The metrics in this family assume strong access to the ideal probability distribution for a bitstring $s$, $p(s)$, typically derived from simulation of the quantum circuit.  

We use the linear XEB fidelity in this work for all cases where XEB is referred to:
\begin{align}
    \mathcal{F}_\mathrm{XEB}(E) = \frac{\langle p(s) \rangle_{s \in E} - \langle p(s) \rangle_{s \in \mathcal U} }{\langle p(s) \rangle_{s \in p} - \langle p(s)\rangle_{s \in \mathcal U}}.
\end{align}
Here, $\langle \rangle_{s \in S}$ denotes average over samples from a distribution $S$, $E$ is the experimental distribution, $\mathcal{U}$ is the uniform distribution over bitstrings, and $p$ is true distribution for the exact target.  In the large system size limit $n \gg 1$, this converges to the simpler and more common expression
\begin{align}
    \mathcal{F}_\mathrm{XEB}(E) = \left\langle 2^n p(s) - 1 \right\rangle_{s \in E}
\end{align}
where $n$ is the total number of qubits. While both estimators require the true probability distribution that can be hard to obtain in general, it can also be used for specialized circuit families where probabilities are easy to obtain like Clifford circuits to facilitate benchmarking as we do in the main text.

\subsection{Quantum learning theory}

Quantum machine learning explores how quantum systems can be leveraged for learning tasks. A quantum neural network (QNN) \cite{farhi2018classification} is a parameterized quantum circuit $U(\vec{\theta})$ that maps input quantum states to output quantum states. The parameters $\vec{\theta}$ are trained to minimize a cost function that depends on the specific learning task. To overcome the barren plateau problem in QNN training \cite{mcclean2018barren}, where gradients vanish exponentially with system size, local cost functions are commonly employed \cite{cerezo2021cost}. These functions focus on local properties rather than global properties, providing more favorable optimization landscapes.

Quantum learning theory investigates fundamental questions about what unknown objects are efficiently learnable in the quantum world. In this manuscript, we will rely heavily on the classical shadow formalism \cite{huang2020predicting, huang2021efficient, huang2021provably, elben2022randomized}. Classical shadows provide an efficient method to extract classical information from quantum states using randomized measurements. For an unknown state $\rho$, a classical shadow $S_T(\rho)$ is constructed from $T$ randomized measurements, which enable the estimation of expectation values of $\exp(\mathcal{O}(T))$ observables up to a small error.

\section{Device details}
\label{app:device-details}
The experimental data reported in this work was collected on a superconducting qubit processor~\cite{arute2019quantum,morvan2024phase}.  This device had 68 qubits with a typical connectivity per-qubit of roughly $3.47$.  While the detailed per qubit and pair error rates for this specific experiment are reported below, typical error rates for this architecture are around $10^{-3}$ for single qubit gates, $10^{-2}$ for two-qubit gates, and $10^{-2}$ for measurement operations. We use the same device used in ~\cite{morvan2024phase} with similar performance. We provide detailed error statistics in Figs.~\ref{fig:benchmark3},~\ref{fig:benchmark4} to illustrate the difference in single qubit and two qubit gate performance.  This architecture has enabled both successful surface code error correction experiments~\cite{google2023suppressing}, and random circuit sampling experiments~\cite{arute2019quantum,morvan2024phase,gao2025establishing} that are expected to require $>10^{25}$ years to complete classically with current hardware and algorithms. We employ dynamical decoupling and readout-debiasing in all of these experiment runs. The following figures provide the benchmarks for the hardware used in our experiments.

\begin{figure}[ht]
    \centering
    \includegraphics[width=1\textwidth]{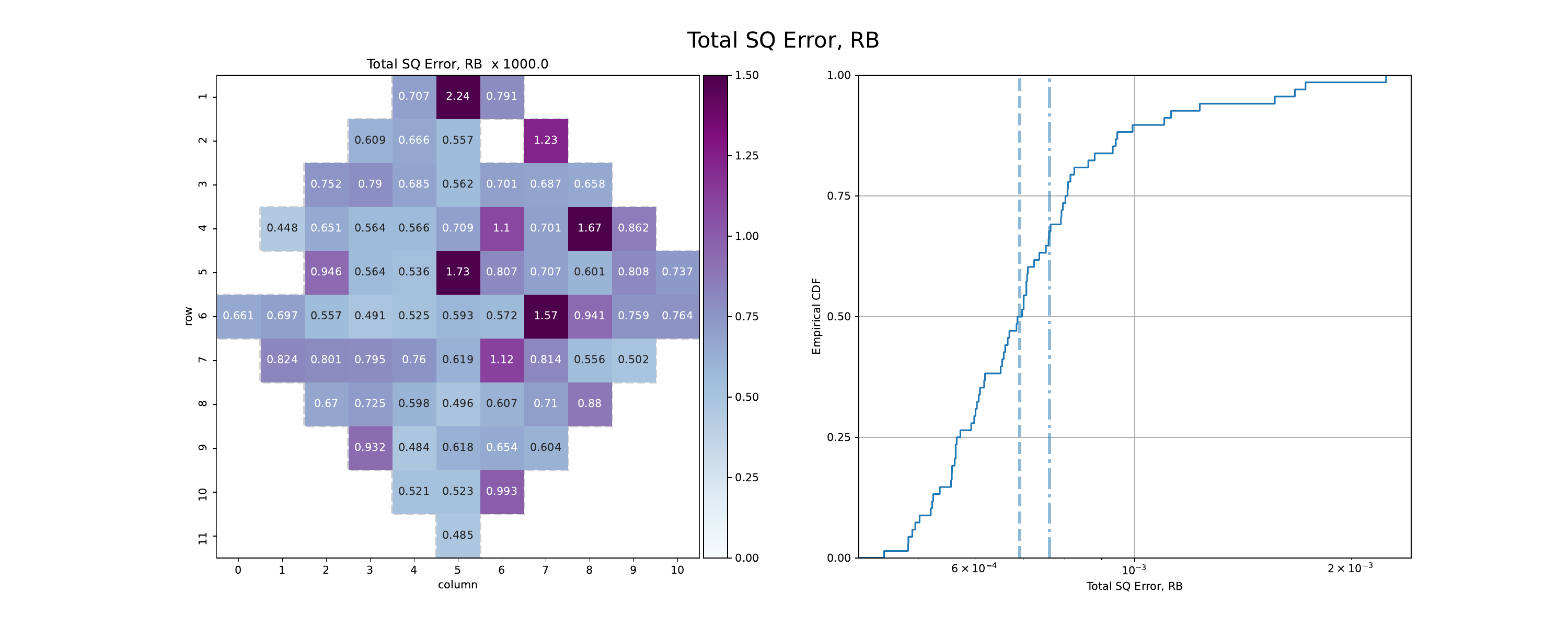}
    \caption{Single qubit pauli errors, computed with randomized benchmarking \cite{emerson2005scalable}. We show the Pauli error on each qubit as well as the integrated histogram with mean and median. To convert between Pauli Error and Average RB error under a single qubit depolarizing channel we use $p_q = \frac{3}{2} \epsilon_{\text{RB}}$}
    \label{fig:benchmark3}
\end{figure}

\begin{figure}[ht]
    \centering
    \includegraphics[width=1\textwidth]{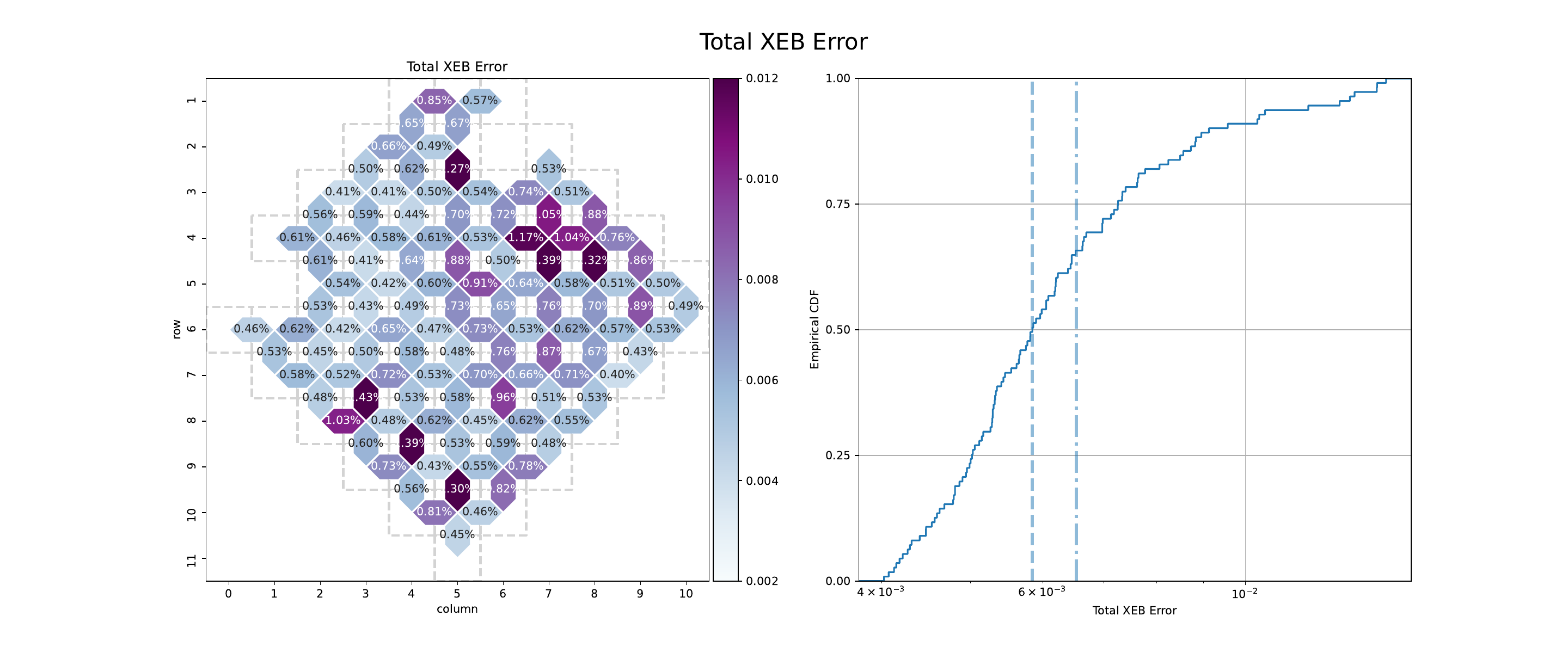}
    \caption{Two qubit total Pauli error, measured with parallel two qubit XEB (includes single qubit errors as well) \cite{arute2019quantum}. We show the error percentage between all couplers in the grid as well as the integrated histogram with mean and median.}
    \label{fig:benchmark4}
\end{figure}

\section{Experiments on learning generative model for classical bitstrings}

\subsection{Problem formulation and experimental setup}
Of the applications discussed, perhaps the most interesting is the ability to efficiently learn and sample from hard classical bitstring distributions as this closes the loop on generative quantum advantage for a problem that may be specified entirely within the classical domain.  Here we consider the task that given an $n$ bit string $x$ as input, we wish to sample an $n$ bit string $y$ from a distribution $p(y|x)$ where this distribution will be hard to sample from classically in the large size limit, and $p(y|x)$ is known through a relatively small number $N$ of training samples of the form $(x_i, y_i)$ for $i$ from $1$ to $N$. 

As a machine learning model, we consider an instantaneously deep quantum neural network (IDQNN) in 3D with free parameters $\beta$, as detailed in the general case in Appendix ~\ref{app:inst-deep-QNN}, which also elaborates on the learning and inference procedures.  The algorithm proceeds in two parts: training and inference.
\begin{enumerate}
  \item \textbf{Training:} Given a training dataset $\{(x_i, y_i)\}_{i=1}^N$, we learn the $\beta$'s that completely determine our machine learning model.  The learning step may be done efficiently classically for this model and to high precision as the number of samples $N$ scales as $1/\varepsilon^2$ for the desired error $\varepsilon$.
  \item \textbf{Inference:} In this step, given any input string $x$, we use our trained model on the quantum computer to draw samples $y$. The quality of the inference is evaluated by cross-entropy benchmarking (XEB) using the drawn samples and the true probability distribution $p(y|x)$ determined by the exact quantum model.
\end{enumerate}

To evaluate empirical performance at larger sizes than are easy to evaluate exactly with XEB, we use related Clifford circuits and a technique to map existing beyond-classical results into our framework. Both are detailed in subsequent sections.
We compare to the performance of classical ML models such as transformers and XGBoost given the same training dataset $\{(x_i, y_i)\}_{i=1}^N$ and test input $x$. To show that we can learn and perform in the classically intractable regime, we exploit a correspondence between shallow and deep circuits to map recent beyond-classical demonstrations~\cite{morvan2024phase} exactly into our shallow circuit framework.

\subsection{Numerical methodology}
\subsubsection{Learning the quantum model}
\label{app:num-learning-method}

The shallow quantum generative model that we focus on for this experiment has the important property that spatially local measurement statistics completely determine the values of $\beta$ and vice versa, while sampling from the entire, non-local distribution remains hard without a quantum computer. This means that we can efficiently emulate the learning process on an arbitrarily sized network and understand its behavior as a function of the number of samples, the size of the network, and the specific distribution from which $x$ is drawn.

To achieve this, we take a model with a given set of $\beta$ parameters and a distribution $D(x)$ over input bitstrings $x$ with a local decoupling property and use it to produce local measurement data for each qubit for a given number of samples. The local decoupling property requires that for any index $j$ in the bitstring $x$, the probability that the neighboring bits of the $j$-th bit are $1$ while the $j$-th bit itself is $0$ for an input bitstring $x$ drawn from $D$ is at least inverse-polynomial in the system size. This property is satisfied by most input distributions $D(x)$ that do not have full support on all-zero or all-one input bitstrings. Examples of distributions that satisfy the local decoupling property include:
\begin{itemize}
    \item Uniform distribution over all input bitstrings.
    \item Each bit in the input bitstring is sampled independently and identically distributed (i.i.d.) to be $0$ with probability $1-p$ or $1$ with probability $p$ for any $\frac{1}{\mathrm{poly}(n)} \leq p \leq 1 - \frac{1}{\mathrm{poly}(n)}$.
    \item Any mixture of local-decoupling distributions.
    \item Any mixture of distributions such that at least $\frac{1}{\mathrm{poly}(n)}$ fraction is local decoupling.
\end{itemize}
The local decoupling property allows us to efficiently learn the model parameters $\beta$ for an arbitrarily large system from a classical training dataset $\{(x_i, y_i)\}_{i=1}^N$ with polynomial sample size $N = \mathrm{poly}(n)$, where $x_i$ is sampled from $D$ and $y_i$ is sampled according to the ideal distribution $p(y_i | x_i)$.

From the training dataset $\{(x_i, y_i)\}_{i=1}^N$, we consider a subset of the classical samples to learn the value of $\beta^j$ for any $j$. In particular, we focus on the $x_i$'s such that the $j$-th bit is $0$ while all the bits neighboring bit $j$ are $1$. We denote the $j$-th bit of the output bitstring $y_i$ such that the corresponding input bitstring $x_i$ satisfies the above condition as $y^{j}_{(t)}$ for $t$ from $1$ to $N_{\mathrm{sp}}$. Using this subset of the training set, an empirical estimate for $\beta^j$ is given by
\begin{align}
  \hat{\beta}^j = \frac{1}{2} \arccos\left( 1 - 2\frac{1}{N_{\mathrm{sp}}} \sum_{t=1}^{N_{\mathrm{sp}}} y^{j}_{(t)} \right).
\end{align}
This learning procedure is described in more detail in Appendix~\ref{app:quantum-advantage-bitstring}. The learned values of $\beta$ for a finite sample size can then be used to define a learned model to be compared in both experiment and simulation to the exact model. We observe empirically that this learning is extremely efficient even for beyond-classical circuits (described later) in Fig.~\ref{fig:classical_learning}c, requiring only $10^{6}$ samples to reach accuracies beyond the precision of current quantum hardware, achieving typical NISQ infidelity of $10^{-3}$ on all gates across the model as determined by the best two-qubit gates found in highly tuned experiments like Ref.~\cite{morvan2024phase}.

In the experiments, we consider the distribution $D(x)$ to be the following:
\begin{itemize}
    \item With probability $1/3$, every bit in $x$ is $0$, i.e., $x = 0^n$.
    \item With probability $1/3$, every bit in $x$ is sampled i.i.d. to be $0$ with probability $0.6$ and $1$ with probability $0.4$.
    \item With probability $1/3$, every bit in $x$ is sampled i.i.d. to be $0$ with probability $0.2$ and $1$ with probability $0.8$.
\end{itemize}
Since this is a uniform mixture of two local-decoupling distributions and a singleton distribution, $D(x)$ is a local decoupling distribution. Hence, the training of the generative quantum model can be performed efficiently even on a classical computer.

We intentionally place heavy weight on the all-zero input bitstring $x = 0^n$ since we will focus on performing inference for this input bitstring. This input bitstring $x = 0^n$ provides no useful information for learning the quantum generative model parameters using our proposed method since it lacks any $1$ bits. In contrast, having data from $x = 0^n$ makes the task much easier for classical models as they can directly observe the output distributions for the specific input that will be tested during inference. Note that the learning task remains highly nontrivial and classically hard even when we provide significant training data on the input bitstring $x = 0^n$ to be tested in the inference stage. This differs from classification tasks where having training data for the test point immediately reveals the answer to the classifier. For generative tasks, having access to the testing input $x = 0^n$ in the training set can still remain challenging because the generative models must learn the distribution $p(y | x = 0^n)$. One could imagine asking a large language model to ``Create a random cat image,'' which remains a nontrivial task even if this question has appeared many times in its training set since the generative AI model needs to learn how to generate a random image of a cat.

As we have shown above, for any locally decoupling distribution $D(x)$, the learning of the quantum model can be performed efficiently on a classical computer since it only depends on the local neighborhood. However, inference on the same quantum model remains hard for a wide range of input bitstrings $x$. As shown explicitly in Figure~\ref{fig:2Dshallow-1Ddeep}, each bit in $x$ equal to $1$ corresponds to performing a mid-circuit measurement and qubit reset, whereas each bit equal to $0$ corresponds to performing a quantum gate. As the number of $1$'s increases, the inference problem becomes easier for methods like tensor network simulations that can take advantage of intermediate measurement results to perform vertex projections. However, as long as the fraction of $1$'s is sufficiently low, one should expect the inference problem to remain computationally hard for classical computers. 

For example, when the probability of having $1$ on each bit is $\frac{1}{\mathrm{poly}(n)}$ and i.i.d., the resulting distribution $D(x)$ will satisfy the local decoupling property, and hence the quantum model will remain easy to learn. However, performing inference on input bitstrings $x$ sampled from $D(x)$ remains hard since having mid-circuit measurement and reset at a few random locations within a deep quantum circuit remains computationally hard for classical computers. Furthermore, the study of measurement-induced phase transitions suggests that even when the fraction of gate locations where a mid-circuit measurement occurs is a constant, when the constant is below a threshold, the resulting quantum state conditioned on the measurement outcomes will still be highly entangled and hence sampling from a deep quantum circuit with a small constant fraction of mid-circuit measurements remains classically hard. Moreover, we note that even in the presence of local, finite temperature thermal noise in the quantum device, the simulation of shallow quantum circuits similar to the generative quantum models considered in this work is still expected to be classically hard asymptotically~\cite{bergamaschi2024quantum}.

Clearly, completely cutting a system into pieces via measurement will reduce the classical simulation difficulty. More generally, the fraction of $1$'s could reduce the effective number of qubits substantially for methods like tensor network approaches, which must be taken into consideration when trying to identify specific beyond-classical regimes up to constant factors. Moreover, if one dimension grows much more rapidly in size than the other, shape bottlenecks may permit additional simulation enhancements that improve classical approximations of the quantum circuit even further, which must be carefully accounted for. We leave these detailed studies of finite crossovers of computational hardness transitions as we vary the fraction of $1$'s in the input bitstring $x$ as open questions.

\subsubsection{Learning with classical ML models}

We compare the performance of all of these cases to state of the art sequence to sequence models on the data such as transformers and boosting models.  The classical input data $x$ is input first as the sequence of bits, then each of the output bits is predicted using the inputs and the previous output bits as input, $(x, y)$.  The order of the input and output bits is flattened into 1D consistently in each step by unwinding the 3 dimensions in a standard way to match the input and output formats of these models.  In the case of the boosting model, we use independent boosting models for each output bit in order, trained on the data to predict the next bit.  

In more detail, to compare to classical ML models for the same data, we train a combination of XGBoost and Transformer models to autoregressively predict the next output bit. We train $d * n$ XGBoost classifiers (one for each next bit prediction task) each classifier has the following configuration:

\begin{verbatim}
  n_estimators=100, eta=0.3, max_depth=6, subsample=1.0, tree_method='auto'  
\end{verbatim}

For the transformer model, we use a decoder only architecture with between $10^6$ and $10^8$ parameters (the only modification is we don't need token embeddings for a vocabulary of just 0 and 1 bits). We train a single model for each problem instance with the following configuration

\begin{verbatim}
n_layers=4-12, ff_dim=128-1024, n_heads=2-12, hidden_size=32-512
\end{verbatim}

We train with a linear ramp cosine decay learning rate for 50 epochs.  Structurally and intutively these models are very similar to the GPT2-small networks~\cite{radford2019language}.

Once these three generative models have been "trained", we sample $200000$ bitstrings from each of them and compare the KL divergence as well as the XEB score between these samples and the original problem instance.  This data is plotted in Fig.~\ref{fig:classical_learning} alongside the data from the quantum models and related Clifford models.  We see that while these models are accurate at very small sizes, their performance quickly degrades in comparison to the quantum models.

Both classical models achieve a similar level of performance as the system size grows as quantified by XEB from the true output distribution.  However the performance of both was quite poor, and substantially worse than the quantum models even implemented with real hardware noise up to scaling.  This however, does not address the question of quantum inspired models powered by advanced classical simulation, which we address subsequently.

\subsubsection{Beyond classical performance}

In order to benchmark the learning and inference steps we run both exact numerical simulations for small sizes as well as examine a set of related Clifford experiments for which efficient simulation is possible at larger sizes. This is shown in Fig.~\ref{fig:classical_learning} d) (left panel). However, this does not directly address the question of learning and inference with circuits that are believed to be beyond classical, such as those implemented in Ref.~\cite{morvan2024phase}.

To address performance and learning in the beyond classical regime, we take an approach that utilizes the universality of our shallow circuit learning distributions, an emulation of the learning process, and existing results in the literature to provably predict performance.  More specifically, we take the following approach:

\begin{enumerate}
    \item We take explicit beyond classical circuits from Ref.~\cite{morvan2024phase}, and using the universality of our sampling distribution, map them into a shallow 3D quantum circuit via refitting to layers of Hadamards and CZ's as in our deep construction, specialized to the case of $x=000...$, $y=000...$. This means that in the case of an all 1 input state and all $0$ output, we sample exactly from any particular circuit taken from Ref.~\cite{morvan2024phase}, and for any other output state $y$, we sample from an equivalently hard circuit with some number of single qubit $Z$ gates randomly injected.  This yields our shallow model.
    \item Then, as before, we emulate the training of our model $\beta$s with local sampling and fitting, and study the convergence of the error in $\beta$ with respect to number of samples.
    \item After learning the model to high precision for a number of different samples, we translate these errors back to the original circuits, i.e. their original compilations with adjusted parameters, by reversing step 1 numerically.  We then plot the individual two-qubit gate errors as a function of the number of samples in Fig.~\ref{fig:classical_learning} c), to show that even for these circuits, it is easy to get below the precision of the machine with about $10^6$ samples.  We define typical NISQ infidelity to be the best value of the two-qubit gate error achieved in Ref.~\cite{morvan2024phase}.
\end{enumerate}

As we are able to learn from data efficiently in these circuits to beyond the precision of the machine implementing these experiments in addition to an exact circuit correspondence, we use as a proxy for our performance on these circuits their reported XEB values.  This is plotted in Fig.~\ref{fig:classical_learning} d) (right panel) with the 3-Patch data from Ref.~\cite{morvan2024phase}. Also in this plot are the estimated maximum possible fidelities achievable in simulation with the Frontier supercomputer when allotted different total FLOP budgets. Simulation methods to attempt to spoof XEB follow closely to the ones described in Ref.~\cite{morvan2024phase}, with key assumptions found in TABLE 1. Using $\text{FrontierFLOPS}_{\text{target fidelity}} \approx  \text{FrontierFLOPS}_{\text{perfect fidelity}} * F_{\text{target}}$, where we know $\text{FrontierFLOPS}_{\text{perfect fidelity}}$ from Ref.~\cite{morvan2024phase} and fix $\text{FrontierFLOPS}_{\text{target fidelity}}$ to one of $10^{20}$, $10^{25}$, $10^{30}$ we can plot the maximum possible $F_{\text{target}}$ XEB fidelities on Frontier for these FLOP budgets. Finally we plot the beyond classical supremacy frontier (vertical line) corresponding to the beyond classical circuits run in Ref.~\cite{morvan2024phase}. We see from this that our learning technique achieves generative quantum advantage by performing inference on a select input present in our training set $x=000...0$, exploiting the availability of quantum models against both classical models and advanced classical simulation of quantum models, in the beyond-classical circuit regime.

\subsection{Experimental implementation}
\label{app:classical_learning_expt}
Performing a prototype experiment helps to clarify the concrete steps required and the current state of the art of quantum hardware.  As always, generality is at odds with parsimoniousness of resources, so we make a number of choices to aid in experimental implementation without explicitly sacrificing opportunities for quantum advantage in some reasonable limit.  Here we detail the experimental design for the generative classical bitstring task.

As we are implementing on a 2D Sycamore architecture that has as strengths high quality and fast gates, we consider a shallow 3D, IDQNN model that will be implemented as $(2+1)$D borrowing one dimension for time as detailed in more generality in Appendix~\ref{app:inst-deep-QNN}.  The input bitstring $x$ in the deep model determines both the initial state, and classically controlled measurement and feedforward operations later in the circuit.  The output bitstring is determined by the measurement of all the qubits in the shallow circuit, or equivalently mid-circuit measurements in the deep network.

The results for these experiments and other classical simulations are shown in Fig.~\ref{fig:classical_learning}.  To provide a bit more detail in the presentation of the results we expand on some of the subfigures here.  In Fig.~\ref{fig:classical_learning}c, we first map a particular beyond classical circuit from Ref.~\cite{morvan2024phase} into the form of our IDQNN circuits.  This circuit has 67 physical qubits, and the mapping into this representation corresponds to the point at the dashed vertical line in Fig.~\ref{fig:classical_learning}e, mapping to roughly $34,000$ shallow qubits.  We then perform a classical simulation of the learning process on the distribution $D(x)$ for this circuit representative, and evaluate the distribution of gate fidelities with respect to the true gates as a function of the number of samples, and this is plotted as the violin histogram in Fig.~\ref{fig:classical_learning}c.

In Fig.~\ref{fig:classical_learning}d, we plot the new experimental results taken on a device in this paper.  For our implementation, we bypass the need for mid circuit measurement by only evaluating the circuits in hardware for the $x=000...$ state and randomly initializing all but the final layer of $y$.  This corresponds to the most challenging circuits in terms of entanglement generation, and hence is a good benchmark for performance.

To evaluate performance experimentally, we then implement the deep forms of the circuit on the superconducting quantum processor. Three repetitions of the CZ circuits sketched in Fig.~\ref{fig:classical_learning}a are used, totaling to a model depth of $12$, to balance classical hardness of simulation and experimental fidelity and represents a spatial dimension of $12$ represented in the time axis.  The total shallow qubits is $12$ times the number of physical qubits used in the deep circuit, which we range up to a maximum value of $n=68$.  This corresponds to a shallow circuit performed on $816$ qubits.  
In order to control for different sources of error, we implement both the circuits with the ideal $\beta$'s (Q. Expt(Ideal)) and the learned $\beta$'s (Q. Expt(Learned)) from the emulation and learning procedure above, and evaluate their XEB score with the ideal circuit (Quantum Sim).  In simulation, we find that the learning procedure works to a very high degree of accuracy (XEB score $\approx 1$ for all sizes) and in experiment we find that the learned and ideal values of $\beta$ perform similarly, with degradation only due to finite fidelity of the device.  For these experiments we use a fixed number of samples of $200000$ to learn the values of $\beta$ to a finite precision.  Moreover, for each size of circuits we draw $5$ independent sets of $\beta$s to repeat the learning and inference process.  We note that in the experiments, we perform only dynamical decoupling during long periods of idling and a basic form of readout debiasing that involves performing an X gate on all qubits before the end of half the runs, and flipping the values to collect half the bitstrings in this fashion and half without any additional action.  No additional error mitigation is implemented. 

In order to predict performance of quantum learning at sizes that are not easy to simulate and verify classically, we also implement a set of Clifford test experiments (Q. Expt Clifford).  This is done by taking the $\beta$'s that define the shallow learning model, and rounding $2 \beta / \pi$ to the nearest integer to give a related Clifford model.  Note that this model is evaluated with respect to the ideal Clifford model rather than the ideal original $\beta$ circuit.  This model can then be evaluated efficiently classically, and used to check the accuracy of the corresponding model in the quantum case.  We observe that up until the point we cannot predict performance classically, this model tracks the experimental performance extremely well with regards to XEB.  This model allows us to observe strong experimental performance up to the equivalent of $816$ shallow qubits by using our deep model.  

In Fig.~\ref{fig:classical_learning}e, we use the circuits from mapped from Ref.~\cite{morvan2024phase} into our IDQNN model to determine the number of shallow qubits that would be in our model, and use the simulation and performance data from that same work to establish what the XEB would be for a high performance classical simulation of the same data.  We note that this inference is once again for the $x=000...0$ input, which maximizes the classical hardness.  Although the learned circuit in Fig.~\ref{fig:classical_learning}c is the same circuit, it is learned on a different input distribution, and we discuss the expected changes in hardness if one were to do inference on a broad class of distributions in Appendix~\ref{app:num-learning-method}.

\subsubsection{Example 1D IDQNN mapping}
\label{app:example_IDQNN_mapping}
The description of the most general IDQNN found in Appendix.~\ref{app:inst-deep-QNN} can be made easier to understand through a simple example on 4 qubits, where the model is shallow in $2D$ and the mapping to $(1+1)D$ is easy to visualize in a quantum circuit diagram.  For illustration purposes, in 1D at small sizes the deep representation of the shallow circuit may be depicted as in Fig.~\ref{fig:2Dshallow-1Ddeep}
\begin{figure}
    \centering
    \includegraphics[width=1.0\linewidth]{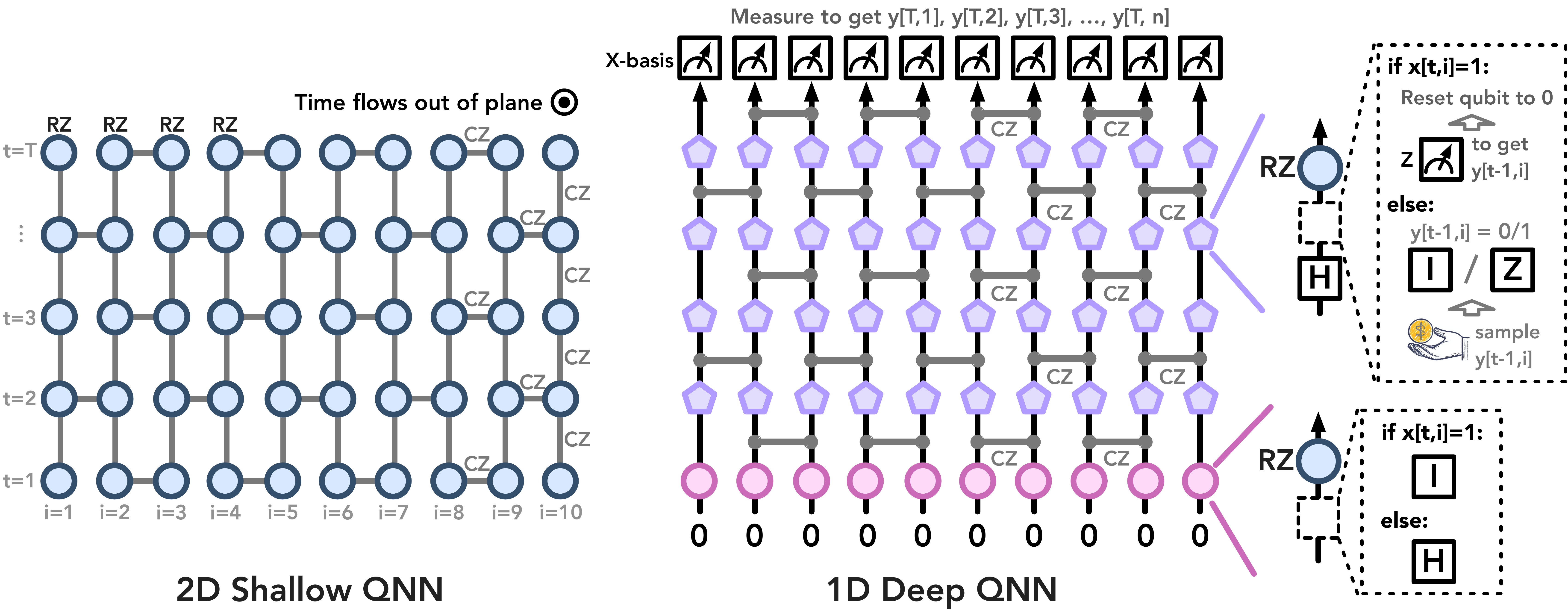}
    \caption{Example mapping from a shallow 2D IDQNN to its (1+1)D deep representation.  On the left, we depict a shallow circuit on a 2D grid of qubits indexed by $(i, t)$.  It consists of a layer of controlled Hadamards, RZ rotations by a single qubit parameter $\beta_{i, t}$, CZ gates indicated by connections between qubits, and measurements in the X basis to map to outputs $y_{i, t}$.  On the right, we depict the deep $(1+1)$D deep mapping amenable to being run on current hardware with time running from bottom to top, and the physical qubits being those in the lowest layer.  The operations are expanded upon in the side graphic showing the details of feedforward required in general cases, and we include a pseudo-code description of both in Appendix~\ref{app:example_IDQNN_mapping} for additional clarity.}
    \label{fig:2Dshallow-1Ddeep}
\end{figure}

To make this figure more clear, we spell out here the steps on the shallow circuit in 2D as well as the equivalent construction for a deep 1D circuit in the following pseudocode.  We assume that the 2D shallow model is defined by a graph $G$ that has qubits as nodes and connections between them as edges. For simplicity and maximum classical hardness, we assume here the time direction is fully connected, and the non-trivial edges will only be needed in the spatial direction.  For example, see the gray lines in Fig.~\ref{fig:2Dshallow-1Ddeep} on the 2D Shallow QNN.  The shallow and deep models are equivalent in that given the inputs $x$ indexed by $(t, i)$, they sample the bits $y$ also indexed by $(t, i)$ with precisely the same probability.  The equivalent shallow and deep models are made more explicit by their implementations which we precisely give as Algorithm~\ref{alg:2d_shallow_model} and ~\ref{alg:1p1d_deep_model}.

\begin{algorithm}[ht]
\caption{2D Shallow Model}\label{alg:2d_shallow_model}
\begin{algorithmic}[1] 
\Require inputs $\{x_{t,i}\}$ 
\Statex $T$ total time.
\Statex $G$ is a 2D graph with node labels indexing qubits by $(t, i)$.
\Statex $\beta_{t,i}$ are the model parameters.
\Statex The final circuit will be 2D shallow on $T * N$ qubits.
\Ensure Bitstring measurements $\{y_{t,i}\}$.
\State Initialize all qubits to the state $\ket{0}$.
\State Apply $H^{1-x_{t, i}}$ to all qubits $(t, i)$.
\State Apply $R_Z(\beta_{t, i})$ to all qubits $(t, i)$.
\State Apply CZ gates between all neighboring qubits. \Comment{Neighbors are defined by the edges in $G$, see Fig.~\ref{fig:2Dshallow-1Ddeep}}
\State Measure all qubits $(t, i)$ in the X-basis.
\State Assign the measurement result for qubit $(t, i)$ to the classical bit $y_{t, i}$.
\end{algorithmic}
\end{algorithm}

\begin{algorithm}[ht]
\caption{(1+1)D Deep Model}\label{alg:1p1d_deep_model}
\begin{algorithmic}[1] 
\Require  Input $\{x_{t,i}\}$.
\Statex $T$ total time.
\Statex $G$ is a 2D graph with node labels indexing qubits by $(t, i)$.
\Statex $\beta_{t,i}$ are the model parameters. 
\Statex The final circuit will be 1D Deep on $N$ qubits
\Ensure Bitstring measurements $\{y_{t,i}\}$.

\State Initialize all qubits to the state $\ket{0}$.
\State Initialize all bits $y_{t,i}$ randomly to 0 or 1 with equal probability.

\Comment{Operations at time $t=1$}
\State Apply $H^{1-x_{1,i}}$ to all qubits $i$.
\State Apply $R_Z(\beta_{1,i})$ to all qubits $i$.
\ForAll{edges $((1,i), (1,j))$ in $G$}
    \State Apply CZ gate between qubits $i$ and $j$.
\EndFor

\Comment{Main loop for time steps $t=2$ to $T$}
\For{$t \gets 2$ to $T$}
    \ForAll{edges $((t,i), (t,j))$ in $G$}
        \State Apply CZ gate between qubits $i$ and $j$.
    \EndFor
    
    \ForAll{qubits $i$}
        \State Apply $H^{1-x_{t,i}}$ to qubit $i$.
        \If{$x_{t,i} = 1$}
            \State Measure qubit $i$ in the Z-basis and assign the result to $y_{t-1, i}$.
        \EndIf
        \State Apply $Z^{y_{t-1,i}}$ followed by $R_Z(\beta_{t,i})$ to qubit $i$.
    \EndFor
\EndFor

\Comment{Final measurement at time $T$}
\ForAll{qubits $i$ at time $T$}
    \State Measure qubit $i$ in the X-basis and assign the result to $y_{T,i}$.
\EndFor

\end{algorithmic}
\end{algorithm}

\section{Experiments on generating quantum circuits to speed up simulation}
\label{app:speeding-simulation-experiment}
\subsection{Problem formulation and experimental setup}

We consider the problem where we are given a local Hamiltonian:
\begin{equation}
    H = \sum_{i} h_i,
\end{equation}
where each $h_i$ is a Hermitian operator acting on a constant number of qubits in an $n$-qubit system, and an evolution time $t$.
Our objective is to efficiently implement the $n$-qubit unitary:
\begin{equation}
    U = \exp(-i t H),
\end{equation}
on a quantum device. For most quantum simulation algorithms, $U$ would require a very deep quantum circuit, especially for long evolution times $t$. However, when the Hamiltonian possesses certain inherent structures, one may be able to implement $U$ with a significantly more efficient quantum circuit at some times or at all times.
Our goal is to develop an algorithm that automatically generates a low-depth circuit implementation of $U$ when such an implementation exists. 

The learning algorithm can operate in two distinct modes:
\begin{enumerate}
    \item \textbf{Classical simulation:} Perform Heisenberg evolution of local operators under $H$ via light-cone simulation to learn the local inversions of the unitary $U$.
    \item \textbf{Quantum simulation:} Utilize quantum algorithms (e.g., Trotter~\cite{childs2021theory} or QSP~\cite{low2017optimal}) to simulate Hamiltonian dynamics under $H$ and learn the local inversions of the unitary $U$.
\end{enumerate}
If a low-depth implementation of the unitary $U$ exists, one can learn local inversions that are low-depth. Then one can compile the local inversions into a low-depth circuit implementation of $U$.
Classical simulation (Mode 1) works efficiently when the Hamiltonian dynamics exhibits a slowly growing light cone.  Recall that even though observable light cone simulation is efficient here, taking samples from the final circuit in inference is still classically intractable in the most general case.  Using classical simulation has the significant advantage of not requiring the execution of any deep quantum circuits on quantum hardware during the learning process and generally much better scaling as a function of precision. Only after a low-depth circuit is generated do we need to use quantum hardware to implement $U$. In contrast, quantum simulation (Mode 2) remains viable even when the Hamiltonian dynamics produces a very large light cone during intermediate evolution times (though at the final time $t$, the existence of a low-depth implementation of $U$ implies that the light cone must eventually contract).

Mode 1 is particularly suitable for near-term and early fault-tolerant quantum computers where implementing deep quantum circuits is prohibitively challenging.
Mode 2 is designed for future large-scale fault-tolerant quantum computers where we can execute deep quantum circuits a limited number of times to gather necessary data, after which we can deploy the much more efficient circuit implementation to evolve under $U$.
In our current experiments, we focus on Mode 1, as we do not yet have access to a full-fledged fault-tolerant quantum computer.

The remainder of this appendix describes the details of our experiments, divided into numerical methodology and experimental implementation.
In the numerical methodology subsection, we elaborate on how the classical simulation and the learning of local inversions are performed for a family of Hamiltonians exhibiting slowly-growing light cones.
In the experimental implementation subsection, we detail how we generate a low-depth circuit for implementing the time-evolved operator $U = \exp(-i t H)$ for times $t \geq 20000$, how we map this circuit onto a Sycamore quantum processor, our data collection procedures, and subsequent data processing techniques.

\subsection{Numerical methodology}
\label{app:fitting_numerical_methodology}

From the problem formulation, the goal is to develop an algorithm that automatically generates a low-depth circuit implementation of $U = \exp(- i t H)$. In this experiment, we consider $H$ to be a known 1D local Hamiltonian with a given spatial geometry. To find a low-depth circuit implementation of $U$, we utilize a learning algorithm that focuses on finding local inversions of $k$ qubits under the action of the unitary $U$ on $n \geq k$ qubits. Depending on the level of knowledge of the structure of the unknown unitary, there are a number of strategies that can be employed to efficiently determine both the structure and content of the local inversion, as detailed in Ref.~\cite{huang2024learning}.

If we operate in Mode 2 (Quantum Simulation), we can obtain local measurements from each of the possible local bases of the $k$ input qubits ($\{ \ket{0}, \ket{1}, \ket{+}, \ket{-}, \ket{+y}, \ket{-y}\}$) and the known lightcones of those $k$ input qubits in a procedure analogous to classical shadows~\cite{huang2020predicting} and described in more detail in Appendix~\ref{sec:math-learning-QNC0}.  This method allows one to efficiently determine the local inversion for each of the groups of $k$ qubits using only samples from the device with a modest number of measurements.  However, as a sampling procedure it has a scaling of $\text{poly}(\epsilon^{-1})$ and for small systems where one has more knowledge of the unitary it can be advantageous to use a different procedure. 

In particular, we consider operating in Mode 1 (Classical simulation) where the light cone is small enough to be simulated efficiently for some time.  This allows us to study the training procedure with a lower overhead than sampling while keeping the loss landscape and phenomenological features of learning the same.  In particular, for the known full unitary $U$ and a trial unitary that we are optimizing over $U_t$, we directly implement the loss function
\begin{align}
    \sum_k||\text{Tr}_{n \setminus k}\; \mathcal{L} \left(U U_t^\dagger\right) - I_4||_2^2
\end{align}
where $\mathcal{L}(A)$ is the superoperator map, and the trace $\tr_{n \setminus k}$ is taken over all qubits except the input qubit $k$, such that the resulting map is $4$ dimensional in its superoperator matrix representation and the difference with the identity, $I_4$, can be computed.  Note that the actions of taking an operator to its super operator and partial trace require some nuance to correctly implement as detailed in Ref.~\cite{wood2015tensornetworksgraphicalcalculus}, and our notation hides some of that detail.  For efficiency, one can truncate qubit elements and gates of $U$ and $U_t$ not in the relevant lightcones which allows this procedure to be efficient as a function of the number of qubits.  

For each of the ansatz described above, we utilize random restarts with the circuit parameters initialized uniformly in the interval $[0, 1)$.  After initialization, we train the circuits with either ADAM or gradient descent with a backtracking line search (BGD).  We note that while ADAM is expected to perform well in either case, BGD is only expected to perform well with a deterministic or high precision cost function like we expect in Mode 1.  In such cases it can be much more efficient than a standard learning schedule, but it is not generally as robust to stochastic sampling techniques.

For ADAM we use a fairly standard set of learning parameters $\beta_1=0.85, \beta_2=0.9995, \epsilon=10^{-8}$ with a learning rate that we vary between $0.1$ and $0.5$.  For BGD we use $\alpha=0.5, \beta=0.8$ with $\alpha$ setting the threshold for decrease and $\beta$ setting the geometric decrease factor, and we similarly vary the step size / learning rate between $0.1$ and $0.5$.  For both methods we take 150 iterations for all experiments.

Due to the deterministic nature of the cost function for Mode 1, it is safe to take as the best point for each training run that which had the lowest overall cost function value.  For each target set of qubits we are optimizing over, we take the best optima from all restarts to define the local inversion for that set of qubits.  We then combine the best local inversion for each subset of qubits to find the full circuit.

In order to evaluate the difficulty of learning the circuit with a given grouping of qubits, for example with a single qubit, 4 qubits, or the whole circuit, we look at the probability of success upon restart of successful learning.  Due to the way in which we group qubits and can learn each qubit block independently, this definition is a little bit nuanced.  Specifically, we define success within a qubit block, as the maximum error for the channel for any individual qubit to be below a threshold $\delta = 10^{-2}$.  Then, because the cost is dominated by the number of restarts required to succeed on the hardest block of qubits, we take the overall success probability for a given grouping to be the minimum success probability across the blocks.  Together this defines the numerical success probability for learning the circuit as
\begin{align}
\label{eq:fitting_probability}
    \min_{B_i} \text{Pr} \left[ \max_{k \in B_i} ||\text{Tr}_{n \setminus k}\; \mathcal{L} \left(U U_{it}^\dagger\right) - I_4||_2^2 \leq \delta \right]
\end{align}
where $B_i$ are the qubit blocks in the grouping and $U_{it}$ is the trial local inversion specialized to qubit block $B_i$. 

To examine the case of compressing physical evolution circuits with hidden structure, we consider physical evolutions of Hamiltonians of separable Hamiltonians hidden by constant depth circuits.  More specifically, consider the Hamiltonian
\begin{align}
    H = U^\dagger \left( \sum_j h_j Z_j \right) U
\end{align}
where $h_j$ is draw uniformly at random from $[-1,1]$, $Z_j$ is the Pauli Z operator on the $j$'th qubit, and $U$ is a constant depth random unitary.  In this case, we construct $U$ as a 1D random unitary, where in the first layer random phased XZ gates are added to each pair of adjacent qubits starting from qubit $1$, followed by a CZ gate between the qubits, followed by another set of phased XZ gates on each of the qubits.  This is then repeated starting at qubit 2, such that a 1D brickwork is created. The PhasedXZ gate is defined as $\text{Ph}_{XZ} = Z^{-a} X^x Z^a Z^z$ in time-order left to right, where $X, Z$ are Pauli operators, and $x, z, a$ are fixed precision numbers taken in the interval $[0, 4]$.  In our experiments, $x$, $z$, $a$ are initialized uniformly at random inside the interval $[0,4]$.  As an example on 4 qubits, with each of the gates randomized independently, 

{\footnotesize
\begin{align}
\Qcircuit @C=1em @R=1.5em {
    & \gate{\text{Ph}(x^{1}, z^{1}, a^{1})} & \ctrl{1} & \gate{\text{Ph}(x^{2}, z^{2}, a^{2})} & \qw       & \gate{\text{Ph}(x^{3}, z^{3}, a^{3})} & \qw \\
    & \gate{\text{Ph}(x^{4}, z^{4}, a^{4})} & \ctrl{-1} & \gate{\text{Ph}(x^{5}, z^{5}, a^{5})} & \ctrl{1}& \gate{\text{Ph}(x^{6}, z^{6}, a^{6})}&\qw \\
    & \gate{\text{Ph}(x^{7}, z^{7}, a^{7})} & \ctrl{1}  & \gate{\text{Ph}(x^{8}, z^{8}, a^{8})} &\ctrl{-1} & \gate{\text{Ph}(x^{9}, z^{9}, a^{9})}&\qw  \\
    & \gate{\text{Ph}(x^{10}, z^{10}, a^{10})} & \ctrl{-1} & \gate{\text{Ph}(x^{11}, z^{11}, a^{11})} &\qw       &  \gate{\text{Ph}(x^{12}, z^{12}, a^{12})} &  \qw \\
}
\end{align}}
\;

Once the scrambling unitary $U$ has been fixed, we use it to express the time evolution circuits for the physical simulation.  If one were to naively expand $U H U^\dagger$ into Pauli terms, without recognizing the overall structure, implementing this naively with CNOT ladders under Trotterization, parallelizing over qubit blocks, and performing a single term at a time is untenable at reasonable levels of noise today on almost any processor.  

However, for this particular Hamiltonian, the time-evolution circuit obviously has a simple form for any time $t \in \mathbb{R}$, namely the unitary $U$, $R_z(-2 h_j t)$ applied to all qubits, and the inverse $U^\dagger$, where $R_z(\theta)$ is a parameterized rotation about $Z$ by $\theta$ equivalent to $\exp(-i Z \theta / 2)$.  This structure can be learned even without detailed knowledge of its presence via the sewing technique. We use this circuit to simulate the time-dynamics across a range of time steps at both short and long times.  Namely, $t \in 3 \pi / 40 * k + 0.001 \; \forall \; k \in [0, 40] \cup [10^6, 10^6+40]$, and we read out $\langle Z_j \rangle$ for each qubit $j$. 

For learning these circuits, we use Mode 1 of Classical Simulation as described above on blocks of 4 qubits, further described in the experimental implementation section.  We use an ansatz that trains on single qubit rotations around $x$, $y$, and $z$, with a fixed 2 qubit gate structure that respects the symmetry of performing both a $U$ and $U^\dagger$.  In particular, we use 

{\footnotesize
\begin{align}
\Qcircuit @C=1em @R=1.5em {
     & \gate{\text{R}(x^{1},y^{1},z^{1})} & \ctrl{1} & \gate{\text{R}(x^{2},y^{2},z^{2})} & \qw       & \gate{\text{R}(x^{3},y^{3},z^{3})} & \qw       & \gate{\text{R}(x^{4},y^{4},z^{4})} & \qw & \ctrl{1} & \qw       & \gate{\text{R}(x^{5},y^{5},z^{5})}\\
     & \gate{\text{R}(x^{6},y^{6},z^{6})} & \ctrl{-1} & \gate{\text{R}(x^{7},y^{7},z^{7})} & \ctrl{1} & \gate{\text{R}(x^{8},y^{8},z^{8})} & \ctrl{1}& \gate{\text{R}(x^{9},y^{9},z^{9})} &\qw & \ctrl{-1} & \qw       & \gate{\text{R}(x^{10},y^{10},z^{10})}\\
     & \gate{\text{R}(x^{11},y^{11},z^{11})} & \ctrl{1}  & \gate{\text{R}(x^{12},y^{12},z^{12})} &\ctrl{-1}  & \gate{\text{R}(x^{13},y^{13},z^{13})} &\ctrl{-1} & \gate{\text{R}(x^{14},y^{14},z^{14})}&\qw & \ctrl{1} & \qw       & \gate{\text{R}(x^{15},y^{15},z^{15})}\\
     & \gate{\text{R}(x^{16},y^{16},z^{16})} & \ctrl{-1} & \gate{\text{R}(x^{17},y^{17},z^{17})} &\qw       &  \gate{\text{R}(x^{18},y^{18},z^{18})} &\qw       &  \gate{\text{R}(x^{19},y^{19},z^{19})} &  \qw & \ctrl{-1} & \qw       & \gate{\text{R}(x^{20},y^{20},z^{20})}\\
}
\end{align}}
\;

for a 4-qubit circuit, where $\text{R}(x,y,z) = \text{R}_X(x) \text{R}_Y(y) \text{R}_Z(z)$, $\text{R}_P(\theta) = \exp(-i P \theta / 2)$, and every $x,y,z$ in the ansatz are allowed to vary independently, and are initialized randomly uniformly in $[0,1)$.  We note that several of the internal rotations could be combined into a single rotation, but that does not nessecarily preserve the overall training landscape, and this is the precise form used in our learning experiments here.

We save optimizer state and metrics along the way and save the results for all trials run.  For the combination of the compressing physical dynamics and SWAP circuits, we run approximately 24,000 trial runs with $150$ optimizer steps each across a variety of different ansatz and starting conditions as described above. These simulation runs were carried out using a combination of Numpy~\cite{harris2020array}, Cirq~\cite{CirqDevelopers_2025} and TFQ~\cite{broughton2020tensorflow} on GCP with a total FLOP budget of around $10^{17}$ (O(1) ExaFLOPs) in total compute - not FLOPs per second).

\subsection{Experimental implementation}
For this experiment, we are testing the more general learning approach described in Appendix~\ref{sec:math-learning-QNC0}, that requires $n$ ancilla qubits to utilize the sewing method.  Both the scrambling unitary that is used to hide the Hamiltonian and the ansatz circuit are defined in 1D, and this structure is mimicked in the experimental layout.  In particular, within the 2D grid of Sycamore, these experiments were laid out on a 1D line on the chip with available ancilla for each qubit on the line on either side (the line is connected with each neighbor, but the ancilla are not). We took 200000 samples for each problem instance circuit and computed the qubit mean values across qubits within the circuit for different values of time.  For these experiments we utilized up to 40 physical qubits, with 20 used for the system and 20 used for the ancilla register in sewing.

An important consideration in experimental implementation using the sewing trick is to minimize the circuit depth if possible.  This both lowers exposure to decoherence errors and improves the time to execute a given circuit for any given implementation.  The learning technique as described in Appendix~\ref{sec:math-learning-QNC0} discusses only the case of a single qubit local inversion where learning is the easiest, but one may learn local inversions for blocks of qubits at a time, and so long as the block size is a constant the learning remains efficient (see next section as well).  This is reflected in the cost function for training in Mode 1 that is summed over all the qubits in a learning block.  While the learning becomes more challenging with a larger block size, there are sometimes experimental benefits in the inference step to using larger blocks, such as shorter depth circuits or opportunities for parallelism.

In this experiment, we choose an experimental block size of $k=4$.  To see why this is an advantageous choice, consider the layout mapped in Fig.~\ref{fig:mbl_results}b that demonstrates this scheme allows execution of the circuit in a depth of 2 layers.  For a single qubit in 1D, the scrambling circuit we are considering can spread the light cone $3$ qubits up and $3$ qubits down from the block under consideration, which represents the blocking light cone of that local inversion.  If one does a single qubit local inversion, the light cone blocks up to 7 qubits from having actions performed on them for other local inversions, while implementing the action for only a single qubit.  Hence, while a constant depth implementation is possible, the depth is likely to need to be at least 7, which makes experimental implementation in conjunction with the sewing trick unwieldy.  In contrast, if the qubit block size is 4 and judiciously chosen, then only 2 qubits on either side of the block are impeded, and 4 qubits are implemented at a time.  This means the overall circuit can be implemented in depth 2 repetitions of the local inversions, a substantial improvement in implementation efficiency.

For the chip being used, the CZ is a native hardware gate that is individually calibrated.  The CZ connectivity is displayed by the couplers between qubits on the superconducting qubit grid depicted in Fig.~\ref{fig:mbl_results}a.  For the single qubit gates, we compile into PhasedXZ gates, and the single qubit error rate has a typical value of $10^{-3}$.  

The raw results are quite good due to the quality of the device, but here we describe a simple error mitigation technique that further improved the results with respect to the exact solution.  This error mitigation technique is only used in the inference step, as the learning is done classically in Mode 1 as described.  When the ansatz circuit above sets all the trainable parameters to $0$, it is logically the identity circuit.  If we perform no circuit cancellation at that point, and run it on the device, it experiences a structurally very similar noise pattern to the true circuit, but with exactly known reference values.  Denote this faulty identity channel as $\mathcal{E}$.  For a given observable on qubit $i$, $O_i$ with measurement outcomes $\{0, 1\}$, we define the zero and one reference values as the values of those observables measured on the all $0$ state $\ket{000...}$. 

That is, let
\begin{align}
    \rho_0 &= \mathcal{E} \left[ \ket{000...} \bra{000...} \right] \\
    O_i^0 &= \tr[\rho_0 O_i] 
\end{align}
with these reference values, we define the error mitigated observables as
\begin{align}
    M(O_i) &= \frac{\langle O_i \rangle - O_i^0}{1 - 2 O_i^0}
\end{align}
where $\langle O_i \rangle$ is the empirically measured expected value of $O_i$ on the ansatz circuit with the desired parameters.  For Pauli operators $P_i$, we relate them to an $O_i$ via $O_i = 1 - (P_i + 1)/2$.  Empirically we find that this mitigation technique improves the results, though we note the raw values are quite good as well for these experiments as can be seen in Figs.~\ref{fig:mbl_results}c and ~\ref{fig:mbl_results}d.

\section{Numerical experiments on learning SWAP circuits}
\label{app:num-exp-swap-circuit-learning}
In order to best isolate the consequences for efficiency of training, we also run numerical experiments on extremely simple circuits for which we have analytical results on the landscapes and detailed analysis in Appendix~\ref{sec:exp-many-local-minima-param-shallow}.  More specifically, we design experiments to study the probability of success as a function of the system size and size of the block under consideration in terms of the proliferation of local minima.  This has a direct consequence in terms of the probability of success of restarted local descent methods as described in Appendix~\ref{app:fitting_numerical_methodology}, and the probability of success described in Eq.~\ref{eq:fitting_probability} is evaluated and plotted in Fig.~\ref{app:fitting_numerical_methodology}.  Here, we focus on an extremely simple example to control for excess complications, namely a classical circuit that SWAPs every 4th qubit.  More explicitly one may write this circuit as
\begin{align}
\Qcircuit @C=1em @R=1.5em {
    & \qswap & \qw \\
    & \qw \qwx & \qw \\
    & \qw \qwx & \qw \\
    & \qswap \qwx & \qw\\
}
\end{align}
repeated over all disjoint groups of $4$ qubits, arranged in 1D.  To further emphasize the essential challenges, we also use a stilted ansatz that closely mimics the form of the target circuit with only slightly different connectivity.  Explicitly we use

\begin{align}
\Qcircuit @C=1em @R=1.5em {
     & \multigate{1}{\text{SWAP}^{p^{0,0}}} & \qw & \qw & \multigate{1}{\text{SWAP}^{p^{0,1}}} & \qw & \qw\\
     & \ghost{\text{SWAP}^{p^{0,0}}} & \multigate{1}{\text{SWAP}^{p^{1,0}}} & \qw & \ghost{\text{SWAP}^{p^{0,0}}} & \multigate{1}{\text{SWAP}^{p^{1,1}}} & \qw \\
     & \multigate{1}{\text{SWAP}^{p^{2,0}}} & \ghost{\text{SWAP}^{p^{0,0}}} & \qw & \multigate{1}{\text{SWAP}^{p^{2,1}}} & \ghost{\text{SWAP}^{p^{0,0}}} & \qw \\
     & \ghost{\text{SWAP}^{p^{0,0}}} & \qw & \qw & \ghost{\text{SWAP}^{p^{0,0}}} & \qw & \qw \\
}
\end{align}
\;

where each $p^{i,j}$ is allowed to be independently initialized and optimized, the gate $\text{SWAP}^p$ is the SWAP gate raised to the $p$'th power, and this is two layers of of nearest neighbor SWAPs in 1D, clearly capable of expressing the SWAP-4 target circuit.   Perhaps surprisingly, without the use of the sewing trick, even with knowledge of the circuit structure, it is incredibly hard to learn due to the proliferation of local minima in the training landscape.  As this circuit is classical, this is a consequence of circuit architecture and information readout rather than genuinely quantum effects, however divide and conquer learning combined with the sewing trick is able to overcome this limitation regardless.  The SWAP gate to the $p$'th power has the more specific matrix representation in the computational basis
\begin{align}
    \text{SWAP}^p = \left( 
    \begin{array}{cccc}
    1 & 0 & 0 & 0 \\
    0 & e^{i \pi p / 2} \text{cos}(\pi p / 2) & -i e^{i \pi p / 2} \text{sin}(\pi p / 2) & 0 \\
    0 & -i e^{i \pi p / 2} \text{sin}(\pi p / 2) & e^{i \pi p / 2} \text{cos}(\pi p / 2) & 0 \\
    0 & 0 & 0 & 1
    \end{array}
    \right)
\end{align} 

However, the sewing technique allows one to fix the number of qubits under consideration in the target, and scalable stop the proliferation of local minima, allowing simple restart and local minimization to be an efficient strategy.  For any fixed number of target qubits $k$.  In this work, we compare $k=1$ and $k=4$ to the cost of learning the whole circuit with this ansatz and empirically support the exponential advantage at practical sizes of interest.

We use the same methods as outlined in Appendix~\ref{app:fitting_numerical_methodology} to fit the above ansatz circuits and determine their probability of success.  For this set of experiments, we restrict ourselves to numerical simulation.

\section{Mathematical framework for learning shallow quantum circuits} \label{sec:math-learning-QNC0}

In this appendix, we present the mathematical framework for learning and training shallow quantum circuits from a classical dataset.

\subsection{Dataset} \label{sec:dataset-learn-shallow-qc}

Our approach uses a classical dataset obtained from randomized measurements on the unknown quantum circuit. The dataset construction procedure is as follows:

\begin{definition}[Randomized measurement dataset]
Given an unknown $n$-qubit unitary $U$, a randomized measurement dataset of size $N$ is constructed as:
\begin{equation}
\mathcal{T}_U(N) = \left\{ \ket{\psi_\ell} = \bigotimes_{i=1}^n \ket{\psi_{\ell, i}}, \ket{\phi_\ell} = \bigotimes_{i=1}^n \ket{\phi_{\ell, i}} \right\}_{\ell=1}^N,
\end{equation}
where each sample is obtained through the following procedure:
\begin{enumerate}
    \item Sample an input product state $\ket{\psi_\ell} = \bigotimes_{i=1}^n \ket{\psi_{\ell, i}}$, with each $\ket{\psi_{\ell, i}}$ being a uniformly random single-qubit stabilizer state from $\{\ket{0}, \ket{1}, \ket{+}, \ket{-}, \ket{y+}, \ket{y-}\}$.
    \item Apply the unknown unitary $U$ to $\ket{\psi_\ell}$.
    \item Measure each qubit of $U \ket{\psi_\ell}$ in a random Pauli basis chosen from $\{X, Y, Z\}$. The measurement collapses the state to a product state $\ket{\phi_\ell} = \bigotimes_{i=1}^n \ket{\phi_{\ell, i}}$.
\end{enumerate}
This dataset can be represented efficiently on a classical computer using $\mathcal{O}(Nn)$ bits.
\end{definition}

\subsection{Algorithm 1: Known structure with direct Heisenberg sewing} \label{sec:alg1-known-direct}

We begin by presenting an algorithm for learning a shallow quantum neural network when both the geometric connectivity and circuit architecture are known. Our goal is to learn an unknown $n$-qubit unitary $U$ that belongs to a specified parameterized circuit family.

\subsubsection{Circuit architecture and lightcone structure}

Consider a parameterized quantum circuit architecture $\mathcal{A}$. While the architecture is known, the specific gate parameters, and hence the particular unitary $U \in \mathcal{A}$ we wish to learn, are unknown. The key insight is that the geometric structure of $\mathcal{A}$ constrains the support of Heisenberg-evolved Pauli observables.

\begin{definition}[Backward lightcone for circuit architecture]
For a constant-depth circuit architecture $\mathcal{A}$ and qubit $i \in \{1, \ldots, n\}$, the backward lightcone $C_i$ is the maximal set of qubits that can influence the $i$-th qubit output across all possible unitaries in the architecture:
\begin{equation}
C_i = \{j \in \{1, \ldots, n\} : \exists P \in \{X, Y, Z\}, U \in \mathcal{A}, \text{ such that } [U^\dagger P_i U, P_j] \neq 0\}.
\end{equation}
Since $\mathcal{A}$ consists of constant-depth circuits, we have $|C_i| = \mathcal{O}(1)$ for all $i$.
\end{definition}

The lightcones $\{C_i\}_{i=1}^n$ are determined purely by the circuit architecture and can be computed efficiently from the geometric connectivity pattern. For any specific unitary $U \in \mathcal{A}$, the Heisenberg-evolved observables $U^\dagger P_i U$ have support contained within $C_i$.

\subsubsection{Pauli coefficient parameterization}

We parameterize the Heisenberg-evolved observables directly through their Pauli decompositions:

\begin{definition}[Pauli coefficient parameterization]
For each qubit $i \in \{1, \ldots, n\}$ and single-qubit Pauli observable $P \in \{X, Y, Z\}$, we parameterize:
\begin{equation}
U^\dagger P_i U = \sum_{Q \in \mathrm{Pauli}(C_i)} \alpha_{P_i, Q} Q,
\end{equation}
where 
\begin{equation}
\mathrm{Pauli}(C_i) = \{Q \in \{I, X, Y, Z\}^{\otimes n} : Q_j = I \text{ for all } j \notin C_i\}
\end{equation}
is the set of Pauli operators supported on the lightcone $C_i$. We denote the entire parameter vector as $\vec{\alpha}$. The total number of parameters is $3n \cdot \max_i 4^{|C_i|} = \mathcal{O}(n),$ since each $|C_i| = \mathcal{O}(1)$.
\end{definition}

\subsubsection{Loss function}

Given the randomized measurement dataset $\mathcal{T}_U(N)$, we define the loss function:
\begin{equation}
\mathcal{L}(\vec{\alpha}) = \frac{1}{N} \sum_{\ell=1}^N \sum_{i=1}^n \sum_{P \in \{X, Y, Z\}} \sum_{Q \in \mathrm{Pauli}(C_i)} \left| \alpha_{P_i, Q} - 3^{|Q|+1} \bra{\phi_{\ell, i}} P \ket{\phi_{\ell, i}} \cdot \bra{\psi_\ell} Q \ket{\psi_\ell} \right|^2.
\end{equation}
This is a standard least-squares optimization in the parameter space. As we will demonstrate rigorously in Section~\ref{sec:all-local-minima-are-global}, this optimization problem is strongly convex and possesses a unique global minimum $\vec{\alpha}^*$, ensuring that any local optimization algorithm will converge to the optimal solution.
Furthermore, it is not hard to show that $\sum_{Q \in \mathrm{Pauli}(C_i)} \alpha^*_{P_i, Q} Q \approx U^\dagger P_i U$ up to an $\epsilon$ error under $\norm{\cdot}_\infty$ using a dataset $\mathcal{T}_U(N)$ of size $N = \mathcal{O}(1/\epsilon^2)$.

\subsubsection{Direct Heisenberg sewing}

\begin{definition}[Direct Heisenberg sewing]
Given learned coefficients $\vec{\alpha}^*$, we construct the sewed $2n$-qubit unitary through direct projection of Heisenberg operators. For each qubit $i$, define the sewing operator:
\begin{equation}
\mathcal{S}_i(\vec{\alpha}^*) = \frac{1}{2} I \otimes I + \frac{1}{2} \sum_{P \in \{X, Y, Z\}} \sum_{Q \in \mathrm{Pauli}(C_i)} P_i \otimes \alpha^*_{P_i, Q} Q.
\end{equation}
Since each $\mathcal{S}_i(\vec{\alpha}^*)$ may not be unitary, we apply unitary projection $\mathrm{Proj}_{\mathsf{U}}(\cdot)$, where for a matrix $M$:
\begin{equation}
\mathrm{Proj}_{\mathsf{U}}(M) = \arg\min_{U: U^\dagger U = I} \|U - M\|_F^2 = VW^\dagger
\end{equation}
if $M = V\Sigma W^\dagger$ is the singular value decomposition. The complete sewed unitary is:
\begin{equation}
U_{\mathrm{sew}}^{\mathrm{Heis}}(\vec{\alpha}^*) := S \left[\prod_{i=1}^n \mathrm{Proj}_{\mathsf{U}} \left( \mathcal{S}_i(\vec{\alpha}^*) \right) \right],
\end{equation}
where $S$ denotes the global swap operator between the two $n$-qubit registers. The learned $n$-qubit channel is:
\begin{equation}
\hat{\mathcal{E}}^{\mathrm{Heis}}(\rho) := \mathrm{Tr}_{\leq n}\left(U_{\mathrm{sew}}^{\mathrm{Heis}}(\vec{\alpha}^*) \left(\ketbra{0^n}{0^n} \otimes \rho\right) \left(U_{\mathrm{sew}}^{\mathrm{Heis}}(\vec{\alpha}^*)\right)^\dagger\right).
\end{equation}
\end{definition}

We can obtain an upper bound on the error between the learned channel $\hat{\mathcal{E}}^{\mathrm{Heis}}$ and the ideal unitary channel $\mathcal{U}(\cdot) = U (\cdot) U^\dagger$ as follows. The proof of Theorem~\ref{thm:direct-sewing-accuracy} is deferred to Section~\ref{subsec:technical-proofs}.

\begin{theorem}[Error bound on direct Heisenberg sewing] \label{thm:direct-sewing-accuracy}
Given an unknown $n$-qubit unitary $U$ and learned coefficients $\vec{\alpha}^*$ such that for each $i \in \{1,\ldots,n\}$, $\sum_{P \in \{X, Y, Z\}} \left\|\sum_{Q \in \mathrm{Pauli}(C_i)} \alpha^*_{P_i, Q} Q - U^\dagger P_i U\right\|_\infty \leq \varepsilon_i,$ then:
\begin{equation}
\|U_{\mathrm{sew}}^{\mathrm{Heis}}(\vec{\alpha}^*) - U^\dagger \otimes U\|_\infty \leq \sum_{i=1}^n \varepsilon_i.
\end{equation}
Hence, we have $\norm{\hat{\mathcal{E}}^{\mathrm{Heis}} - \mathcal{U}}_\diamond \leq 2 \sum_{i=1}^n \varepsilon_i$.
\end{theorem}

\subsection{Algorithm 2: Known structure with local inversion sewing} \label{sec:alg2-known-inversion}

The second algorithm learns the same Heisenberg-evolved observables but replaces the unitary group projection with local inversion training for each qubit to construct the final circuit.

\subsubsection{Local inversion construction}

We begin with the general concept of local inversion before specializing to our learned coefficients.

\begin{definition}[$\varepsilon$-approximate local inversion]
Given $n$-qubit unitaries $U$ and $V_i$, we say $V_i$ is an $\varepsilon$-approximate local inversion of $U$ on the $i$-th qubit if:
\begin{equation}
\sum_{P \in \{X, Y, Z\}} \| V_i^\dagger U^\dagger P_i U V_i - P_i\|_\infty \leq \varepsilon
\end{equation}
where $P_i$ denotes the Pauli operator $P$ acting on the $i$-th qubit and identity on all other qubits.
\end{definition}

Given the learned coefficients $\vec{\alpha}^*$ from the convex optimization, we construct local inversion unitaries:

\begin{definition}[Local inversion from coefficients]
For each qubit $i$, define the target operator:
\begin{equation}
\hat{O}_{i,P} = \sum_{Q \in \mathrm{Pauli}(C_i)} \alpha^*_{P_i, Q} Q.
\end{equation}
The local inversion $V_i$ is optimized to satisfy:
\begin{equation}
\sum_{P \in \{X, Y, Z\}} \| V_i^\dagger \hat{O}_{i,P} V_i - P_i\|_\infty \leq \varepsilon_i.
\end{equation}
\end{definition}

Since $\hat{O}_{i,P}$ has support on a constant-size lightcone $C_i$, the local inversion $V_i$ can be constrained to act only on a constant number of qubits and have constant depth. Since $V_i$ acts on $\mathcal{O}(1)$ qubits, the optimization can be performed using various methods, such as brute-force search, gradient descent with restart, L-BFGS, Newton method, Bayesian optimization. All of these approaches have constant runtime for any fixed $\varepsilon_i$ due to the constant dimensionality of the optimization space. Furthermore, the optimization landscape will only contain a constant number of local minima due to the constant dimensionality; see Section~\ref{subsubsection:const-local-min-local-inv}.

\subsubsection{Local inversion sewing}

\begin{definition}[Sewing local inversions]
Given $n$-qubit unitaries $V_1, \ldots, V_n$, define the sewed $2n$-qubit unitary as:
\begin{equation}
U_{\mathrm{sew}}^{\mathrm{inv}}(V_1, \ldots, V_n) := S \left[\prod_{i=1}^n \left(V_i \otimes \Id \right) S_{i} \left(V_i^\dagger \otimes \Id\right) \right],
\end{equation}
where $S_{i}$ is the swap operator for the $i$-th qubit between the two $n$-qubit systems, and $S$ is the swap operator for all $n$ qubits.
The learned $n$-qubit channel is then defined as:
\begin{equation}
\hat{\mathcal{E}}^{\mathrm{inv}}(\rho) := \mathrm{Tr}_{>n}(U_{\mathrm{sew}}^{\mathrm{inv}}(\rho \otimes \ketbra{0^n}{0^n})(U_{\mathrm{sew}}^{\mathrm{inv}})^\dagger).
\end{equation}
\end{definition}

We can obtain an upper bound on the error between the learned channel $\hat{\mathcal{E}}^{\mathrm{inv}}$ and the ideal unitary channel $\mathcal{U}(\cdot) = U (\cdot) U^\dagger$ as follows. The proof of Theorem~\ref{thm:sewing-accuracy} is deferred to Section~\ref{subsec:technical-proofs}.

\begin{theorem}[Error bound on sewing local inversions] \label{thm:sewing-accuracy}
Given an unknown $n$-qubit unitary $U$, if $V_i$ is an $\varepsilon_i$-approximate local inversion of $U$ on the $i$-th qubit for each $i$, then:
\begin{equation}
\|U_{\mathrm{sew}}(V_1, \ldots, V_n) - U \otimes U^\dagger\|_\infty \leq \frac{1}{2} \sum_{i=1}^n \varepsilon_i.
\end{equation}
Hence, we have $\norm{\hat{\mathcal{E}}^{\mathrm{inv}} - \mathcal{U}}_\diamond \leq \sum_{i=1}^n \varepsilon_i$.
\end{theorem}

\subsection{Algorithm 3: Unknown structure with local inversion sewing} \label{sec:alg3-unknown}

When the circuit architecture is unknown, we cannot precompute the lightcones $C_i$. Instead, we must learn the support structure of the Heisenberg-evolved observables. This can be achieved using the algorithm from Lemma~\ref{lem:learn-H-evolved-ops} below, then proceed with local inversion sewing as in Algorithm 2. The algorithm proceeds as follows:
\begin{enumerate}
    \item Use Lemma~\ref{lem:learn-H-evolved-ops} to learn the Heisenberg-evolved Pauli observables $\hat{O}_{i,P} \approx U^\dagger P_i U$.
    \item For each qubit $i$, optimize a local inversion $V_i$ such that $V_i^\dagger \hat{O}_{i,P} V_i \approx P_i$ as in Algorithm 2.
    \item Apply local inversion sewing as in Algorithm 2 to obtain the learned channel $\hat{\mathcal{E}}(\rho)$.
\end{enumerate}
In the following, we state the main lemma for learning the Heisenberg-evolved observables from the dataset $\mathcal{T}_U(N)$ with an arbitrary circuit architecture for the unknown shallow quantum neural network.

\begin{lemma}[Learning Heisenberg-evolved Pauli observables] \label{lem:learn-H-evolved-ops}
Given an unknown $n$-qubit unitary~$U$ implemented by a constant-depth circuit, a failure probability $\delta$, an error parameter $\varepsilon$, and a randomized measurement dataset $\mathcal{T}_U(N)$ of size
\begin{equation}
N = \mathcal{O}\left( \frac{2^{\mathcal{O}(k)} \log(n / \delta)}{\varepsilon^2} \right),
\end{equation}
where $k = \mathcal{O}(1)$ is the maximum size of the lightcone in the circuit, we can learn approximations $\hat{O}_{i,P}$ to each Heisenberg-evolved Pauli observable $U^\dagger P_i U$ for $i \in \{1,\ldots,n\}$ and $P \in \{X,Y,Z\}$ using the following algorithm:
\begin{enumerate}    
    \item For all $Q \in \{I, X, Y, Z\}^{\otimes n}$ with $|Q| \leq k$, i.e., $Q$ acts as non-identity on at most $k$ qubits, compute the coefficient:
    \begin{equation} \label{eq:unbiased-est-Pauli-alpha}
    \hat{\alpha}_Q = \frac{3^{|Q|}}{N} \sum_{\ell=1}^N 3 \bra{\phi_{\ell, i}} P \ket{\phi_{\ell, i}} \cdot \bra{\psi_\ell} Q \ket{\psi_\ell},
    \end{equation}
    then apply a threshold function to the coefficient:
    \begin{equation}
    \hat{\beta}_Q = 
    \begin{cases}
        \hat{\alpha}_Q, & |\hat{\alpha}_Q| \geq 0.5 \varepsilon / (2\sqrt{2})^k \\
        0, & |\hat{\alpha}_Q| < 0.5 \varepsilon / (2\sqrt{2})^k
    \end{cases}.
    \end{equation}
    \item Construct the approximation as:
    \begin{equation}
    \hat{O}_{i,P} = \sum_{Q:|Q| \leq k} \hat{\beta}_Q Q.
    \end{equation}
\end{enumerate}
This algorithm guarantees that with probability at least $1-\delta$:
\begin{equation}
\|{\hat{O}_{i,P} - U^\dagger P_i U}\|_\infty \leq \varepsilon \quad \text{and} \quad \mathrm{supp}(\hat{O}_{i,P}) \subseteq \mathrm{supp}(U^\dagger P_i U)
\end{equation}
for all $i \in \{1,\ldots,n\}$ and $P \in \{X,Y,Z\}$, where $\mathrm{supp}(A)$ is the set of qubits that $A$ acts nontrivially on.
\end{lemma}

\subsection{Main learning theorem}

\begin{theorem}[Learning shallow quantum circuits] \label{thm:main-learn-QNC0}
Given a failure probability $\delta$, an approximation error $\varepsilon$, and an unknown $n$-qubit unitary $U$ generated by a constant-depth circuit, with a randomized measurement dataset $\mathcal{T}_U(N)$ of size
\begin{equation}
N = \mathcal{O}\left( \frac{n^2 \log(n / \delta)}{\varepsilon^2} \right),
\end{equation}
we can learn an $n$-qubit quantum channel $\hat{\mathcal{E}}$ that can be implemented by a constant-depth quantum circuit over $2n$ qubits, such that
\begin{equation}
\|\hat{\mathcal{E}} - \mathcal{U}\|_\diamond \leq \varepsilon,
\end{equation}
with probability at least $1 - \delta$.
\end{theorem}
\begin{proof}
The proof proceeds by applying our learning algorithm and analyzing its performance guarantees. First, we apply Lemma~\ref{lem:learn-H-evolved-ops} with error parameter $\varepsilon' = \frac{\varepsilon}{6n}$ to learn approximations $\hat{O}_{i,P}$ to the Heisenberg-evolved Pauli observables $U^\dagger P_i U$ for all $i \in \{1,\ldots,n\}$ and $P \in \{X,Y,Z\}$. With probability at least $1-\frac{\delta}{2}$, we have:
\begin{equation}
\|\hat{O}_{i,P} - U^\dagger P_i U\|_\infty \leq \frac{\varepsilon}{6n}.
\end{equation}
Second, for each qubit $i$, we train a local inversion unitary $V_i$ to minimize:
\begin{equation}
\sum_{P \in \{X,Y,Z\}} \|U^\dagger V_i^\dagger P_i V_i U - P_i\|_\infty.
\end{equation}
Since $U$ is a constant-depth circuit, the Heisenberg-evolved observables $U^\dagger P_i U$ have bounded support of size $\mathcal{O}(1)$. This means the local inversion unitaries $V_i$ can also be constrained to be $\mathcal{O}(1)$ depth and act only on $\mathcal{O}(1)$ qubits. By optimizing over this constant-sized space, we can find local inversions $V_i$ such that
\begin{equation}
\sum_{P \in \{X,Y,Z\}} \|U^\dagger V_i^\dagger P_i V_i U - P_i\|_\infty \leq \frac{\varepsilon}{n}.
\end{equation}
Third, we construct the sewed unitary $U_{\mathrm{sew}}(V_1, \ldots, V_n)$ according to Definition 2. By Theorem~\ref{thm:sewing-accuracy},
\begin{equation}
\|U_{\mathrm{sew}}(V_1, \ldots, V_n) - U^\dagger \otimes U\|_\infty \leq \frac{1}{2}\sum_{i=1}^n \frac{\varepsilon}{n} = \frac{\varepsilon}{2}.
\end{equation}
Finally, we define our learned channel as:
\begin{equation}
\hat{\mathcal{E}}(\rho) := \mathrm{Tr}_{>n}(U_{\mathrm{sew}}(\ketbra{0^n}{0^n} \otimes \rho)U_{\mathrm{sew}}^\dagger).
\end{equation}
The diamond distance between our learned channel and the target channel is:
\begin{align}
\|\hat{\mathcal{E}} - \mathcal{U}\|_\diamond &= \|\mathrm{Tr}_{>n}(U_{\mathrm{sew}}(\ketbra{0^n}{0^n} \otimes \cdot)U_{\mathrm{sew}}^\dagger) - U(\cdot)U^\dagger\|_\diamond\\
&= \|\mathrm{Tr}_{>n}(U_{\mathrm{sew}}(\ketbra{0^n}{0^n} \otimes \cdot)U_{\mathrm{sew}}^\dagger) - \mathrm{Tr}_{>n}((U^\dagger \otimes U)(\ketbra{0^n}{0^n} \otimes \cdot)(U^\dagger \otimes U)^\dagger)\|_\diamond\\
&\leq 2\|U_{\mathrm{sew}} - U^\dagger \otimes U\|_\infty \leq 2 \cdot \frac{\varepsilon}{2} = \varepsilon,
\end{align}
where we used the fact that for channels derived from unitaries, the diamond norm is bounded by twice the spectral norm of the difference between the unitaries. The implementation of $\hat{\mathcal{E}}$ as a constant-depth quantum circuit over $2n$ qubits follows from the fact that each $V_i$ has constant depth and acts on a constant number of qubits. If we implement the sewing construction in a chosen order, we can preserve the constant-depth property of the overall circuit. Therefore, with probability at least $1-\delta$, we can learn an $n$-qubit quantum channel $\hat{\mathcal{E}}$ that approximates $\mathcal{U}$ with diamond distance at most $\varepsilon$ using $N = \mathcal{O}(n^2 \log(n/\delta)/\varepsilon^2)$ samples, and this channel can be implemented by a constant-depth quantum circuit over $2n$ qubits.
\end{proof}

In the case of geometrically-local circuits with a fixed geometry (such as 1D, 2D, or 3D lattices), both the computational complexity and the implementation of the learned channel can be further optimized to maintain the geometric locality of the original system. In the experiments in the main text regarding compressing circuits or speeding up simulation as well as Appendix~\ref{app:speeding-simulation-experiment}, we choose to learn on constant blocks of size 4 on a line, as this reduces the overall block complexity to a constant of $2$, easing experimental implementation.  This is depicted visually in Fig.~\ref{fig:mbl_results}b.  This is also explored in our numerical experiments detailed in Appendix~\ref{app:num-exp-swap-circuit-learning} regarding local minima, where we show that any constant block size learning improves the learning landscape over attempting to learn the circuit all at once.

\subsection{Technical proofs}
\label{subsec:technical-proofs}

\begin{proof}[Proof of Theorem~\ref{thm:sewing-accuracy}]
Because $(U \otimes \Id) (U^\dagger \otimes \Id) = \Id \otimes \Id$, we have
\begin{equation}
    S (U^\dagger \otimes \Id) S_{n} (U \otimes \Id) \ldots (U^\dagger \otimes \Id) S_{1} (U \otimes \Id) = S (U^\dagger \otimes \Id) S (U \otimes \Id) = U \otimes U^\dagger.
\end{equation}
From triangle inequality and telescoping sum, we have
\begin{equation}
\norm{U_{\mathrm{sew}}(V_1, \ldots, V_n) - U \otimes U^\dagger}_\infty \leq \sum_{i=1}^n \norm{\left(V_i \otimes \Id \right) S_{i} \left(V_i^\dagger \otimes \Id\right) - (U^\dagger \otimes \Id) S_{i} (U \otimes \Id)}_\infty.
\end{equation}
For each term in the sum, note that the swap operator $S_i$ can be expressed as:
\begin{equation}
S_i = \frac{1}{2} \Id \otimes \Id + \frac{1}{2}\sum_{P \in \{X,Y,Z\}} P_i \otimes P_i
\end{equation}
This gives us:
\begin{align}
&\norm{\left(V_i \otimes \Id \right) S_{i} \left(V_i^\dagger \otimes \Id\right) - (U^\dagger \otimes \Id) S_{i} (U \otimes \Id)}_\infty \\
&\leq \frac{1}{2}\sum_{P \in \{X,Y,Z\}} \norm{\left(V_i P_i V_i^\dagger \otimes P_i\right) - (U^\dagger P_i U \otimes P_i)}_\infty \\
&= \frac{1}{2}\sum_{P \in \{X,Y,Z\}} \norm{(V_i P_i V_i^\dagger - U^\dagger P_i U) \otimes P_i}_\infty \\
&= \frac{1}{2}\sum_{P \in \{X,Y,Z\}} \norm{V_i P_i V_i^\dagger - U^\dagger P_i U}_\infty.
\end{align}
Now, we use the definition of $\varepsilon_i$-approximate local inversion. By the unitary invariance of the operator norm, we have:
\begin{align}
\norm{V_i P_i V_i^\dagger - U^\dagger P_i U}_\infty = \norm{P_i - V_i^\dagger U^\dagger P_i U V_i}_\infty.
\end{align}
Therefore:
\begin{equation}
\norm{\left(V_i \otimes \Id \right) S_{i} \left(V_i^\dagger \otimes \Id\right) - (U^\dagger \otimes \Id) S_{i} (U \otimes \Id)}_\infty \leq \frac{1}{2}\sum_{P \in \{X,Y,Z\}} \norm{P_i - V_i^\dagger U^\dagger P_i U V_i}_\infty \leq \frac{\varepsilon_i}{2}.
\end{equation}
Finally, combining these bounds:
\begin{equation}
\norm{U_{\mathrm{sew}}(V_1, \ldots, V_n) - U \otimes U^\dagger}_\infty \leq \sum_{i=1}^n \frac{\varepsilon_i}{2}.
\end{equation}
This completes the proof.
\end{proof}

\begin{proof}[Proof of Theorem~\ref{thm:direct-sewing-accuracy}]
Let $\hat{O}_{i,P} = \sum_{Q \in \mathrm{Pauli}(C_i)} \alpha^*_{P_i, Q} Q$ be our learned Heisenberg-evolved operators. We utilize the key identity that the ideal construction gives:
\begin{equation}
U^\dagger \otimes U = S \prod_{i=1}^n \left[ \frac{1}{2} I \otimes I + \frac{1}{2}\sum_{P \in \{X, Y, Z\}} P_i \otimes (U^\dagger P_i U) \right].
\end{equation}
This follows from the swap decomposition $S_i = \frac{1}{2} I \otimes I + \frac{1}{2}\sum_{P \in \{X,Y,Z\}} P_i \otimes P_i$ and the identity $(U^\dagger \otimes I) S_i (U \otimes I) = \frac{1}{2} I \otimes I + \frac{1}{2}\sum_{P} P_i \otimes (U^\dagger P_i U)$.
We define the ideal, approximate, and projected operators respectively:
\begin{align}
V_i &:= \frac{1}{2} I \otimes I + \frac{1}{2}\sum_{P \in \{X, Y, Z\}} P_i \otimes (U^\dagger P_i U),\\
\widetilde{W}_i &:= \frac{1}{2} I \otimes I + \frac{1}{2}\sum_{P \in \{X, Y, Z\}} P_i \otimes \hat{O}_{i,P},\\
W_i &:= \mathrm{Proj}_{\mathsf{U}}\left( \widetilde{W}_i \right).
\end{align}
We use the standard telescoping identity for products. For any sequences of matrices $A_1, \ldots, A_n$ and $B_1, \ldots, B_n$:
\begin{equation}
\prod_{i=1}^n A_i - \prod_{i=1}^n B_i = \sum_{j=1}^n \left( \prod_{i=j+1}^n A_i \right) (A_j - B_j) \left( \prod_{i=1}^{j-1} B_i \right).
\end{equation}
Applying this with $A_i = W_i$ and $B_i = V_i$:
\begin{align}
\left\|\prod_{i=1}^n W_i - \prod_{i=1}^n V_i\right\|_\infty &\leq \sum_{j=1}^n \left\|\left( \prod_{i=j+1}^n W_i \right) (W_j - V_j) \left( \prod_{i=1}^{j-1} V_i \right)\right\|_\infty \\
&\leq \sum_{j=1}^n \left\|\prod_{i=j+1}^n W_i\right\|_\infty \|W_j - V_j\|_\infty \left\|\prod_{i=1}^{j-1} V_i\right\|_\infty.
\end{align}
Since all $W_i$ and $V_i$ are unitary (the $V_i$ by construction and the $W_i$ after projection), each product has operator norm 1. Therefore:
\begin{equation}
\left\|\prod_{i=1}^n W_i - \prod_{i=1}^n V_i\right\|_\infty \leq \sum_{j=1}^n \|W_j - V_j\|_\infty.
\end{equation}
Since $S$ is unitary:
\begin{equation}
\|U_{\mathrm{sew}}^{\mathrm{Heis}}(\vec{\alpha}^*) - U^\dagger \otimes U\|_\infty \leq \sum_{i=1}^n \|W_i - V_i\|_\infty.
\end{equation}
We bound each term $\|W_i - V_i\|_\infty$ using the triangle inequality:
\begin{equation}
\|W_i - V_i\|_\infty \leq \|W_i - \widetilde{W}_i\|_\infty + \|\widetilde{W}_i - V_i\|_\infty.
\end{equation}
For the first term, since $W_i = \mathrm{Proj}_{\mathsf{U}}(\widetilde{W}_i)$ is the closest unitary to $\widetilde{W}_i$ and $V_i$ is unitary:
\begin{equation}
\|W_i - \widetilde{W}_i\|_\infty = \min_{B: \text{unitary}} \|\widetilde{W}_i - B\|_\infty \leq \|\widetilde{W}_i - V_i\|_\infty.
\end{equation}
Therefore:
\begin{equation}
\|W_i - V_i\|_\infty \leq 2\|\widetilde{W}_i - V_i\|_\infty.
\end{equation}
For the approximation error:
\begin{align}
\|\widetilde{W}_i - V_i\|_\infty &= \frac{1}{2}\left\|\sum_{P \in \{X, Y, Z\}} P_i \otimes (\hat{O}_{i,P} - U^\dagger P_i U)\right\|_\infty \\
&\leq \frac{1}{2}\sum_{P \in \{X, Y, Z\}} \|P_i \otimes (\hat{O}_{i,P} - U^\dagger P_i U)\|_\infty \\
&= \frac{1}{2}\sum_{P \in \{X, Y, Z\}} \|\hat{O}_{i,P} - U^\dagger P_i U\|_\infty \leq \frac{\varepsilon_i}{2}.
\end{align}
Combining all bounds:
\begin{equation}
\|U_{\mathrm{sew}}^{\mathrm{Heis}}(\vec{\alpha}^*) - U^\dagger \otimes U\|_\infty \leq \sum_{i=1}^n 2 \cdot \frac{\varepsilon_i}{2} = \sum_{i=1}^n \varepsilon_i,
\end{equation}
which concludes the proof.
\end{proof}

\begin{proof}[Proof of Lemma~\ref{lem:learn-H-evolved-ops}]
We analyze the algorithm described in the lemma statement. For constant-depth circuits, each Heisenberg-evolved Pauli observable $U^\dagger P_i U$ has support on at most $k = \mathcal{O}(1)$ qubits due to the limited lightcone of the constant-depth circuit. We can express each Heisenberg-evolved Pauli observable in the Pauli basis as:
\begin{equation}
U^\dagger P_i U = \sum_{Q \in \{I,X,Y,Z\}^{\otimes n}: |Q| \leq k} \alpha_Q Q
\end{equation}
where $|Q|$ denotes the weight of the Pauli operator $Q$ (number of qubits that $Q$ acts nontrivially on), and $\alpha_Q$ are real coefficients.

First, we show that the estimated Pauli coefficients $\hat{\alpha}_Q$ are unbiased estimators of the true Pauli coefficients~$\alpha_Q$. For each sample in the dataset $\mathcal{T}_U(N)$, the random value $3 \bra{\phi_{\ell, i}} P \ket{\phi_{\ell, i}}$ serves as an unbiased estimator of $\bra{\psi_\ell} U^\dagger P_i U \ket{\psi_\ell}$. This follows from the properties of the randomized measurements \cite{huang2020predicting, huang2022learning, elben2022randomized, huang2023learning, huang2024learning}. When the $i$-th qubit of $U \ket{\psi_\ell}$ is measured in basis $B$ (chosen uniformly from $\{X,Y,Z\}$), we get the following identity:
\begin{align}
\mathbb{E}[3 \bra{\phi_{\ell, i}} P \ket{\phi_{\ell, i}}] &= 3 \cdot \frac{1}{3} \cdot \mathbb{E}[\bra{\phi_{\ell, i}} P \ket{\phi_{\ell, i}} \mid B = P] \\
&= \bra{\psi_\ell} U^\dagger P_i U \ket{\psi_\ell}
\end{align}
Next, we observe that for a Pauli operator $Q$, the coefficient in the Pauli decomposition is:
\begin{equation}
\alpha_Q = \mathbb{E}_{\ket{\psi}}[\bra{\psi} U^\dagger P_i U \ket{\psi} \cdot \bra{\psi} Q \ket{\psi} \cdot 3^{|Q|}]
\end{equation}
where the factor $3^{|Q|}$ accounts for the fact that $\mathbb{E}[\bra{\psi}Q\ket{\psi}^2] = 3^{-|Q|}$ when $\ket{\psi}$ is a random stabilizer state.
Our estimator:
\begin{equation}
\hat{\alpha}_Q = \frac{3^{|Q|}}{N} \sum_{\ell=1}^N 3 \bra{\phi_{\ell, i}} P \ket{\phi_{\ell, i}} \cdot \bra{\psi_\ell} Q \ket{\psi_\ell}
\end{equation}
is therefore an unbiased estimator of $\alpha_Q$. By Hoeffding's inequality and the fact that each term in the sum above is bounded (since $|\bra{\phi_{\ell, i}} P \ket{\phi_{\ell, i}}| \leq 1$ and $|\bra{\psi_\ell} Q \ket{\psi_\ell}| \leq 1$), with $N = \mathcal{O}(2^{2k} \log(3n/\delta)/\varepsilon^2)$ samples, we have:
\begin{equation}
\Pr\left[|\hat{\alpha}_Q - \alpha_Q| \geq \frac{\varepsilon}{3 \cdot (2\sqrt{2})^k}\right] \leq \frac{\delta}{3n \cdot 4^k}
\end{equation}
Since there are at most $\mathcal{O}(n^k)$ Pauli operators with weight at most $k$, by the union bound, with probability at least $1-\delta/3n$:
\begin{equation}
|\hat{\alpha}_Q - \alpha_Q| \leq \frac{\varepsilon}{3 \cdot (2\sqrt{2})^k} \quad \forall Q, \,\, \mbox{s.t.} \,\, |Q| \leq k
\end{equation}

Now we analyze the effect of the thresholding step. For any Pauli operator $Q$ with $\alpha_Q = 0$ (i.e., not in the support of $U^\dagger P_i U$), we have with high probability:
\begin{equation}
|\hat{\alpha}_Q| \leq \frac{\varepsilon}{3 \cdot (2\sqrt{2})^k} < \frac{\varepsilon}{2 \cdot (2\sqrt{2})^k}
\end{equation}
Therefore, after thresholding, $\hat{\beta}_Q = 0$ for all $Q$ not in the support of $U^\dagger P_i U$. This ensures that $\mathrm{supp}(\hat{O}_{i,P}) \subseteq \mathrm{supp}(U^\dagger P_i U)$.
For our constructed approximation $\hat{O}_{i,P} = \sum_{Q:|Q| \leq k} \hat{\beta}_Q Q$, we can bound the spectral norm of the error:
\begin{align}
\|\hat{O}_{i,P} - U^\dagger P_i U\|_\infty &= \left\|\sum_{Q:|Q| \leq k} (\hat{\beta}_Q - \alpha_Q) Q\right\|_\infty \\
&\leq \sqrt{\sum_{Q:|Q| \leq k} |\hat{\beta}_Q - \alpha_Q|^2}
\end{align}
For each Pauli operator $Q$, we have:
\begin{equation}
|\hat{\beta}_Q - \alpha_Q| \leq |\hat{\beta}_Q - \hat{\alpha}_Q| + |\hat{\alpha}_Q - \alpha_Q| \leq \frac{\varepsilon}{2 \cdot (2\sqrt{2})^k} + \frac{\varepsilon}{3 \cdot (2\sqrt{2})^k} < \frac{\varepsilon}{(2\sqrt{2})^k}
\end{equation}
Since there are at most $4^k$ Pauli operators in the sum, we get:
\begin{align}
\|\hat{O}_{i,P} - U^\dagger P_i U\|_\infty &\leq \sqrt{\sum_{Q:|Q| \leq k} \left(\frac{\varepsilon}{(2\sqrt{2})^k}\right)^2} \\
&\leq \sqrt{4^k \cdot \left(\frac{\varepsilon}{(2\sqrt{2})^k}\right)^2} \\
&= \frac{\varepsilon \cdot \sqrt{4^k}}{(2\sqrt{2})^k} = \frac{\varepsilon \cdot 2^k}{2^k \cdot 2^{k/2}} = \frac{\varepsilon}{2^{k/2}} \leq \varepsilon
\end{align}
By the union bound over all $3n$ Pauli observables, with probability at least $1-\delta$:
\begin{equation}
\|\hat{O}_{i,P} - U^\dagger P_i U\|_\infty \leq \varepsilon \quad \text{and} \quad \mathrm{supp}(\hat{O}_{i,P}) \subseteq \mathrm{supp}(U^\dagger P_i U)
\end{equation}
for all $i \in \{1,\ldots,n\}$ and $P \in \{X,Y,Z\}$.
The runtime is $\mathcal{O}(n^k \cdot N) = \mathcal{O}(n^{k+1} \log(n/\delta)/\varepsilon^2)$, which is polynomial in the system size $n$ since $k = \mathcal{O}(1)$.
\end{proof}

\section{Optimization landscape}

\subsection{Standard landscape has exponentially many local minima}
\label{sec:exp-many-local-minima-param-shallow}

In this section, we analyze the optimization landscape of the standard parameterization for learning/training shallow quantum circuits. We demonstrate that even simple shallow quantum circuits exhibit exponentially many local minima that can trap optimization algorithms and prevent them from reaching the global optimum.

\subsubsection{Problem formulation}

Consider a one-dimensional parameterized shallow quantum circuit with three layers of parameterized SWAP operations defined as follows:
\begin{equation}
    U(\vec{\theta}) := \prod_{j} e^{i \theta_{1,j} \SWAP_{2j+1, 2j+2}} \prod_{j} e^{i \theta_{2,j} \SWAP_{2j, 2j+1}} \prod_{j} e^{i \theta_{3,j} \SWAP_{2j+1, 2j+2}},
\end{equation}
where $\vec{\theta} = (\theta_{1,j}, \theta_{2,j}, \theta_{3,j})$ is a vector of real-valued parameters.

Our goal is to train this parameterized shallow quantum circuit to generate a target unitary~$U_S$ composed of SWAP operations between specific qubit pairs:
\begin{equation}
    U_S = \prod_{j \in S} \SWAP_{4j+1, 4j+4},
\end{equation}
where $S \subseteq \{0, 1, 2, \ldots, \lfloor n / 4 \rfloor - 1 \}$ is a subset of block indices with $|S| = \Theta(n)$. For any such subset $S$, there exists a parameter configuration $\vec{\theta}$ such that $U_S = U(\vec{\theta})$. This is stated formally in the following proposition.

\begin{proposition}[Implementation of target unitary $U_S$]
Let $S \subseteq \{0, 1, 2, \ldots, \lfloor n / 4 \rfloor - 1 \}$ be any subset of block indices, and define the target unitary 
\begin{equation}
    U_S = \prod_{j \in S} \SWAP_{4j+1, 4j+4}.
\end{equation}
Consider the parameterized quantum circuit
\begin{equation}
    U(\vec{\theta}) := \prod_{j} e^{i \theta_{1,j} \SWAP_{2j+1, 2j+2}} \prod_{j} e^{i \theta_{2,j} \SWAP_{2j, 2j+1}} \prod_{j} e^{i \theta_{3,j} \SWAP_{2j+1, 2j+2}},
\end{equation}
Then there exists a parameter configuration $\vec{\theta}^*$ such that $U_S = U(\vec{\theta}^*)$.
\end{proposition}

\begin{proof}
We prove this by constructing an explicit parameter configuration $\vec{\theta}^*$ that implements $U_S$.

First, observe that to implement a SWAP operation between non-adjacent qubits $4j+1$ and $4j+4$ in block $j$, we need to compose a sequence of SWAP operations between adjacent qubits. Specifically, we use the following identity:
\begin{align}
    \SWAP_{4j+1, 4j+4} &= \SWAP_{4j+1, 4j+2} \cdot \SWAP_{4j+3, 4j+4} \cdot \SWAP_{4j+2, 4j+3}\\
    &\cdot \SWAP_{4j+1, 4j+2} \cdot \SWAP_{4j+3, 4j+4}.
\end{align}
Now, to implement this sequence using our parameterized circuit $U(\vec{\theta})$, we set the following parameters for each block $j \in S$:
\begin{align}
    \theta^*_{1, 2j+1} &= \pi/2 \quad \text{(implements $\SWAP_{4j+1, 4j+2}$ in the first layer)}\\
    \theta^*_{1, 2j+2} &= \pi/2 \quad \text{(implements $\SWAP_{4j+3, 4j+4}$ in the first layer)}\\
    \theta^*_{2, 2j+1} &= \pi/2 \quad \text{(implements $\SWAP_{4j+2, 4j+3}$ in the second layer)}\\
    \theta^*_{3, 2j+1} &= \pi/2 \quad \text{(implements $\SWAP_{4j+1, 4j+2}$ in the third layer)}\\
    \theta^*_{3, 2j+2} &= \pi/2 \quad \text{(implements $\SWAP_{4j+3, 4j+4}$ in the third layer)}
\end{align}
For all other parameters, we set $\theta^*_{k,j} = 0$, which results in identity operations.
To verify that this parameter configuration implements $U_S$, we analyze its action on each block of 4 qubits:
\begin{itemize}
    \item For blocks $j \in S$: The combination of 5 SWAP operations implements $\SWAP_{4j+1, 4j+4}$.
    \item For blocks $j \notin S$: All parameters are zero, implementing identity operations on these blocks.
\end{itemize}
Thus, the overall effect of $U(\vec{\theta}^*)$ is precisely to apply $\SWAP_{4j+1, 4j+4}$ for each $j \in S$ and the identity operation for each $j \notin S$, which is exactly the definition of $U_S$.
\end{proof}

\subsubsection{Local cost function}

To measure how well our parameterized circuit approximates the target, we need a suitable cost function. However, generic quantum circuit optimization faces a well-known challenge called the ``barren plateau problem'' \cite{mcclean2018barren}, where the optimization landscape becomes exponentially flat as the system size increases, making gradient-based optimization ineffective. While it is not a complete solution for general circuits, for shallow circuits one can significantly mitigate the barren plateau problem \cite{mcclean2018barren} by using a local cost function \cite{cerezo2021cost}:
\begin{equation} \label{eq:cost-local-theta}
    C_S(\vec{\theta}) := \mathop{\mathbb{E}}_{\substack{\ket{\psi} = \bigotimes_{i=1}^n \ket{\psi_i}\\
    \in \mathrm{stab}_1^{\otimes n}}} \sum_{i=1}^n \left( 1 -  \Tr\left(\bra{\psi_i} U(\vec{\theta})^\dagger U_S \ketbra{\psi}{\psi} U_S^\dagger U(\vec{\theta}) \ket{\psi_i}\right) \right) \geq 0.
\end{equation}
Let us define all the notation used in this equation:
\begin{itemize}
    \item $\mathop{\mathbb{E}}_{\ket{\psi}}$ denotes an expectation value over random input states $\ket{\psi}$.
    \item $\ket{\psi} = \bigotimes_{i=1}^n \ket{\psi_i}$ is a product state, where each qubit $i$ is in an independent state $\ket{\psi_i}$.
    \item $\mathrm{stab}_1^{\otimes n}$ represents the set of all $n$-qubit product states where each single-qubit state $\ket{\psi_i}$ is a stabilizer state. These are states that can be efficiently described classically, specifically the eigenstates of Pauli operators $\ket{0}, \ket{1}, \ket{+}, \ket{-}, \ket{y+}, \ket{y-}$.
\end{itemize}
This cost function has a natural interpretation:
\begin{itemize}
    \item We randomly sample product states $\ket{\psi}$ where each qubit is in one of six possible stabilizer states: $\ket{0}$, $\ket{1}$, $\ket{+}$, $\ket{-}$, $\ket{y+}$, or $\ket{y-}$.
    \item $U(\vec{\theta})^\dagger U_S \ketbra{\psi}{\psi} U_S^\dagger U(\vec{\theta})$ corresponds to evolving the input state $\ketbra{\psi}{\psi}$ under the target unitary $U_S$ followed by the inverse of the parameterized circuit $U(\vec{\theta})$.
    \item The term $\Tr\left(\bra{\psi_i} U(\vec{\theta})^\dagger U_S \ketbra{\psi}{\psi} U_S^\dagger U(\vec{\theta}) \ket{\psi_i}\right)$ represents the probability that qubit $i$ returns to its initial state $\ket{\psi_i}$ after $U(\vec{\theta})^\dagger U_S$. When $U(\vec{\theta}) = U_S$, this probability is $1$.
    \item We subtract this probability from 1 to get an error measure, sum over all qubits, and average over random input states.
\end{itemize}
This cost function has two important properties:
\begin{enumerate}
    \item \textbf{Faithfulness}: When the local cost function $C_S(\vec{\theta}) \approx 0$, our parameterized circuit $U(\vec{\theta})$ closely approximates the target $U_S$. Conversely, when $U(\vec{\theta})$ is $\varepsilon$-close to $U_S$ in the average-case distance, the cost will be at most $\mathcal{O}(n\varepsilon)$. See \cite{cerezo2021cost, caro2022generalization, caro2023out} for details.
    \item \textbf{No barren plateaus}: Because we're measuring local quantities (individual qubit properties) rather than global properties of the entire system and the parameterized circuit $U(\vec{\theta})$ is shallow, the optimization landscape doesn't flatten exponentially with system size. This makes gradient-based optimization feasible. See \cite{cerezo2021cost} for details.
\end{enumerate}
The minimum value of this cost function is $0$, which occurs exactly when $U(\vec{\theta}) = U_S$ (up to global phase factors). Higher values indicate greater differences between the parameterized circuit and the target unitary. While the local cost function avoids the barren plateau problem \cite{cerezo2021cost}, we will show that the optimization landscape still contains significant challenges.

\subsubsection{Exponentially many suboptimal local minima}
\label{app:exp-many-local-minima-training}
The following theorem reveals a fundamental limitation in quantum circuit optimization: even when using cost functions and circuits that avoid barren plateaus, the optimization landscape contains an exponential number of suboptimal local minima that can trap optimization algorithms.

\begin{theorem}[Exponentially many strictly suboptimal local minima] \label{thm:exp-many-local-minima-training}
    Consider a subset $S \subseteq \{0, 1, 2, \ldots, \lfloor n / 4 \rfloor - 1 \}$ with $|S| = \Theta(n)$.
    For the cost function $C_S(\vec{\theta})$ defined earlier, there exist exponentially many parameter configurations $\{\vec{\theta}_x\}_{x=0}^{2^{|S|} - 2}$ that are:
    \begin{align}
        &\text{1. \textit{Strictly suboptimal:}} \quad C_S(\vec{\theta}_x) \geq 1 + \min_{\vec{\theta}} C_S(\vec{\theta}),\\
        &\text{2. \textit{Local minima:}} \quad C_S(\vec{\theta}_x) \leq C_S(\vec{\theta}), \,\,\,\, \forall \vec{\theta} \,\, \text{ with } \norm{\vec{\theta} - \vec{\theta}_x}_\infty < \pi / 4,
    \end{align}
    for all $x = 0, 1, \ldots, 2^{|S|} - 2$.
\end{theorem}
Each local minimum $\vec{\theta}_x$ in the theorem statement corresponds to a circuit that correctly implements the target SWAP operations for blocks $j$ where $b_{\mathrm{id}(j)}(x) = 1$. The binary representation of $x$ directly encodes which blocks are correctly optimized. The optimization landscape contains exponentially many suboptimal local minima, each is a local minimum of an exponentially large region with a volume of $(\pi / 2)^{3 \lfloor n/4 \rfloor}$. As the optimization landscape is filled with exponentially many traps each with a large basin of attraction, an optimization algorithm that starts from any such region will get trapped by the suboptimal local minima.

As we show numerically in the main text experimental demonstrations as well as Appendix~\ref{app:num-exp-swap-circuit-learning}, as the system size $n$ increases, the probability of being trapped by a suboptimal local minimum increases rapidly to one, and the probability of being able to navigate to the global minimum decreases exponentially with system size $n$.

\subsubsection{Proof of Theorem~\ref{thm:exp-many-local-minima-training}}

For simplicity, assume $n$ is divisible by $4$ (if not, we can neglect the last $n \, \mathrm{mod} \, 4$ qubits). We group the parameters $\vec{\theta}$ as follows:
\begin{align}
    \vec{\theta}_{B, j} &:= (\theta_{1, 2j + 1}, \theta_{1, 2j + 2}, \theta_{2, 2j+1}, \theta_{3, 2j + 1}, \theta_{3, 2j + 2}), \quad \forall j = 0, \ldots, (n / 4) - 1,\\
    \theta_{L, j} &:= \theta_{2, 2j+2}, \quad  \forall j = 0, \ldots, (n / 4) - 2.
\end{align}
Here, $\vec{\theta}_{B, j}$ represents the parameters for gates acting on 4 consecutive qubits (qubits $4j+1$ through $4j+4$), while $\theta_{L, j}$ represents the parameter for a gate linking adjacent blocks.

Each integer $x \in \{0, \ldots, 2^{|S|} - 1\}$ corresponds to a specific parameter configuration $\vec{\theta}_x$. Let $b_0(x), \ldots, b_{|S|-1}(x)$ be the binary representation of $x$ using $|S|$ bits. We sort the elements of set~$S$ in ascending order and define a mapping $\mathrm{id}$ from $j \in S$ to its index in the sorted set (ranging from $0$ to $|S|-1$).
For each configuration $\vec{\theta}_x$, we set:
\begin{equation}
    \vec{\theta}_{x, B, j} := \begin{cases}
        (\pi/2, \pi/2, \pi/2, \pi/2, \pi/2) & \text{if } j \in S \text{ and } b_{\mathrm{id}(j)}(x) = 1\\
        (0, 0, 0, 0, 0) & \text{otherwise}
    \end{cases}
\end{equation}
And for all linking parameters: $\theta_{x, L, j} := 0$ for all $j = 0, \ldots, (n / 4) - 2$.
By construction, these parameter configurations yield the following cost values:
\begin{align}
    C_S(\vec{\theta}_{x}) &= 0, &\text{for } x = 2^{|S|} - 1,\\
    C_S(\vec{\theta}_{x}) &= |S| - \sum_{j=0}^{|S|-1} b_j(x) \geq 1, &\text{for } x = 0, \ldots, 2^{|S|} - 2.
\end{align}
Therefore, $\vec{\theta}_{2^{|S|} - 1}$ (where all bits equal 1) gives the global minimum with $C_S(\vec{\theta}_{2^{|S|} - 1}) = 0$, while all other configurations $\vec{\theta}_x$ have $C_S(\vec{\theta}_x) \geq 1$, establishing that they are strictly suboptimal.

Now we demonstrate that each $\vec{\theta}_x$ for $x = 0, \ldots, 2^{|S|} - 2$ is a local minimum. Consider any parameter vector $\vec{\theta}$ such that $\norm{\vec{\theta} - \vec{\theta}_x}_\infty < \pi / 4$.
For each block $j \in \{0, \ldots, (n/4)-1\}$, we label the qubits:
\begin{equation}
    a := 4j+1, \quad b := 4j+2, \quad c := 4j+3, \quad d := 4j+4.
\end{equation}
The block's contribution to the cost function is:
\begin{equation}
    C_{S, j}(\vec{\theta}) := \mathop{\mathbb{E}}_{\ket{\psi} = \bigotimes_{i=1}^n \ket{\psi_i} \in \mathrm{stab}_1^{\otimes n}} \sum_{i \in \{a, b, c, d\}} \left( 1 -  \Tr\left(\bra{\psi_i} U(\vec{\theta})^\dagger U_S \ketbra{\psi}{\psi} U_S^\dagger U(\vec{\theta}) \ket{\psi_i}\right) \right) \geq 0.
\end{equation}
If $j \notin S$ or if $j \in S$ and $b_{\mathrm{id}(j)}(x) = 1$, then $C_{S, j}(\vec{\theta}_x) = 0 \leq C_{S, j}(\vec{\theta})$.
The critical case is when $j \in S$ and $b_{\mathrm{id}(j)}(x) = 0$, where $U(\vec{\theta}_x)$ acts as identity on block $j$ while $U_S$ performs a SWAP between qubits $a$ and $d$. In this case, $C_{S, j}(\vec{\theta}_x) = 1$.
For each qubit $i$, the fidelity term can be expressed as:
\begin{align}
    &\mathop{\mathbb{E}}_{\ket{\psi} = \bigotimes_{i=1}^n \ket{\psi_i}} \left( 1 -  \Tr\left(\bra{\psi_i} U(\vec{\theta})^\dagger U_S \ketbra{\psi}{\psi} U_S^\dagger U(\vec{\theta}) \ket{\psi_i}\right)\right)\\
    &= \frac{2}{3}\left( 1 - \frac{1}{4}\Tr_{\neq i}\left(\Tr_i\left(U(\vec{\theta}) U_S\right)^\dagger \left(\frac{I_{n-1}}{2^{n-1}}\right) \Tr_i\left(U(\vec{\theta})^\dagger U_S\right)^\dagger \right)\right).
\end{align}
Through analysis of tensor contractions, for qubit $a$:
\begin{equation}
    \frac{1}{4}\Tr_{\neq a}\left(\Tr_a\left(U(\vec{\theta}) U_S\right)^\dagger \left(\frac{I_{n-1}}{2^{n-1}}\right) \Tr_a\left(U(\vec{\theta})^\dagger U_S\right)^\dagger \right) = \lambda_a + (1 - \lambda_a) \frac{1}{4},
\end{equation}
where $\lambda_a := \sin(\theta_{B, j, 2})^2 \sin(\theta_{B, j, 3})^2 \sin(\theta_{B, j, 4})^2$. Similarly for qubit $d$:
\begin{equation}
    \frac{1}{4}\Tr_{\neq d}\left(\Tr_d\left(U(\vec{\theta}) U_S\right)^\dagger \left(\frac{I_{n-1}}{2^{n-1}}\right) \Tr_d\left(U(\vec{\theta})^\dagger U_S\right)^\dagger \right) = \lambda_d + (1 - \lambda_d) \frac{1}{4},
\end{equation}
where $\lambda_d := \sin(\theta_{B, j, 1})^2 \sin(\theta_{B, j, 3})^2 \sin(\theta_{B, j, 5})^2$.
For qubits $b$ and $c$, after algebraic manipulations:
\begin{align}
    &\frac{1}{4}\Tr_{\neq b}\left(\Tr_b\left(U(\vec{\theta}) U_S\right)^\dagger \left(\frac{I_{n-1}}{2^{n-1}}\right) \Tr_b\left(U(\vec{\theta})^\dagger U_S\right)^\dagger \right)\\
    &\leq 1 - \frac{3}{4} \sin(\theta_{B, j, 3})^2 \cos(\theta_{B, j, 1})^2 \cos(\theta_{B, j, 4})^2
\end{align}
And:
\begin{align}
    &\frac{1}{4}\Tr_{\neq c}\left(\Tr_c\left(U(\vec{\theta}) U_S\right)^\dagger \left(\frac{I_{n-1}}{2^{n-1}}\right) \Tr_c\left(U(\vec{\theta})^\dagger U_S\right)^\dagger \right)\\
    &\leq 1 - \frac{3}{4} \sin(\theta_{B, j, 3})^2 \cos(\theta_{B, j, 2})^2 \cos(\theta_{B, j, 5})^2
\end{align}
Combining these bounds:
\begin{align}
    C_{S, j}(\vec{\theta}) \geq 1 &- \frac{1}{2} \sin(\theta_{B, j, 2})^2 \sin(\theta_{B, j, 3})^2 \sin(\theta_{B, j, 4})^2 - \frac{1}{2} \sin(\theta_{B, j, 1})^2 \sin(\theta_{B, j, 3})^2 \sin(\theta_{B, j, 5})^2\\
    &+ \frac{1}{2} \sin(\theta_{B, j, 3})^2 \cos(\theta_{B, j, 1})^2 \cos(\theta_{B, j, 4})^2 + \frac{1}{2} \sin(\theta_{B, j, 3})^2 \cos(\theta_{B, j, 2})^2 \cos(\theta_{B, j, 5})^2
\end{align}
The constraint $\norm{\vec{\theta} - \vec{\theta}_x}_\infty < \pi / 4$ implies:
\begin{align}
    |\sin(\theta_{B, j, k})| &< 0.5, \quad \forall k = 1, 2, 3, 4, 5,\\
    |\cos(\theta_{B, j, k})| &> 0.5, \quad \forall k = 1, 2, 3, 4, 5.
\end{align}
Therefore, $C_{S, j}(\vec{\theta}) \geq 1 = C_{S, j}(\vec{\theta}_x)$.

Since the total cost function is $C_{S}(\vec{\theta}) = \sum_{j=0}^{(n/4)-1} C_{S, j}(\vec{\theta})$, we have established that $C_S(\vec{\theta}_x) \leq C_S(\vec{\theta})$ for all $\vec{\theta}$ with $\norm{\vec{\theta} - \vec{\theta}_x}_\infty < \pi / 4$, proving that $\vec{\theta}_x$ is indeed a local minimum.

\subsection{Our landscape is benign}
\label{app:benign-landscapes}

\subsubsection{All local minima are global in learning Pauli coefficients}
\label{sec:all-local-minima-are-global}

In this section, we analyze the optimization landscape for the loss function $\mathcal{L}(\vec{\alpha})$ used in Algorithm~1. We show that this loss function is strongly convex, ensuring that all local minima are global minima. This guarantees that local optimization methods such as gradient descent will reliably converge to the optimal solution, providing a rigorous justification for the empirical efficiency of our learning procedure.

Recall from Section~\ref{sec:alg1-known-direct} that for shallow quantum circuits with known architecture, we can precompute the backward lightcones $\{C_i\}_{i=1}^n$. The Pauli coefficient parameterization introduces parameters $\alpha_{P_i, Q}$ for:
\begin{itemize}
    \item Each qubit $i \in \{1, \ldots, n\}$,
    \item Each single-qubit Pauli observable $P \in \{X, Y, Z\}$,
    \item Each multi-qubit Pauli $Q \in \mathrm{Pauli}(C_i)$ supported on $C_i$.
\end{itemize}
The total number of parameters is $3n \cdot \max_i 4^{|C_i|} = \mathcal{O}(n)$ since $|C_i| = \mathcal{O}(1)$ for constant-depth circuits.
Given the randomized measurement dataset
\begin{equation}
\mathcal{T}_U(N) = \left\{ \ket{\psi_\ell} = \bigotimes_{i=1}^n \ket{\psi_{\ell,i}},\ \ket{\phi_\ell} = \bigotimes_{i=1}^n \ket{\phi_{\ell,i}} \right\}_{\ell=1}^N,
\end{equation}
the loss function for the parameter vector $\vec{\alpha}$ is
\begin{equation}
\mathcal{L}(\vec{\alpha}) = \frac{1}{N} \sum_{\ell=1}^N \sum_{i=1}^n \sum_{P \in \{X, Y, Z\}} \sum_{Q \in \mathrm{Pauli}(C_i)} \left| \alpha_{P_i, Q} - 3^{|Q|+1} \bra{\phi_{\ell, i}} P \ket{\phi_{\ell, i}} \cdot \bra{\psi_\ell} Q \ket{\psi_\ell} \right|^2.
\end{equation}

\begin{theorem}[Strong convexity] \label{thm:strongly-convex-landscape}
The loss function $\mathcal{L}(\vec{\alpha})$ is strongly convex in the parameter space $\vec{\alpha}$.
\end{theorem}

\begin{proof}
The loss function can be rewritten in matrix form as:
\begin{equation}
\mathcal{L}(\vec{\alpha}) = \frac{1}{N} \sum_{\ell=1}^N \|\vec{\alpha} - \vec{d}_\ell\|_2^2 = \left\|\vec{\alpha} - \frac{1}{N}\sum_{\ell=1}^N \vec{d}_\ell\right\|_2^2 + \text{constant},
\end{equation}
where $\vec{d}_\ell$ is the vector of empirical estimates:
\begin{equation}
(\vec{d}_\ell)_{P_i, Q} = 3^{|Q|+1} \bra{\phi_{\ell, i}} P \ket{\phi_{\ell, i}} \cdot \bra{\psi_\ell} Q \ket{\psi_\ell}.
\end{equation}
Since $\mathcal{L}(\vec{\alpha})$ is a sum of squared $\ell_2$ norms, it is strongly convex with:
\begin{itemize}
    \item Hessian $\nabla^2 \mathcal{L}(\vec{\alpha}) = 2I$ (where $I$ is the identity matrix)
    \item Minimum eigenvalue $\lambda_{\min} = 2 > 0$
\end{itemize}
Therefore, $\mathcal{L}(\vec{\alpha})$ is $2$-strongly convex.
\end{proof}

\begin{corollary}[Global optimality of local minima] \label{cor:global-local-minima}
The loss function $\mathcal{L}(\vec{\alpha})$ has a unique global minimum $\vec{\alpha}^*$, and:
\begin{enumerate}
    \item Every local minimum is global.
    \item Every critical point satisfies $\vec{\alpha} = \vec{\alpha}^*$.
\end{enumerate}
\end{corollary}
\begin{proof}
These properties follow directly from strong convexity: strictly convex functions have a unique global minimizer, and any stationary point must be this minimum.
\end{proof}

\subsubsection{Constant local minima in learning local inversions}
\label{subsubsection:const-local-min-local-inv}

We now analyze the landscape for the non-convex optimization problem underlying Algorithms~2 and~3. Specifically, we show that although the objective is non-convex, it is defined over a compact manifold of constant dimension, and the number of local minima remains constant. As a result, gradient descent with a constant number of random restarts converges to a global minimum with high probability. This is formally stated below.

\begin{theorem}[Benign Landscape] \label{thm:benign-landscape}
Let
\begin{equation}
\hat{O}_{i,P} = \sum_{Q \in \mathrm{Pauli}(C_i)} \alpha^*_{P_i, Q} Q, \quad P \in \{X, Y, Z\},
\end{equation}
be the learned Heisenberg observables supported on a constant-size lightcone \(C_i\) with \(|C_i| = k = \mathcal{O}(1)\). Consider the optimization problem:
\begin{equation}
\min_{V_i \in \mathrm{SU}(2^k)} f(V_i), \quad \text{where} \quad f(V_i) = \sum_{P \in \{X, Y, Z\}} \|V_i^\dagger \hat{O}_{i,P} V_i - P_i\|_\infty.
\end{equation}
Then the landscape of $f$ over the compact manifold $\mathrm{SU}(2^k)$ exhibits the following properties:
\begin{enumerate}
    \item[\bf (a)] The set of critical points of $f$, which contains all local minima, consists of a finite number of connected components. This number is bounded by a constant that depends only on $k$.
    \item[\bf (b)] The basin of attraction corresponding to the set of global minima occupies a non-zero fraction of the total volume of $\mathrm{SU}(2^k)$. Consequently, a gradient-based method with a constant number of random restarts finds a global minimum with arbitrarily high probability.
\end{enumerate}
\end{theorem}
\begin{proof}
The objective function is $f(V_i) = \sum_{P \in \{X, Y, Z\}} \|V_i^\dagger \hat{O}_{i,P} V_i - P_i\|_\infty$. The domain of optimization is the special unitary group $\mathrm{SU}(2^k)$, which is a compact real-algebraic manifold. The function $f$ is a finite sum of terms, each being a composition of a smooth (real-analytic) unitary conjugation map $V_i \mapsto V_i^\dagger \hat{O}_{i,P} V_i$ and the spectral norm $\|\cdot\|_\infty$. The spectral norm is a locally Lipschitz function. As a composition and finite sum of such functions, $f$ is itself locally Lipschitz. Furthermore, both the function $f$ and the domain $\mathrm{SU}(2^k)$ are semi-algebraic. A function is semi-algebraic if its graph can be described by a finite number of polynomial equations and inequalities. The operators $\hat{O}_{i,P}$ and $P_i$ are matrices with constant entries, and the operations of matrix multiplication, addition, and conjugation are polynomial in the entries of $V_i$. The spectral norm is the maximum of the absolute values of eigenvalues, hence $f$ is semi-algebraic.

For a locally Lipschitz function on a compact domain, local minima can only occur at Clarke critical points, i.e., points $V_i$ where the Clarke subdifferential $\partial f(V_i)$ contains the zero vector. Clarke subdifferential $\partial f(V_i)$ is the convex hull of all possible limit points of gradients from nearby points. A fundamental theorem of real algebraic geometry states that the set of critical points of a semi-algebraic function on a compact semi-algebraic set has a finite number of connected components \cite{bochnak1998real}. This number is bounded by a constant dependent only on the degree of the polynomials defining the function and domain, which in this case depends only on $k = \mathcal{O}(1)$. As the local minima form a subset of the critical points, the set of local minima also consists of a finite number of connected components. This proves the first statement.

For the second statement, we consider the dynamics of gradient descent on $\mathrm{SU}(2^k)$. Since $\mathrm{SU}(2^k)$ is a compact manifold, it possesses a finite, non-zero volume (Haar measure), which we denote $\mathrm{Vol}(\mathrm{SU}(2^k))$. The basins of attraction corresponding to the finite number of local minimal components partition the manifold. Let $\mathcal{B}_g$ be the union of the basins of attraction for all global minima. As there are a finite number of basins partitioning a finite volume, the basin for any given minimum must have a non-zero measure. Thus, $\mathrm{Vol}(\mathcal{B}_g) > 0$. The probability $p$ that an initialization point drawn uniformly at random from $\mathrm{SU}(2^k)$ lies within the basin of a global minimum is therefore non-zero with probability $p = \frac{\mathrm{Vol}(\mathcal{B}_g)}{\mathrm{Vol}(\mathrm{SU}(2^k))} > 0.$ With $R$ independent random restarts, the probability of failing to initialize in $\mathcal{B}_g$ in all attempts is $(1-p)^R$. Since $0 < p \le 1$, this failure probability can be made arbitrarily small by choosing a constant $R$. Thus, gradient descent with a constant number of restarts converges to a global minimum with high probability. This establishes the second statement.
\end{proof}


\section{Generative quantum advantage for classical bitstrings}
\label{app:sec-gen-adv}

In this section, we establish generative quantum advantage for sampling classical bitstrings. These results follow relatively directly from our efficient algorithm for learning any shallow quantum circuits (see Appendix~\ref{sec:math-learning-QNC0} and \cite{huang2024learning}) and standard techniques for establishing the classical computational hardness for sampling from constant-depth quantum circuits \cite{gao2017quantum, bermejo2018architectures, haferkamp2020closing}. However, as we are explicitly using the mappings between shallow $D$-dimensional quantum circuits and deep $(D-1)$-dimensional quantum circuits considered in \cite{gao2017quantum, bermejo2018architectures, haferkamp2020closing} in the experiments, we will introduce this mapping in detail using the language of generative models.

\subsection{Problem formulation and universal approximation}

We begin by defining the task of learning to generate classical bitstrings.

\begin{definition}[The task of learning to generate classical bitstrings]
We are given a dataset of input-output bitstring pairs $\{ (x_i, y_i) \, | \, x_i \in \{0, 1\}^n, y_i \in \{0, 1\}^m \}_{i=1}^N$. Each output bitstring $y_i$ is sampled according to an unknown conditional distribution $p(y_i | x_i)$. The goal is to learn a model from the dataset that can generate new output bitstrings $y$ according to the unknown distribution $p(y | x)$ for any given new input bitstring $x$.
\end{definition}

Next, we define generative quantum neural networks (QNN). These are trainable quantum machine learning models that can be used to learn conditional probability distributions and generate new classical bitstrings from any given input bitstrings. The parameters are given by a vector $\vec{\theta}$.

\begin{definition}[Generative QNNs]\label{def:gen-qnn}
A generative QNN over $\ell$ qubits performs the following operations:
\begin{enumerate}
    \item Encodes the input bitstring $x \in \{0, 1\}^n$ into an $\ell$-qubit input product state $\ket{\psi_x}$ and a local measurement basis $\mathcal{M}_{x}$ over $m \leq \ell$ qubits.
    \item Evolves the input product state $\ket{\psi_x}$ under an $\ell$-qubit parameterized quantum circuit $C(\vec{\theta})$.
    \item Measures the first $m$ qubits of the $\ell$-qubit output state $C(\vec{\theta}) \ket{\psi_x}$ in the basis $\mathcal{M}_{x}$.
\end{enumerate}
The QNN generates an $m$-bit string $y \in \{0, 1\}^m$ according to a distribution $p_{\text{QNN}}(y | x; \vec{\theta})$ that depends on the input bitstring $x \in \{0, 1\}^n$ and the trainable circuit parameters $\vec{\theta}$.
\end{definition}

Even when the measurement is generalized to a POVM conditioned on $x$, $\mathcal{M}_x$ forms a probability distribution for $p(y | x)$.  The specific mapping between the outcome of the generalized measurement and the bitstring $y$ provides additional freedom in the mapping.  We note that for any particular choice of $\vec \theta$, we could absorb $C(\vec \theta) U_x$ into a definition of $\mathcal{M}_x(\vec \theta)$, where $U_x \ket{0} = \ket{\psi_x}$, but we choose to keep them conceptually separate here.

\subsubsection{Tomographically completeness}

Some generative QNNs are tomographically complete, meaning that by going through all exponential choices of input bitstring $x$ and looking at the conditional probability distribution $p_{\mathrm{QNN}}(y | x; \vec{\theta})$, one could obtain complete knowledge about the parameterized quantum circuit $C(\vec{\theta})$. This is formally stated as follows.

\begin{definition}[Tomographically-complete generative QNNs] \label{def:tomo-comp-gen-QNNs}
A generative QNN associated with $\ket{\psi_x}, \mathcal{M}_x, C(\vec \theta)$ over $\ell$ qubits is tomographically-complete if the following are true.
\begin{itemize}
    \item The input bitstring $x \in \{0, 1\}^n$ can be represented as blocks of bits,
    \begin{equation}
        x = (x^{(1, \mathsf{st})}, \ldots, x^{(\ell, \mathsf{st})}, \ldots, x^{(1, \mathsf{mt})}, \ldots, x^{(\ell, \mathsf{mt})}),
    \end{equation}
    where $x^{(i, \mathsf{st})}$ specifies a single-qubit state $\ket{\psi_{x^{(i, \mathsf{st})}}}$ on the $i$-th qubit and $x^{(i, \mathsf{mt})}$ specifies a single-qubit two-outcome measurement $\mathcal{M}_{x^{(i, \mathsf{mt})}}$ on the $i$-th qubit.
    \item $\forall i \in \{1, \ldots, \ell\}$, the span of $\left\{ \ketbra{\psi_{x^{(i, \mathsf{st})}}}{\psi_{x^{(i, \mathsf{st})}}}, \forall x^{(i, \mathsf{st})} \right\}$ is the space for single-qubit density matrices ($2\times 2$ positive-semidefinite complex matrices with trace one).
    \item $\forall i \in \{1, \ldots, \ell\}$, the $i$-th qubit measurement $\left\{ \mathcal{M}_{x^{(i, \mathsf{mt})}}, \forall x^{(i, \mathsf{mt})} \right\}$ is tomographically complete.
    \item All $\ell$ qubits are measured to generate $y \in \{0, 1\}^m$, i.e., $\ell = m$.
\end{itemize}
These conditions imply that the generative model $p_{\mathrm{QNN}}(y | x; \vec{\theta})$ encode full knowledge of the unitary $C(\vec{\theta})$.
\end{definition}

\subsubsection{Universal approximation}

We establish that generative QNNs possess universal approximation capabilities for arbitrary conditional distributions over classical bitstrings.

\begin{proposition}[Universal approximation] \label{app:prop-universal-approx}
For any number of bits $n, m$, any conditional distribution $p(y | x)$ where $x \in \{0, 1\}^n, y \in \{0, 1\}^m$, and any approximation error $\epsilon > 0$, there exists a generative QNN that can approximate $p(y | x)$ up to $\epsilon$ error in total variation distance (TVD).
\end{proposition}
\begin{proof}
For any conditional distribution $p(y | x)$ and any approximation error $\epsilon > 0$, there exists a reversible classical circuit $C$ taking as input the $n$-bit string $x$, $r$ random ancilla bits, and $l$ deterministic bits initialized at $0$, such that the distribution over the first $m$ output bits is $\epsilon$-close to the distribution $p(y | x)$ in total variation distance. This follows from standard circuit approximation results for probability distributions \cite{Knuth1976TheCO, bennett1978reversible, toffoli1980reversible}.

By setting $\ell = n + r + l$, we can define a quantum circuit that implements exactly the same operations as the reversible classical circuit $C$. The quantum circuit acts on an initial $\ell$-qubit state $\ket{x} \otimes \ket{+^r} \otimes \ket{0^l}$ and measures in the computational basis $\{\ket{b} : b \in \{0, 1\}^m\}$ after the computation.
Such a quantum circuit forms a generative QNN that approximates $p(y | x)$ up to $\epsilon$ error in total variation distance.
\end{proof}

While this proposition establishes the theoretical expressivity of QNNs for approximating any conditional probability distributions, the circuit size required for a given approximation error $\epsilon$ can be substantial. In analogy to the preparation of arbitrary quantum states, in the worst case, the number of gates scales exponentially in $n$ and $m$. However, for many practical distributions, the required circuit complexity is significantly lower.  We will focus here on cases here that permit efficient quantum training and inference.

\subsection{Instantaneously-deep quantum neural networks}
\label{app:inst-deep-QNN}

We now introduce a generative quantum model that we refer to as instantaneously-deep QNNs. These shallow quantum neural networks effectively sample and simulate the behavior of much deeper quantum circuits for certain generative tasks using a remarkable phenomenon unique to quantum computation. This phenomenon is well known in quantum computing, e.g., it was exploited in measurement-based quantum computation \cite{briegel2009measurement} to achieve universal quantum computation. It is also the central building block for proving that shallow quantum circuits are classically hard to simulate \cite{gao2017quantum, bermejo2018architectures, haferkamp2020closing}. The formal statement of this phenomenon is given below. A schematic example for the special case of 2D mapped to (1+1)D is given in Fig.~\ref{fig:2Dshallow-1Ddeep} in Appendix~\ref{app:num-learning-method}.

\begin{definition}[Instantaneously-deep QNNs] \label{def:inst-deep-QNN}
Given any number of qubits $n$ and a connectivity graph $G$ over $n$ qubits, an instantaneously-deep QNN is a generative QNN with a circuit geometry defined by $G$. The QNN\footnote{For simplicity, we assume a single parameter $n$. However, one could easily extend the definition to having different parameters.} has $n$ real parameters $\vec{\theta} \in \mathbb{R}^n$, takes an input bitstring $x \in \{ 0, 1 \}^{n}$, and generates an output bitstring $y \in \{0, 1\}^n$. The $i$-th bit of the bitstring $x$ determines the input state for the $i$-th qubit chosen from $\{ \ket{0}, \ket{+} = \frac{1}{\sqrt{2}}(\ket{0} + \ket{1}) \}$. We use the mapping convention $0 \mapsto \ket{+}$ and $1 \mapsto \ket{0}$. The quantum circuit associated with the QNN consists of:
\begin{enumerate}
    \item Preparing the $n$-qubit product input state $\ket{\psi_x}$ according to bitstring $x$.
    \item Applying a layer of $R_Z(\theta_i) = e^{-i\theta_i Z/2}$ rotation gates to all $n$ qubits, where $\{\theta_i\}_{i=1}^n$ are $n$ trainable parameters.
    \item Applying a layer of $CZ$ gates on all edges in the connectivity graph $G$.
    \item Measuring the $n$-qubit final state in the Pauli-$X$ basis.
\end{enumerate}
This quantum circuit is shallow with $\mathcal{O}(1)$ circuit depth regardless of the number of qubits.
\end{definition}

The following lemma shows that despite their constant depth, these shallow QNNs can effectively implement deep random quantum circuits when arranged with an appropriate connectivity structure.
This lemma uses exactly the same mechanism as in measurement-based quantum computation \cite{briegel2009measurement}.

\begin{lemma}[Generating deep QNNs with shallow QNNs]
\label{lemma:deep-from-shallow}
For any dimension $D \geq 2$, consider a geometry $G$ consisting of qubits arranged in a $D$-dimensional lattice of size $n = n_1 \times n_2 \times \ldots \times n_D$. We label each qubit by a $D$-dimensional coordinate $(i_1, \ldots, i_D)$ with $1 \leq i_j \leq n_j, \forall 1 \leq j \leq D$. The geometry $G$ is a standard Euclidean lattice where some edges with the same $i_1$ coordinate removed. This can be formally stated as follows:
\begin{itemize}
    \item Any edge in $G$ connects $(i_1, i_2, \ldots, i_D)$ and $(i_1', i_2', \ldots, i'_D)$ with $\sum_{j=1}^D |i_j - i'_j| = 1$.
    \item For all $(i_1, i_2, \ldots, i_D)$ and $(i_1', i_2', \ldots, i'_D)$ with $\sum_{j=1}^D |i_j - i'_j| = 1$, there is always an edge in $G$ if $i_1 \neq i_1'$, but may not have an edge in $G$ if $i_1 = i_1'$.
\end{itemize}
The output distribution of the instantaneously-deep QNN associated with~$G$ is equivalent to sampling and running a deep random quantum circuit over a $(D-1)$-dimensional lattice with depth $n_1$. Given an input bitstring $x \in \{0, 1\}^n$, the distribution over $y \in \{0, 1\}^n$ can be generated as follows:
\begin{enumerate}
    \item Initialize an $(n_2 \cdot n_3 \cdot \ldots \cdot n_D)$-qubit state $\ket{\phi}$ to be the product state given by $x$ restricted to the coordinates $(1, i_2, \ldots, i_D)$, where the mapping is $0 \mapsto \ket{+}$ and $1 \mapsto \ket{0}$ for each qubit.
    \item Set $i_1 \leftarrow 1$. This parameter indexes both the first coordinate of the shallow quantum circuit and the circuit layer of the corresponding deep quantum circuit.
    \item Apply a layer of $R_Z$ rotations to qubits at positions $(i_1, i_2, \ldots, i_D)$ with rotation angles $\theta_{(i_1, i_2, \ldots, i_D)}$ determined by the QNN, for all $i_2, \ldots, i_D$.
    \item Apply a layer of $CZ$ gates to $\ket{\psi}$ on edges between $(i_1, i_2, \ldots, i_D)$ and $(i_1, i_2', \ldots, i'_D)$ in $G$ for all coordinates $1 \leq i_2, i_2' \leq n_2, \ldots, 1 \leq i_D, i'_D \leq n_D$.
    \item If $i_1 < n_1$, generate the bits for $y_i$ with $i = (i_1, i_2, \ldots, i_D)$ for all $i_2, \ldots, i_D$ through the following steps:
    \begin{itemize}
        \item If $x_{(i_1+1, i_2, \ldots, i_D)} = 0$, apply an $H$ gate irrespective of $y_i$, set $y_i$ to be a random bit, and apply a $Z$ gate to the $(i_2, \ldots, i_D)$-th qubit of $\ket{\phi}$ if $y_i = 1$ (do nothing if $y_i = 0$).
        \item If $x_{(i_1+1, i_2, \ldots, i_D)} = 1$, measure the $(i_2, \ldots, i_D)$-th qubit of $\ket{\phi}$ in the $X$ basis, set $y_i$ to be the measurement outcome, and reset the $(i_2, \ldots, i_D)$-th qubit of $\ket{\phi}$ to $\ket{0}$.
    \end{itemize}
    Increment $i_1$ by one and return to Step 3.
    \item If $i_1 = n_1$, generate the random bits for $y_i$ with $i = (n_1, i_2, \ldots, i_D)$ for all $i_2, \ldots, i_D$ by measuring all $n_2 \cdot n_3 \cdot \ldots \cdot n_D$ qubits in the state $\ket{\phi}$ in the $X$ basis.
\end{enumerate}
\end{lemma}

\begin{proof}
We show that the shallow QNN simulates a depth-$n_1$ circuit on a $(D-1)$-dimensional lattice of size $n_2 \times \cdots \times n_D$. The first dimension coordinate $i_1$ in our shallow QNN effectively encodes the circuit depth, while the remaining dimensions $(i_2,\ldots,i_D)$ encode the spatial layout of qubits in the simulated deep circuit.

\vspace{0.75em}
\noindent \textit{Mapping Between Shallow QNN and Deep QNN}
\vspace{0.5em}

\noindent Let us define the correspondence precisely:
\begin{itemize}
    \item Each ``slice'' of the $D$-dimensional lattice in the shallow QNN at a fixed coordinate $i_1$ corresponds to one circuit layer of the deep QNN.
    \item Qubits at positions $(i_1, i_2, \ldots, i_D)$ in the shallow QNN correspond to qubits at positions $(i_2, \ldots, i_D)$ in layer $i_1$ of the deep QNN.
    \item Edges within a slice that connect vertices with the same $i_1$ coordinate implement the entangling operations within that circuit layer of the deep QNN.
    \item The quantum state flows from slice $i_1$ to slice $i_1+1$ through a quantum teleportation mechanism \cite{briegel2009measurement} enabled by the $CZ$ gates across different slices.
\end{itemize}
We denote $\mathcal{S}_{D-1} = \{(i_2, \ldots, i_D) : 1 \leq i_j \leq n_j, 2 \leq j \leq D\}$ to be the set of all $(D-1)$-dimensional coordinates in our lattice. Each coordinate in $\mathcal{S}_{D-1}$ corresponds to a qubit position in the deep QNN.

Let $|\phi^{(i_1)}\rangle$ denote the state of the $(D-1)$-dimensional system after processing layer~$i_1$. We prove by induction that this state correctly represents the state of the deep QNN after $i_1$ layers. To process the shallow QNN layer by layer, we use the fact that $CZ$ gates and $R_Z$ gates commute. Hence, we can order them in the following way.

\vspace{0.75em}
\noindent \textit{Base Case: Initialization and First Layer ($i_1 = 1$)}
\vspace{0.5em}

\noindent We initialize an $(n_2 \cdot \ldots \cdot n_D)$-qubit quantum register with states determined by the input bits $x_{(1, i_2, \ldots, i_D)}$, mapping $0 \mapsto |+\rangle$ and $1 \mapsto |0\rangle$ for each qubit. This exactly matches the initialization of the corresponding deep QNN. After initialization, we apply:
\begin{enumerate}
    \item $R_Z(\theta_{(1, i_2, \ldots, i_D)})$ rotations to each qubit.
    \item $CZ$ gates between qubits at positions $(1, i_2, \ldots, i_D)$ and $(1, i_2', \ldots, i_D')$ where adjacent positions differ by exactly one coordinate.
\end{enumerate}
These operations implement the first layer of entangling gates and single-qubit rotations in the deep QNN. Hence, $|\phi^{(1)}\rangle$ correctly represents the state after the first layer of the deep QNN.

\vspace{0.75em}
\noindent \textit{Inductive Step: Propagation from Layer $i_1$ to $i_1+1$}
\vspace{0.5em}

\noindent Assume that after processing layer $i_1$, the state $|\phi^{(i_1)}\rangle$ is equivalent to the state after $i_1$ layers of the deep circuit. We now show that after processing layer $i_1+1$, the resulting state $|\phi^{(i_1+1)}\rangle$ is equivalent to the state after $i_1+1$ layers of the deep circuit. For any $1 \leq i_1 < n_1$, the algorithm processes the next layer as follows:
\begin{itemize}
    \item Entangling and Rotation Operations: First, the algorithm applies:
    \begin{enumerate}
        \item $R_Z(\theta_{(i_1, i_2, \ldots, i_D)})$ rotations to each qubit.
        \item $CZ$ gates between qubits at positions $(i_1, i_2, \ldots, i_D)$ and $(i_1, i_2', \ldots, i_D')$ where spatial coordinates differ by exactly one position.
    \end{enumerate}
    These are the same as the entangling and single-qubit gates at depth $i_1$ of the deep QNN, applied to $|\phi^{(i_1)}\rangle$.
    \item Quantum State Transfer Mechanism: The next step is to transfer the state from layer $i_1$ to layer $i_1+1$. For each position $(i_2, \ldots, i_D) \in \mathcal{S}_{D-1}$, we process the qubit at $(i_1, i_2, \ldots, i_D)$ based on the value of $x_{(i_1+1, i_2, \ldots, i_D)}$:
    \begin{itemize}
        \item Case 1 ($x_{(i_1+1, i_2, \ldots, i_D)} = 0$): We have the following:
        \begin{enumerate}
            \item The qubit at position $(i_1, i_2, \ldots, i_D)$ is a component of $\ket{\phi^{(i_1)}}$,
            \begin{equation}
                \ket{\phi^{(i_1)}} = \ket{0} \otimes \ket{\widetilde{\phi}^{(i_1)}_0} + \ket{1} \otimes \ket{\widetilde{\phi}^{(i_1)}_1},
            \end{equation}
            for some unnormalized states $\ket{\widetilde{\phi}^{(i_1)}_0}, \ket{\widetilde{\phi}^{(i_1)}_1}$.
            \item The qubit at position $(i_1+1, i_2, \ldots, i_D)$ is initialized to $|+\rangle = \frac{1}{\sqrt{2}}(|0\rangle + |1\rangle)$ as specified by the condition on the input bit $x_{(i_1+1, i_2, \ldots, i_D)} = 0$.
            \item The $CZ$ gate between these qubits creates the entangled state:
            \begin{align}
                &\ket{0} \otimes \frac{\ket{0} + \ket{1}}{\sqrt{2}} \otimes \ket{\widetilde{\phi}^{(i_1)}_0} + \ket{1} \otimes \frac{\ket{0} - \ket{1}}{\sqrt{2}} \otimes \ket{\widetilde{\phi}^{(i_1)}_1}\\
                &= \frac{\ket{+} + \ket{-}}{\sqrt{2}} \otimes \frac{\ket{0} + \ket{1}}{\sqrt{2}} \otimes \ket{\widetilde{\phi}^{(i_1)}_0} + \frac{\ket{+} - \ket{-}}{\sqrt{2}} \otimes \frac{\ket{0} - \ket{1}}{\sqrt{2}} \otimes \ket{\widetilde{\phi}^{(i_1)}_1}\\
               &= \frac{\ket{+}}{2} \otimes \left( \ket{0} \otimes \left( \ket{\widetilde{\phi}^{(i_1)}_0} + \ket{\widetilde{\phi}^{(i_1)}_1} \right) + \ket{1} \otimes \left (\ket{\widetilde{\phi}^{(i_1)}_0} - \ket{\widetilde{\phi}^{(i_1)}_1} \right) \right)\\
               &+ \frac{\ket{-}}{2} \otimes \left( \ket{0} \otimes \left( \ket{\widetilde{\phi}^{(i_1)}_0} - \ket{\widetilde{\phi}^{(i_1)}_1} \right) + \ket{1} \otimes \left (\ket{\widetilde{\phi}^{(i_1)}_0} + \ket{\widetilde{\phi}^{(i_1)}_1} \right) \right)\\
               &= \frac{\ket{+}}{\sqrt{2}} \otimes \left( \ket{+} \otimes  \ket{\widetilde{\phi}^{(i_1)}_0} + \ket{-} \otimes \ket{\widetilde{\phi}^{(i_1)}_1} \right) + \frac{\ket{-}}{\sqrt{2}} \otimes \left( \ket{+} \otimes  \ket{\widetilde{\phi}^{(i_1)}_0} - \ket{-} \otimes \ket{\widetilde{\phi}^{(i_1)}_1} \right).
            \end{align}
            \item Measuring the first qubit in the $X$ basis yields a uniformly random outcome $y_{(i_1, i_2, \ldots, i_D)}$ and collapses the state to one of two cases:
            \begin{itemize}
                \item If $y_{(i_1, i_2, \ldots, i_D)} = 0$: the remaining state collapses to
                \begin{equation}
                    \ket{+} \otimes  \ket{\widetilde{\phi}^{(i_1)}_0} + \ket{-} \otimes \ket{\widetilde{\phi}^{(i_1)}_1}.
                \end{equation}
                \item If $y_{(i_1, i_2, \ldots, i_D)} = 1$: the remaining state collapses to
                \begin{equation}
                    \ket{+} \otimes  \ket{\widetilde{\phi}^{(i_1)}_0} - \ket{-} \otimes \ket{\widetilde{\phi}^{(i_1)}_1}.
                \end{equation}
            \end{itemize}
            The first qubit measured in the $X$ basis is excluded in the above equations.
            \item Applying the $H$ gate irrespective of $y_{(i_1, i_2, \ldots, i_D)}$ turns the state into
            \begin{align}
                &\ket{0} \otimes  \ket{\widetilde{\phi}^{(i_1)}_0} + \ket{1} \otimes \ket{\widetilde{\phi}^{(i_1)}_1} & \mbox{if $y_{(i_1, i_2, \ldots, i_D)} = 0$},\\
                &\ket{0} \otimes  \ket{\widetilde{\phi}^{(i_1)}_0} - \ket{1} \otimes \ket{\widetilde{\phi}^{(i_1)}_1} & \mbox{if $y_{(i_1, i_2, \ldots, i_D)} = 1$}.
            \end{align}
            \item Applying a $Z$ gate to the $(i_2, \ldots, i_D)$-th qubit of $\ket{\phi}$ if $y_{(i_1, i_2, \ldots, i_D)} = 1$ maps the state to $\ket{0} \otimes  \ket{\widetilde{\phi}^{(i_1)}_0} + \ket{1} \otimes \ket{\widetilde{\phi}^{(i_1)}_1} = \ket{\phi^{(i_1)}}$ in both cases of $y_{(i_1, i_2, \ldots, i_D)}$.
        \end{enumerate}
        In Case 1, we have teleported the qubit at position $(i_1, i_2, \ldots, i_D)$ to the qubit at position $(i_1+1, i_2, \ldots, i_D)$ while maintaining the state $\ket{\phi^{(i_1)}}$ unchanged.
        \item Case 2 ($x_{(i_1+1, i_2, \ldots, i_D)} = 1$): We have the following:
        \begin{enumerate}
        \item The qubit at position $(i_1+1, i_2, \ldots, i_D)$ is initialized to $|0\rangle$.
        \item The $CZ$ gate between the current qubit and $|0\rangle$ has no effect since $CZ|b\rangle|0\rangle = |b\rangle|0\rangle$ for any $b \in \{0,1\}$.
        \item Measuring the qubit at position $(i_1, i_2, \ldots, i_D)$ in the $X$ basis is equivalent to measuring the qubit at position $(i_2, \ldots, i_D)$ in the state $\ket{\phi^{(i_1)}}$ in the $X$ basis. This generates a measurement outcome $y_{(i_1, i_2, \ldots, i_D)} \in \{0, 1\}$ and collapses the state $\ket{\phi^{(i_1)}}$ according to the measurement outcome.
        \item We reset the qubit at position $(i_2, \ldots, i_D)$ in the state $|\phi^{(i_1)}\rangle$ to $|0\rangle$.
    \end{enumerate}
    In Case 2, the qubit at position $(i_1, i_2, \ldots, i_D)$ is not teleported, but the  qubit at position $(i_2, \ldots, i_D)$ in the state $\ket{\phi^{(i_1)}}$ is measured and reset to $|0\rangle$.
    \end{itemize}
    \item State Evolution: After processing all qubits at layer $i_1$, the state $|\phi^{(i_1+1)}\rangle$ correctly represents the state after $i_1+1$ layers of the deep QNN. This is because:
    \begin{enumerate}
        \item Qubits that received teleported states (at positions where $x_{(i_1+1, i_2, \ldots, i_D)} = 0$) carry forward the quantum information from the previous layer.
        \item Qubits initialized to $|0\rangle$ (at positions where $x_{(i_1+1, i_2, \ldots, i_D)} = 1$) match the mid-circuit measurement and qubit reset in the deep QNN.
        \item The pattern of entanglement created by the $CZ$ gates within each slice matches the entangling operations in the corresponding layer of the deep QNN.
\end{enumerate}
\end{itemize}
By induction, after all $n_1$ layers, the final state $|\phi^{(n_1)}\rangle$ is the same as the state after all $n_1$ layers of the deep QNN.

\vspace{0.75em}
\noindent \textit{Final Measurement}
\vspace{0.5em}

\noindent At the final slice ($i_1 = n_1$), all qubits are measured in the $X$ basis. The measurement outcomes $y_{(n_1, i_2, \ldots, i_D)}$ represent the final measurement results of the deep circuit. The complete output vector $y$ consists of both the intermediate measurement outcomes $y_{(i_1, i_2, \ldots, i_D)}$ for $1 \leq i_1 < n_1$, and the final measurement outcomes $y_{(n_1, i_2, \ldots, i_D)}$.

\vspace{0.75em}
\noindent \textit{Equivalence of the Output Distribution over $y$}
\vspace{0.5em}

\noindent The distribution over output vectors $y$ generated by this process is equivalent to the output distribution of the depth-$n_1$ circuit because the quantum state evolution through layers $i_1 = 1, 2, \ldots, n_1$ precisely simulates the evolution of the deep circuit through its $n_1$ layers, and the final measurement in the $X$ basis matches the measurement of the deep circuit at the last level. Therefore, despite having constant depth in terms of quantum circuit layout, the shallow QNN generates samples from the same distribution as the randomly sampled depth-$n_1$ QNN on an $(D-1)$-dimensional lattice of size $n_2 \times \cdots \times n_D$.
\end{proof}

We now show that for any polynomial-size quantum circuit $C$, there exists a family of deep quantum circuits containing $C$ such that an instantaneously-deep QNN generates output bitstrings that correspond to sampling a deep QNN in this family and generating an output bitstring from the deep QNN. To establish this result, we first compile the circuit $C$ into a 1D circuit $D$ consisting only of $H, R_Z(\theta), CZ$ gates in a layered pattern.

\begin{fact}[Universality of $H, R_Z(\theta), CZ$ \cite{nielsen2010quantum}] \label{fact:univ-HRZCZ}
For any $\epsilon > 0$ and any $m$-qubit circuit~$C$ with $\ell$ two-qubit gates, we can create an $m$-qubit 1D circuit~$D$ with $r = \mathrm{poly}(m, \ell, \log(1 / \epsilon))$ rounds of the following circuit layers:
\begin{enumerate}
    \item a layer of Hadamard gates on all $m$ qubits;
    \item a layer of $R_Z(\theta)$ gates on all qubits with independent $\theta$ on each of the $m$ qubits;
    \item a layer of $CZ$ gates on some choices of nearest neighbors on the 1D line,
\end{enumerate}
such that the unitaries implemented by $C$ and $D$ differ by at most $\epsilon$ error in diamond distance.
\end{fact}

We now define the family of deep quantum circuits, which corresponds to injecting different patterns of single-qubit $Z$ gates in each round of the circuit layers.

\begin{definition}[A family of deep quantum circuits] \label{def:family-dqc}
For any $\epsilon > 0$ and any $m$-qubit circuit~$C$ with $\ell$ two-qubit gates, consider the $m$-qubit 1D circuit $D$ with $r = \mathrm{poly}(m, \ell, \log(1 / \epsilon))$ rounds given in Fact~\ref{fact:univ-HRZCZ}. We define a family of deep quantum circuits $\{ D(y') \}_{y' \in \{0, 1\}^{(r-1) \times m}}$, where each circuit $D(y)$ is the same as $D$ but applies a $Z$ gate on the $i$-th qubit at the end of the $j$-th round if $y_{(j-1) \times m + i} = 1$ for all $i \in \{1, \ldots, m\}$ and $j \in \{1, \ldots, r-1\}$. This family contains a circuit $D(0^{(r-1) \times m})$ such that $D(0^{(r-1) \times m}) = D$ is $\epsilon$-close to $C$. From the $\log(1/\epsilon)$ scaling, we can choose $\epsilon$ to be exponentially small, so $D$ is exponentially close to $C$.
\end{definition}

The following corollary follows from Lemma~\ref{lemma:deep-from-shallow} and Fact~\ref{fact:univ-HRZCZ} and shows that the instantaneously-deep QNN samples a deep quantum circuits in the family $\{ D(y') \}$ and simulates the deep quantum circuit.

\begin{corollary}[Instantaneously-deep QNN for any deep quantum circuit] \label{cor:inst-deep-QNN-any}
For any $\epsilon > 0$ and any $m$-qubit circuit $C$ with $\ell$ two-qubit gates, consider the family of $m$-qubit 1D deep circuits $\{ D(y') \}$ with $r = \mathrm{poly}(m, \ell, \log(1 / \epsilon))$ given in Definition~\ref{def:family-dqc}. There exists an instantaneously-deep QNN over $n = m \times r$ qubits with a geometry $G$ equal to the 2D $m \times r$ square lattice with some edges removed, such that when the input is the all-zero bitstring $x = 0^n$, the output distribution over $y \in \{0, 1\}^n$ corresponds to sampling $y' \in \{0, 1\}^{(r-1) \times m}$ uniformly and setting the first $(r-1)\times m$ bits of $y$ to be equal to $y'$ and the last $m$ bits of $y$ to be equal to the output $m$-bit string from measuring the output state of the deep quantum circuit $D(y')$ in the standard basis.
\end{corollary}

\subsection{Quantum advantage for learning to generate classical bitstrings}
\label{app:quantum-advantage-bitstring}

We begin by stating a conjecture regarding sampling deep quantum circuits and generating bitstrings from them.

\begin{conjecture}[Classical hardness of sampling deep quantum circuits] \label{conj:Hard-RCS}
There exists a family of polynomial-size quantum circuits $\{C_n\}_{n \geq 1}$ with $\{D_n(y')\}_{y'}$ constructed from $C_n$ as in Definition~\ref{def:family-dqc}, such that no families of polynomial-size classical circuits can take $n \geq 1$ and many random bits as input and output a uniformly random bitstring $y'$ and an output bitstring resulting from measuring the output state $D_n(y') \ket{0^n}$ in the standard basis.
\end{conjecture}

Using standard techniques in complexity theory \cite{gao2017quantum, bermejo2018architectures, haferkamp2020closing, bergamaschi2024quantum}, Conjecture~\ref{conj:Hard-RCS} can be established by relying on more fundamental results and conjectures in complexity theory, namely that the non-uniform polynomial hierarchy does not collapse. This is formally given by the following proposition.

\begin{proposition}
Assuming the non-uniform polynomial hierarchy does not collapse, Conjecture \ref{conj:Hard-RCS} is true. \label{prop:Hard-RCS}
\end{proposition}
\begin{proof}
We prove the proposition by contradiction. Assume that Conjecture \ref{conj:Hard-RCS} is false. This means that for any family of polynomial-size quantum circuits $\{C_n\}_{n \geq 1}$, there exists a family of polynomial-size classical circuits $\{ C^{\mathrm{(cl)}}_n \}_{n \geq 1}$ that can sample from the output distribution described in Conjecture \ref{conj:Hard-RCS}. We now show that this implies the collapse of the non-uniform polynomial hierarchy.

Consider the complexity class $\mathsf{PostBQP}$, which contains all languages that can be recognized by a uniform family of polynomial-size quantum circuits $\{C_n\}_n$ with the ability to postselect on the last qubit being zero. Also, consider the complexity class $\mathsf{PostBPP/poly}$, which contains all languages that can be recognized by a probabilistic polynomial-time Turing machine that takes in a polynomial-size advice string and has the ability to postselect on bits. For any language $L$ in $\mathsf{PostBQP}$, we consider the uniform family of polynomial-size quantum circuits $\{C_n\}_n$ for recognizing $L$, and the associated (possibly non-uniform) family of polynomial-size classical circuits $\{C^{\mathrm{(cl)}}_n\}_n$ that can sample a uniform random bitstring $y'$ and a bitstring from measuring $D_n(y') \ket{0^n}$ in the standard basis. We now construct a probabilistic polynomial-time Turing machine $M$ that takes in a polynomial-size advice string and has the ability to postselect on bits to recognize $L \in \mathsf{PostBQP}$. For each input size $n$, we set the advice string to be a specification of $C^{\mathrm{(cl)}}_n$. Recall from Definition~\ref{def:family-dqc} that when $y'$ is all zero, we have $D_n(y')$ is exponentially close to $C_n$. By running $C^{\mathrm{(cl)}}_n$ as specified by the advice string on the Turing machine $M$ and postselecting on obtaining $y'$ being all zero and the last bit of the output from measuring $D_n(y') \ket{0^n}$ being zero, $M$ can simulate postselection on $C_n$ to recognize $L \in \mathsf{PostBQP}$. Hence, this implies that $\mathsf{PostBQP} \subseteq \mathsf{PostBPP/poly}$.

We now restate three well-known results from complexity theory. First, from \cite{aaronson2005quantum} showing the power of postselection in quantum computing, we have $\mathsf{PostBQP} = \mathsf{PP}$. Second, from Toda's theorem \cite{toda1991pp} showing the power of $\mathsf{PP}$, we have $\mathsf{PH} = \mathsf{P}^{\mathsf{PP}}$. Finally, following \cite{han1997threshold} with advice strings being passed through in the proof, we have $\mathsf{P}^{\mathsf{PostBPP/poly}} \subseteq \mathsf{\Sigma_3 /poly}$. Together, by combining $\mathsf{PostBQP} \subseteq \mathsf{PostBPP/poly}$, we have
\begin{equation}
\mathsf{PH} = \mathsf{P}^{\mathsf{PP}} = \mathsf{P}^{\mathsf{PostBQP}} \subseteq \mathsf{P}^{\mathsf{PostBPP/poly}} \subseteq \mathsf{\Sigma_3 /poly}.    
\end{equation}
Because the third level of non-uniform polynomial hierarchy is contained in the non-uniform polynomial hierarchy, $\mathsf{\Sigma_3 /poly} \subseteq \mathsf{PH/poly}$, we have $\mathsf{PH/poly} = \mathsf{\Sigma_3 /poly}$, which contradicts the assumption.
\end{proof}

The result above focuses on exact sampling. To extend the classical hardness to approximate sampling up to some small additive error in total variation distance, one can rely on the classical hardness of some counting problems; see \cite{gao2017quantum, bermejo2018architectures, haferkamp2020closing, bergamaschi2024quantum}. From Conjecture~\ref{conj:Hard-RCS} and the results established regarding the efficient learning of shallow QNNs and the properties of instantaneously-deep QNN, we can establish the following theorems.

\begin{theorem}[Generative advantage with tomographically-complete shallow QNNs] \label{thm:qadv-tomo-shallow-qnn}
Learning generative shallow QNNs that are tomographically-complete is quantumly easy but classically hard under Conjecture~\ref{conj:Hard-RCS}.
\end{theorem}

\begin{theorem}[Generative advantage with instantaneously-deep QNNs] \label{thm:qadv-inst-deep-qnn}
Learning generative instantaneously-deep QNNs is quantumly easy but classically hard under Conjecture~\ref{conj:Hard-RCS}.
\end{theorem}

\subsection{Proof of Theorem~\ref{thm:qadv-tomo-shallow-qnn}~and~\ref{thm:qadv-inst-deep-qnn}}

We first establish quantum easiness for both theorems, then proceed to prove classical hardness.

\subsubsection{Proof of quantum easiness for Theorem~\ref{thm:qadv-tomo-shallow-qnn}}

The quantum easiness follows from the learning algorithm presented in Appendix~\ref{sec:math-learning-QNC0}.
For simplicity, we consider the set of single-qubit states in Definition~\ref{def:tomo-comp-gen-QNNs} to be $\{\ket{0}, \ket{1}, \ket{+}, \ket{-}, \ket{y+}, \ket{y-}\}$ and the set of two-outcome measurements in Definition~\ref{def:tomo-comp-gen-QNNs} to be $X$, $Y$, or $Z$ basis measurements. Other tomographically-complete single-qubit input states and two-outcome measurements apply similarly after adapting the unbiased classical shadow estimator of Equation~\eqref{eq:unbiased-est-Pauli-alpha} in Lemma~\ref{lem:learn-H-evolved-ops}.

A classical dataset consisting of random samples $(x_i, y_i)$ for uniformly random input bitstring $x_i$ and the corresponding output bitstring $y_i$ sampled according to an unknown generative shallow QNN is equivalent to the classical dataset considered in Appendix~\ref{sec:dataset-learn-shallow-qc}. As a result, Theorem~\ref{thm:main-learn-QNC0} shows that a quantum computer can learn any unknown shallow quantum circuits up to a small diamond distance in polynomial time. More precisely, our theorem gives a classical algorithm that can efficiently learn a classical description of the shallow quantum circuit, but the classical algorithm cannot efficiently generate from the learned quantum circuit. This immediately implies a quantum algorithm can efficiently learn. After learning the shallow quantum circuits, a quantum computer can generate $y$ from the distribution $p_{\mathrm{QNN}}(y | x; \vec{\theta})$ associated with the generative QNN efficiently and accurately with small total variation distance for any input $x$.

\subsubsection{Proof of quantum easiness for Theorem~\ref{thm:qadv-inst-deep-qnn}}
\label{app:qadv-inst-deep-qnn-qease}
The quantum easiness can be shown directly using the structure of the instantaneously-deep QNN. Suppose we are given a dataset consisting of uniformly random input bitstrings $x \in \{0, 1\}^n$ and the corresponding output bitstrings $y \in \{0, 1\}^n$ sampled according to an unknown generative instantaneously-deep QNN.
From Lemma~\ref{lemma:deep-from-shallow}, for each qubit labeled by a $D$-dimensional coordinate $(i_1, \ldots, i_D)$ with $1 \leq i_j \leq n_j, \forall 1 \leq j \leq D$, we can consider the input bitstring when $x_{i_1, \ldots, i_D} = 0$ but $x_{i'_1, \ldots, i'_D} = 1$ for all neighboring qubits $(i'_1, \ldots, i'_D)$ with $1 \leq i'_j \leq n_j, \forall 1 \leq j \leq D$ with $\sum_{j=1}^D |i_j - i'_j| = 1$.
Let $N_{\mathrm{sp}}$ be the number of input bitstrings $x$ that satisfy this condition, and $y_i^{(t)}$ for $t = 1, \ldots, N_{\mathrm{sp}}$ be the measurement outcome on the $i$-th qubit for those input bitstrings.

From the $n_1$-depth quantum circuit simulation on the $(D-1)$-dimensional lattice, we can see that the measurement outcome $y_i$ for the input bitstrings satisfying the above condition is distributed according to the following experiment:
\begin{enumerate}
    \item Prepare a single-qubit $\ket{+}$ state,
    \item Apply an $R_Z(\theta_i)$ rotation with an unknown $\theta_i$,
    \item Measure the single-qubit state in the $X$ basis to obtain $y_i$.
\end{enumerate}
In this single-qubit experiment, we have:
\begin{equation}
    \Pr[y_i = 1] = \frac{1 - \cos(2 \theta_i)}{2}.
\end{equation}
Hence, we can learn $\theta_i$ up to a small error using the estimator:
\begin{equation}
    \hat{\theta}_i = \frac{1}{2} \arccos\left( 1 - 2\frac{1}{N_{\mathrm{sp}}} \sum_{t=1}^{N_{\mathrm{sp}}} y_{i}^{(t)} \right).
\end{equation}
Because the specified condition that all adjacent $x_{i'}$ are $1$ while the target $x_i$ is $0$ is satisfied with a constant probability for $D = \mathcal{O}(1)$ over random input bitstrings $x$, we have $N_{\mathrm{sp}} = \mathcal{O}(N)$ with high probability, where $N$ is the size of the training dataset (i.e., the total number of input bitstrings).
This immediately implies that any instantaneously-deep QNN can be efficiently learned to a small diamond distance.
After learning the instantaneously-deep QNN, a quantum computer can generate $y$ from the distribution $p_{\mathrm{QNN}}(y | x; \vec{\theta})$ associated with the generative QNN accurately with small total variation distance for any input bitstring $x \in \{0, 1\}^n$.

\subsubsection{Proof of classical hardness for Theorem~\ref{thm:qadv-inst-deep-qnn}.}

We first establish the classical hardness for learning generative instantaneously-deep QNNs (Theorem~\ref{thm:qadv-inst-deep-qnn}). Then, we use a result established in this proof to demonstrate the classical hardness for learning generative tomographically-complete shallow QNNs (Theorem~\ref{thm:qadv-tomo-shallow-qnn}).

We prove classical hardness for learning generative instantaneously-deep QNNs in Definition~\ref{def:inst-deep-QNN} by contradiction. Suppose that it is classically easy to learn generative models created by any family of polynomial-size instantaneously-deep QNNs. Consider the family of polynomial-size quantum circuits $\{C_n\}_{n \geq 1}$ from Conjecture~\ref{conj:Hard-RCS} with $\{D_n(y')\}_{y'}$ constructed from $C_n$ as in Definition~\ref{def:family-dqc}.  From Corollary~\ref{cor:inst-deep-QNN-any}, there exists a family of polynomial-size instantaneously-deep QNNs such that generating an output bitstring $y$ from the QNNs with the all-zero input bitstring $x$ is equivalent to sampling a uniformly random bitstring $y'$ and an output bitstring resulting from measuring the output state $D_n(y') \ket{0^n}$ in the standard basis. The classical easiness to learn generative models created by any family of polynomial-size instantaneously-deep QNNs, including the one associated with $D_n(y') \ket{0^n}$, implies that for any $n \geq 1$, there exists a $\mathrm{poly}(n)$-time classical algorithm such that given a classical dataset consisting of  $\mathrm{poly}(n)$ pairs of $(x_i, y_i)$ where $y_i$ is sampled according to the generative model $p(y | x_i)$ created by the instantaneously-deep QNN, the classical algorithm can sample from $p(y | x)$ for any $x$ in $\mathrm{poly}(n)$ time.

From this, we can combine both the polynomial-size classical dataset and the polynomial-time classical sampling algorithm into a family of polynomial-size classical circuit (or a polynomial-time classical algorithm with polynomial-size advice string). This family of polynomial-size classical circuits can take $n \geq 1$, the all-zero input bitstring $x$, and many random bits as input and output a uniformly random bitstring $y'$ and an output bitstring resulting from measuring the output state $D_n(y') \ket{0^n}$ in the standard basis. This contradict with Conjecture~\ref{conj:Hard-RCS} that no families of polynomial-size classical circuits can take $n \geq 1$ and many random bits as input and output a uniformly random bitstring $y'$ and an output bitstring resulting from measuring the output state $D_n(y') \ket{0^n}$ in the standard basis.

Together, we have shown that under Conjecture~\ref{conj:Hard-RCS}, there do not exist polynomial-size classical circuits for sampling from the generative model $p(y | x)$ created by polynomial-size instantaneously-deep QNNs with the all-zero input bitstring. Furthermore, this implies that it is classically hard to learn generative models defined by polynomial-size instantaneously-deep QNNs.

\subsubsection{Proof of classical hardness for Theorem~\ref{thm:qadv-tomo-shallow-qnn}}

The classical hardness in Theorem~\ref{thm:qadv-tomo-shallow-qnn} is proven by first proving the classical hardness of a variant of instantaneously-deep QNNs that is tomographically-complete. We begin by defining tomographically-complete instantaneously-deep QNNs:

\begin{definition}[Tomographically-complete instantaneously-deep QNNs]
Given any number of qubits $n$ and a connectivity graph $G$ over $n$ qubits, a tomographically-complete instantaneously-deep QNN is an $n$-qubit circuit with geometry defined by $G$. The QNN has $n$ real parameters $\vec{\theta} \in \mathbb{R}^n$, takes an input bitstring $x \in \{ 0, 1 \}^{5n}$, and generates an output bitstring $y \in \{0, 1\}^n$.  Consider $x$ to be $n$ blocks of $5$ bits. The $i$-th block of the bitstring $x$ determines:
\begin{itemize}
    \item The input state for the $i$-th qubit chosen from $\{ \ket{0}, \ket{1}, \ket{+}, \ket{-}, \ket{y+}, \ket{y-} \}$
    \item The measurement basis for the $i$-th qubit chosen from $\{X, Y, Z\}$
\end{itemize}
The distribution $p(y | x)$ associated with the tomographically-complete instantaneously-deep QNN can be sampled as follows:
\begin{enumerate}
    \item Prepare the $n$-qubit product input state according to bitstring $x$.
    \item Apply a single layer of $R_Z(\theta_i)$ rotation gates to all $n$ qubits, where $\{\theta_i\}_{i=1}^n$ are the $n$ trainable parameters, and $R_Z(\theta_i) = e^{-i\theta_i Z/2}$.
    \item Apply a layer of $CZ$ gates on all edges in the connectivity graph $G$.
    \item Measure the $n$-qubit final state in the Pauli basis according to bitstring $x$.
\end{enumerate}
This quantum circuit is shallow with $\mathcal{O}(1)$ circuit depth regardless of the number of qubits.
\end{definition}

Suppose that it is classically easy to learn generative models created by any polynomial-size tomographically-complete shallow QNNs.
Because a polynomial-size tomographically-complete instantaneously-deep QNN is a polynomial-size tomographically-complete shallow QNN, it is classically easy to learn generative models created by any polynomial-size tomographically-complete instantaneously-deep QNNs.
This implies that there exists a polynomial-time classical algorithm such that given a polynomial-size dataset consisting of pairs $(x_i, y_i)$ where $y_i$ is sampled according to the generative model $p(y | x_i)$ created by a polynomial-size tomographically-complete instantaneously-deep QNN, the classical algorithm can sample from $p(y | x)$ for any input bitstring $x$ in polynomial time.
From this, we can combine both the polynomial-size classical dataset and the polynomial-time classical sampling algorithm into a polynomial-size classical circuit to obtain that there exists a polynomial-size classical circuit $C$ such that given any input bitstring $x$ and many random bits, the classical circuit $C$ can sample the output bitstring $y$ according to the generative model $p(y | x)$.

Now, if we consider the input bitstring $x$ to specify the preparation of the state $\ket{+^n}$ and the measurement in the Pauli-$X$ basis for all qubits, then sampling output bitstring $y$ from the generative model $p(y | x)$ defined by the tomographically-complete instantaneously-deep QNN is the same as sampling from the corresponding instantaneously-deep QNN with the all-zero input bitstring in Definition~\ref{def:inst-deep-QNN}.
Hence, there exist polynomial-size classical circuits for sampling from the generative model created by polynomial-size instantaneously-deep QNNs with the all-zero input bitstring. As we have seen earlier, under Conjecture~\ref{conj:Hard-RCS}, there do not exist polynomial-size classical circuits for sampling from the generative model $p(y | x)$ created by polynomial-size instantaneously-deep QNNs with the all-zero input bitstring. As a result, under Conjecture~\ref{conj:Hard-RCS}, it is classically hard to learn generative models defined by polynomial-size tomographically-complete shallow QNNs.

\section{Generative quantum advantage for speeding up simulation}
\label{app:quantum-adv-compressing}

\begin{definition}[The task of speeding up physical simulation]
The input to the problem is a circuit $C$ that implements an $n$-qubit target unitary $U$ for simulating a physical system using $\mathrm{poly}(n)$ two-qubit gates, $\mathrm{poly}(n)$ ancilla qubits, and $\mathrm{poly}(n)$ depth. We are given the promise that there exists a more efficient approach for implementing the unitary $U$ using an $n$-qubit circuit of $\mathcal{O}(1)$ depth. The goal is to generate an output circuit $C'$ over $n$ qubits with only $\mathcal{O}(1)$ circuit depth that implements a quantum channel inverse-polynomially close to the target unitary $U$.
\end{definition}

\begin{theorem}[Generative quantum advantage for speeding up physical simulation]\label{thm:quantum-adv-compressing}
Under the standard assumption that $\mathsf{BPP} \neq \mathsf{BQP}$, the problem of generating a more efficient circuit $C'$ to speed up an input circuit $C$ is efficiently solvable on a quantum computer but classically intractable. Hence, there exists a generative quantum advantage for the task of speeding up physical simulation.
\end{theorem}
\begin{proof}
We establish this theorem by proving the quantum easiness and the classical hardness for the task of speeding up physical simulation.

\vspace{1em}
\noindent \textbf{Quantum easiness:} We first establish quantum easiness as follows. We begin by implementing the polynomial-size input circuit $C$ on a quantum computer. Let $U$ be the target unitary associated with circuit $C$. We collect the training set defined in Section~\ref{sec:dataset-learn-shallow-qc} by executing circuit $C$ a polynomial number of times. In each run, we input different random product states and measure each qubit of the output state in a random Pauli basis. By the problem's promise, there exists an unknown $n$-qubit circuit $C^*$ of depth $\mathcal{O}(1)$ that also implements the target unitary $U$. Using the algorithm presented in Section~\ref{sec:math-learning-QNC0}, we can leverage the collected dataset to learn an $\mathcal{O}(1)$-depth circuit $C'$ such that the quantum channel implemented by $C'$ is inverse-polynomially close to the target unitary $U$. Since all quantum operations are performed efficiently, we have successfully generated an output circuit $C'$ that solves the task of speeding up physical simulation.

\vspace{1em}
\noindent \textbf{Classical hardness:} For classical hardness, we proceed by contradiction. Assume there exists a polynomial-time classical algorithm that solves this task. We will prove that this assumption implies the existence of a polynomial-time classical algorithm that can decide whether the probability of measuring $1$ on the first qubit of any polynomial-size quantum circuit is above $2/3$ or below $1/3$, given the promise that one of these two cases holds. This would contradict the assumption that $\mathsf{BPP} \neq \mathsf{BQP}$.

Given any polynomial-size quantum circuit $D$ over $n$ qubits, consider another quantum circuit $C$ over $n(n+1)$ qubits defined as:
\begin{equation}
    \bra{0^{n^2}} \cdot \left( D^{\dagger, \otimes n} \otimes \mathsf{Id} \right) \cdot \mathsf{MAJ}{1, n+1, \ldots, n(n-1)+1 \rightarrow n^2+1} \cdot \left( D^{\otimes n} \otimes \mathsf{Id} \right) \cdot \left( \ket{0^{n^2}} \otimes \ket{\psi} \right),
\end{equation}
where $\ket{\psi}$ is an $n$-qubit input state, $\mathsf{MAJ}{1, n+1, \ldots, n(n-1)+1 \rightarrow n^2+1}$ is a quantum circuit that computes the majority bit of the $n$ qubits indexed by $1$, $n+1$, \ldots, $n(n-1)+1$ and applies an $X$ gate to the $(n^2+1)$-th qubit, and $\bra{0^{n^2}}$ represents postselection on the event that the first $n^2$ qubits are all in state $\ket{0}$.
This construction implements a linear map over $n$ qubits. Due to the uncomputation step $D^{\dagger, \otimes n}$, this linear map approximates a unitary operation with exponentially small error (as the postselection succeeds with only an exponentially small failure probability). Therefore, circuit $C$ implements an $n$-qubit unitary $U$ up to an exponentially small error. Moreover, the $n$-qubit unitary $U$ has a specific structure: it acts only on the first qubit and is equivalent to an $X$ gate if the probability of measuring $1$ on the first qubit in circuit $D$ exceeds $2/3$, and equivalent to the identity gate if this probability is below $1/3$. This specific structure follows from the $\mathsf{MAJ}_{1, n+1, \ldots, n(n-1)+1 \rightarrow n^2+1}$ operation.

If a polynomial-time classical algorithm could solve this task by generating a constant-depth quantum circuit $C'$ that approximately implements the $n$-qubit unitary $U$, we could use $C'$ to determine whether $U$ is closer to an $X$ or an $I$ gate on the first qubit. We would simply take the Pauli operator $Z_1$ acting on the first qubit and classically compute the Heisenberg-evolved observable $C'^{\dagger} Z_1 C'$. This calculation is classically efficient because $C'$ has constant depth. By determining whether $C'^{\dagger} Z_1 C'$ is closer to $-Z_1$ or $Z_1$, we can decide whether $U$ more closely resembles an $X$ or an $I$ gate on the first qubit. This decision reveals whether the probability of measuring $1$ on the first qubit of the original quantum circuit $D$ is above $2/3$ or below $1/3$, which would imply $\mathsf{BPP} = \mathsf{BQP}$, contradicting our assumption.
\end{proof}

\end{document}